\documentclass[11pt]{article}

\usepackage{PRIMEarxiv}
\usepackage{hyperref}

\usepackage[utf8]{inputenc} % allow utf-8 input
\usepackage{hyperref}       % hyperlinks
\usepackage{url}            % simple URL typesetting
\usepackage{booktabs}       % professional-quality tables
\usepackage{amsfonts}       % blackboard math symbols
\usepackage{nicefrac}       % compact symbols for 1/2, etc.
\usepackage{listings}
\usepackage{microtype}      % microtypography
\usepackage{lipsum}
\usepackage{fancyhdr}       % header
\usepackage{graphicx}       % graphics
\graphicspath{{media/}}     % organize your images and other figures under media/ folder
\usepackage{threeparttable}
\usepackage{amsmath}
\usepackage{amssymb}
\usepackage{amsthm}
\usepackage{siunitx}
\usepackage{float}
\usepackage[ruled,linesnumbered]{algorithm2e}
\usepackage{etoolbox}
\apptocmd{\thebibliography}{\raggedright}{}{}
\newtheorem{lemma}{{Lemma}}
\newtheorem{theorem}{{Theorem}}

\title{A Dual Cox Model Theory And Its Applications In Oncology}
\author{
  Powei Chen \\
  Department of Statistical Science\\
School of Mathematics \\
  Sun Yat-sen University \\
  Guangzhou, Guangdong 510275 China PRC\\
   ORCiD: 0009-0000-8687-0581\\
  \texttt{chenpw325@outlook.com} \\
     \And
 Siying Hu \\
ByteDance \\
   \And
 Dr. Haojin Zhou \\
Academy of Pharmacy \\
Xi’an Jiaotong - Liverpool University (XJTLU) \\
Suzhou, Jiangsu China PRC\\
   ORCiD: 0000-0003-0802-099X\\
  \texttt{haojin.zhou@xjtlu.edu.cn} \\
}

\begin{document}
\begin{titlepage}
\maketitle

\begin{abstract}
Given the prominence of targeted therapy and immunotherapy in cancer treatment, it becomes imperative to consider heterogeneity in patients'
responses to treatments, which contributes greatly to the widely used proportional hazard assumption invalidated as observed in several
clinical trials. To address the challenge in the data analysis in oncology clinical trials, we develop a Dual Cox model theory including a
Dual Cox model and a fitting algorithm.

As one of the finite mixture models, the proposed Dual Cox model consists of two independent Cox models based on patients' responses to one
designated treatment (usually the experimental one) in the clinical trial. Responses of patients in the designated treatment arm can be
observed and hence those patients are known responders or non-responders. From the perspective of subgroup classification, such a phenomenon
renders the proposed model as a semi-supervised problem, compared to the typical finite mixture model where the subgroup classification is
usually unsupervised.

A specialized expectation-maximization algorithm is utilized for model fitting, where the initial parameter values are estimated from the
patients in the designated treatment arm and then the iteratively reweighted least squares (IRLS) is applied. Under mild assumptions, the
consistency and asymptotic normality of its estimators of effect parameters in each Cox model are established.

In addition to strong theoretical properties, simulations demonstrate that our theory can provide a good approximation to a wide variety of
survival models, is relatively robust to the change of censoring rate and response rate, and has a high prediction accuracy and stability in
subgroup classification while it has a fast convergence rate. Finally, we apply our theory to two clinical trials with cross-overed KM plots
and identify the subgroups where the subjects benefit from the treatment or not.
\end{abstract}

% keywords can be removed
\end{titlepage}

\setcounter{page}{2}

\section{Introduction}
\subsection{Medical Research Background}
\label{sec:background}

Cancer surgical treatment is an operation or procedure by removing the tumors and possibly some adjacent tissues. As one of the oldest cancer
treatments, it still performs well in treating many types of cancer today\cite{wyld_evolution_2015}. With the invention of anesthesia in 1846
\cite{warren_inhalation_1846}and the introduction of aseptic procedures in 1867\cite{lister1867antiseptic}, surgical treatment was
revolutionized and became an acceptable and widely used medical intervention\cite{wyld_evolution_2015}\cite{devita_two_2012}. The
postoperative mortality and morbidity of all types of tumors are significantly lower today than they were 50 years
ago\cite{wyld_evolution_2015}. However, surgical removal of the tumors and some adjacent tissues may cause certain damage to the nearby normal
tissues, so as to some complications.

The era of radiotherapy began in 1895 when Roentgen first reported his discovery of X-rays\cite{devita_two_2012}. After the discovery of
X-rays by Roentgen, radiation therapy developed rapidly, although the methods were immature compared to today's
standards\cite{papac_origins_2001}. Radiotherapy is a kind of treatment that uses high-energy rays or radioactive substances to destroy tumor
cells and prevent their growth and division. This treatment can damage DNA or other important cellular molecules directly (most commonly as a
result of particulate radiation, such as $\alpha$ particles, protons, or electrons) or as a result of indirect cell damage following the
generation of free radicals (such as X-rays or $\gamma$-rays)\cite{gianfaldoni_overview_2017}. Unfortunately, normal cells, especially those
that divide frequently, may also be damaged and killed during radiotherapy\cite{gianfaldoni_overview_2017}.

Prior to the 1960s, surgery, and radiotherapy dominated cancer treatment. Because individuals were not aware of the risk of cancer metastasis,
the cure rate after several surgeries and radiotherapy could only be around 33$\%$. The advent and development of chemotherapy have opened the
era of combining chemotherapy with surgery or radiotherapy to treat cancer\cite{devita2008history}. The origins of modern chemotherapy can be
traced to the discovery of nitrogen mustard, a drug that proved to be an effective treatment for cancer\cite{papac_origins_2001}. Chemotherapy
is the use of chemically synthesized drugs, mainly by interfering with DNA, RNA, or protein synthesis, to inhibit cell proliferation and tumor
propagation. Most chemotherapeutic agents show activity in rapidly proliferating cells and thus may cause effects on rapidly proliferating
cells. However, the use of chemotherapy drugs is inevitably accompanied by certain adverse reactions. Because most chemotherapeutic agents are
metabolized and excreted through the liver or kidneys, some chemotherapeutic agents can produce toxic stimuli to the liver or kidneys. In this
case, toxic substances can accumulate in these organs, leading to disorders in organ functions\cite{amjad_cancer_2022}.

The history of targeted therapy can be dated back to 1975 when Kohler and Milstein discovered monoclonal antibody technology that made it
possible to make large quantities of the same antibodies against specific antigens\cite{kohler_continuous_1975}, and in the mid-1990s they were shown to be clinically effective \cite{devita2008history}. Targeted therapy aims to achieve precise treatment of disease
sites by delivering drugs to specific genes or proteins in specific cancer cells or the tissue environment that promotes tumor growth, with as
few side effects on normal tissues as possible\cite{padma2015overview}. At present, targeted therapy mainly includes monoclonal antibodies and
small molecule drugs, which can kill cancer cells by blocking signals that promote cancer cell growth, interfering with the regulation of cell
cycles, and inducing apoptosis of cancer cells\cite{padma2015overview}. In addition, these agents are able to activate the immune system by
targeting components in cancer cells and their microenvironment, thereby further enhancing the therapeutic effect. These drugs have the effect
of hindering tumor growth and invasion, or in adjuvant chemotherapy can make the tumors resistant to other therapeutic agents to re-sensitize.
Compared to traditional chemotherapy, molecular targeted therapy has the advantages of higher specificity and less
toxicity\cite{lee_molecular_2018}. However, the usage of molecularly targeted therapies to combat cancer carries some inevitable side
effects\cite{lee_molecular_2018}. Because human tumors of different tissue types are genetically diverse, each patient is heterogeneous, and
their responses to drugs varies with only a small proportion of patients responding to a new drug\cite{wu2006targeted}. Many patients fail to
achieve a complete response or show only a partial response and eventually develop complete resistance after a period of time. This may be
related to the complexity of oncogenic pathway interactions with multiple mutations or the inability of some drugs to detect mutations due to
low specificity\cite{lee_molecular_2018}.

Immunotherapy dates back to 1893 when American surgeon William Coley observed that a patient with sarcoma experienced a considerable shrinkage
of tissue tumors after developing lupus erythematosus, and there was no recurrence.  However, the method was not widely accepted by the
scientific community at the time because Coley was unable to elucidate the intrinsic mechanism. In the following years, despite the advent of
effective anticancer treatments such as chemotherapy and radiotherapy, the majority of patients with metastatic diseases cannot be cured, and
there is an urgent need for innovative and more effective treatments \cite{davis_overview_2000}. During the past 30 years, due to an improved
understanding of the basic principles of tumor biology and immunology, there have been major advances in cancer immunotherapy, and
immunotherapy has become the star of cancer treatment, such as one known currently: Chimeric antigen receptor T cells (CAR-T cells), in which
T cells are genetically engineered to express a specific CAR that allows them to recognize a specific cancer antigen to attack the
tumors\cite{srivastava2015engineering}, The therapeutic effect of CAR-T has been remarkably successful. For example, the cure rate in patients
with leukemia can reach 90$\%$\cite{melenhorst_decade-long_2022}. The goal of immunotherapy is to effectively activate anti-tumor responses by
targeting antigens produced by cancer cells and thereby stimulating the patient's immune system to recognize and react to cancer
mechanisms\cite{davis_overview_2000}\cite{chen_statistical_2013}. There are many approaches to elicit antitumor immune responses, involving
techniques such as therapeutic cancer vaccines, myeloablative T-cell therapies, monoclonal antibodies, and immune checkpoint
inhibitors\cite{raval_tumor_2014}. The most exciting of these were immune checkpoints, discovered in the 1990s and followed by major
breakthroughs in the 2010s. Immune checkpoints normally function to control excessive immune activation, but they are also a means for tumors
to evade the immune system. For example, Programmed cell death 1 ligand 1(PD-L1) in cancer cells can evade the immune system through
Programmed cell death protein 1(PD-1) in T cells. Therefore, immune checkpoint inhibitors can activate T cells to inhibit PD-1, allowing the
immune system to destroy the tumors\cite{mahoney_next_2015}. In addition, combining these approaches with other therapies, such as cytotoxic
chemotherapy, radiotherapy, or molecularly targeted therapies, among others, may be the key to reach the true potentials of immunotherapy in
the future management of cancer patients\cite{raval_tumor_2014}. Immune checkpoint inhibitors, ATC metastatic therapies, and cancer vaccines
are far more effective than the most effective chemotherapeutic agents that can currently be found in clinical trials of hard-to-treat tumors.
Although immune-related adverse effects are common, these innovative immune-targeted therapies are better tolerated than traditional
chemotherapeutic agents\cite{waldman_guide_2020}.

\begin{figure}[H]
    \centering
    \includegraphics[width=1\textwidth]{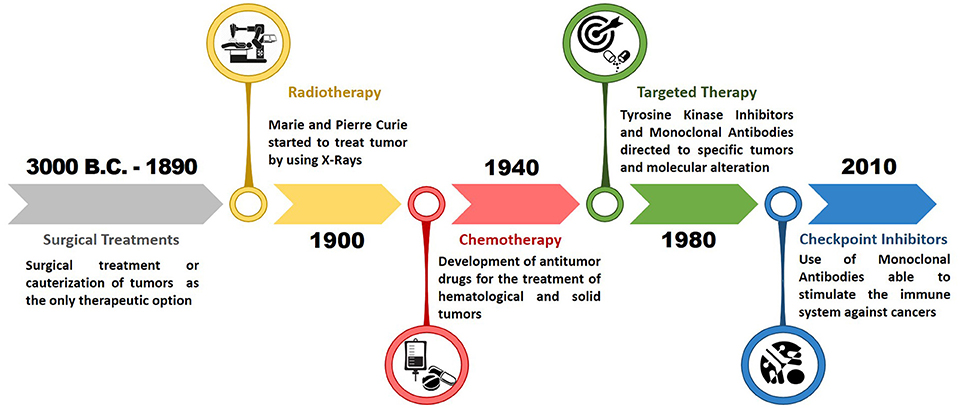}
    \caption{A timeline of significant developments in oncology\cite{falzone_evolution_2018}.}
    \label{fig:1}
\end{figure}

\subsection{Statistical Analysis Technique
}
\label{sec:related_work}

With the rapid development of targeted therapy and immunotherapy, we have entered a new era. In this era, we no longer only look at the
patient's tumor and its treatment from the perspective of fixed organ location and pathology, such as traditional surgery, radiotherapy, and
chemotherapy. Instead, we will study a patient's tumor from the perspective of more potentially dynamic genomics, proteomics, transcriptomics,
and immune abnormalities, and even focus on features that may be specifically targeted by new drugs\cite{renfro_statistical_2017}. In
addition, there are recognized inter-and intra-patient heterogeneity in any given tumor type. This heterogeneity is well established not only
by molecular changes in space (such as primary tumor to metastasis) but also in time (such as the order of
therapy)\cite{catenacci_next-generation_2015}.

The heterogeneity of new cancer treatments poses a great challenge to our traditional models of survival analysis. In the study of
time-to-event endpoints, we often assume that any factor will bear the same hazard ratio over time (the proportional hazards assumption). This
assumption is reasonable in the context of traditional cancer treatment: as treatment progresses, the risk of different factors to the patient
decreases proportionally. However, the presence of heterogeneity may cause the  proportional hazards assumption to no longer hold.
Furthermore, in immunotherapy, long-term survival and delayed clinical effects arise. Long-term survivors result in patients remaining alive
or disease-free after long-term follow-up, as shown in Figure \ref{fig:2}(B), a phenomenon usually observed in
Kaplan-Meier\cite{kaplan_nonparametric_1958} curves with tail nonzero probability. As shown in Figure \ref{fig:2}(C), the delayed clinical
effect can lead to a delayed separation or multiple crossings of the Kaplan-Meier survival curve at the beginning, that is, the immune
function has not played a role at the beginning, and the risk of the treatment group and the control group is not different at the beginning.
These results will reduce the power of some traditional statistical techniques such as log-rank test\cite{bland2004logrank} and partial
likelihood method based on Cox model \cite{cox_regression_1972}\cite{chen_statistical_2013}\cite{liao_flexible_2019}.

\begin{figure}[ht]
    \centering
    \includegraphics[width=1\textwidth]{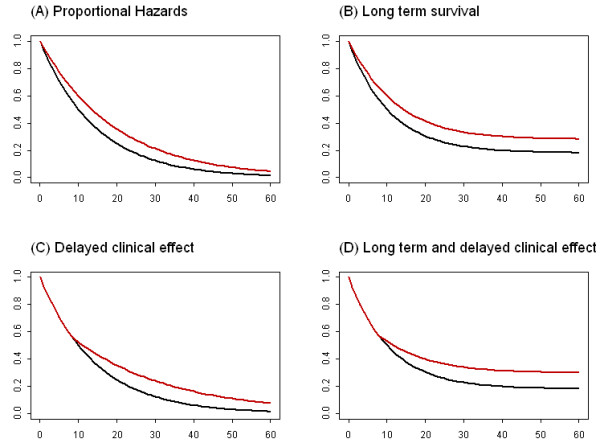}
    \caption{Kaplan-Meier survival curves for various combinations of long-term survival and delayed clinical effects are presented
    graphically. The red and black curves represent novel immuno-oncology agents and control treatments, respectively\cite{chen_statistical_2013}.}
    \label{fig:2}
\end{figure}

The existence of heterogeneity, long-term survival, and delayed clinical effects leads to the inapplicability of traditional statistical
models and poses significant challenges to medical research. In particular, patients with some specific characteristics may respond
differently to treatment. Therefore, the development of a statistical model that can distinguish different characteristics in a population
will help to better understand the efficacy of new drugs, which will have profound clinical significance. Researchers have proposed various
models. Kalbfleisch and Prentice (1981) \cite{kalbfleisch_estimation_1981}considered a piecewise  proportional hazards model, which uses time
segments to allow proportional hazards when the proportional hazards assumption does not hold over the whole time. Boag
(1949)\cite{boag1949maximum} and Yakovlev et al.(1996)\cite{tsodikov1996stochastic} proposed the mixture cure model and promotion time cure
model, respectively, which can be used to study heterogeneity between cancer patients who are long-term survivors and those who are not
long-term survivors.

The finite mixture model\cite{mclachlan_role_1994}\cite{mclachlan_finite_2000} can well describe heterogeneous data consisting of multiple
different subgroups. Therefore, in the field of survival analysis, a variety of mixture models have been proposed. The finite mixture model
can choose different distributions to describe the risk characteristics of different subgroups, such as the Log-normal distribution, the
Weibull distribution, and the gamma distribution\cite{erisoglu2012pak}\cite{mclachlan_role_1994}. Liao et al. (2019)\cite{liao_flexible_2019}
used the mixture Weibull model to divide subgroups and estimate survival and hazard functions, and the fitted curve has the same flexibility
as the Kaplan-Meier curve and can predict future events, survival probabilities, and hazard functions. It can also be used to estimate the
baseline hazard of the Cox proportional hazards model\cite{liu_analysis_2020}. Subsequently, Liao et al. (2021)\cite{jia_inferring_2021}
extended the mixture Weibull model to add the adaptive LASSO penalty term for variable selection, and the mixing probability of the latent
subgroup was modeled by the multinomial distribution depending on the baseline covariates. Finally, they successfully proved the theoretical
convergence property of the estimator. Wang (2014)\cite{wang_exponential_2014} extended the exponential tilt mixture
model\cite{qin_empirical_1999}\cite{zou_empirical_2002} to right-censored, time-event data, which can estimate the mixing probabilities and
treatment effects, and evaluate the survival probabilities of people who are responders and those who are not responders at a time point.
However, due to the non-parametric characteristics of the Cox proportional hazards model, it is difficult to establish a finite mixture Cox
model. Rosen and Tanner(1999)\cite{rosen_mixtures_1999} proposed a mixture model that combines the usual Cox proportional hazards model with a
class of features called the mixture of experts. Wu et al. (2016)\cite{wu_subgroup_2016} similarly proposed the Logistic-Cox mixed model,
which can also be regarded as a kind of mixture of experts. Eng and Hanlon (2014)\cite{eng_discrete_2014} used the EM algorithm
\cite{dempster_maximum_1977}to propose a Cox-assisted clustering algorithm to fit a finite mixture Cox model, which effectively clustered
different subgroups of the data. You et al. (2018)\cite{you2018subtype} extended it by adding a penalty term for the adaptive LASSO, providing
the asymptotic nature of the theory, stating that it has an oracle property and the estimator has a convergence rate of $\sqrt{n}$. Although
their theories may perform well for the classification of subgroups and the analysis of treatment effects, the extension of the finite mixture
Cox model in the semi-supervised scenario is still lacking. Because in specific clinical trials, targeted therapy or immunotherapy is usually
used as the experimental group, while traditional treatment is used as the control group. Due to the properties of targeted therapy or
immunotherapy, we can observe the response conditions of patients to these drugs. In addition, the mechanism and principle of drug use in the
experimental group are different from that in the control group and there is a lack of criteria to evaluate objective response rates in both
groups simultaneously. Therefore, in this context, we have the pioneering to propose a semi-supervised mixture Cox model, and set the number
of subgroups as 2, respectively, for the responders' group and the non-responders' group, and call this model the Dual Cox model.

\subsection{Paper Structure And Section Arrangement}

\label{sec:arrangement}

Emerging cancer therapeutics present difficulties for statistical modeling and mechanisms that can produce semi-supervised scenarios, as was
previously noted.  In this paper, we propose a new model, which extends the work of Eng and Hanlon (2014)\cite{eng_discrete_2014} and You et
al. (2018)\cite{you2018subtype}, and can classify subgroups. And fit the model algorithm to obtain estimates of the parameters. The
theoretical properties of the model are guaranteed, and a large number of experiments are carried out in real data analysis and simulations,
and finally, the validity of the model is verified. The subsequent sections will be in the following order:

In Section 2, the commonly used statistical models of cancer clinical trials, including the Kapla-Meier curve, Cox proportional hazards model,
and piecewise proportional hazards model, etc., are described. Subsequently, the application of various mixture models in cancer clinical
trials is illustrated. Finally, the motivation of the Dual Cox model and the fitting algorithm of the corresponding model are presented.

In Section 3, the theoretical properties of the Dual Cox model fitting algorithm are shown, and the consistency and asymptotic normality of
the estimators are proven.

In Section 4, we focus on evaluating the performance of the Dual Cox model theory, and simulation experiments are used to evaluated the
performance of the fitting algorithm. These simulations setting are designed to explore the impact of different sample sizes and different
censoring rates on estimators’ consistency and stability.  We fit the Dual Cox model to these simulated datasets, evaluate how close the model
fit results are to the true setting, and analyze their convergence properties. 

In Section 5, the PRIME and PACCE cancer clinical trials are selected as cases to clarify the application of the Dual Cox model theory in
cancer clinical data analysis. In this context, firstly, the objetives, design, analysis results, and baseline covariates of the selected
cancer clinical data are briefly introduced. Then, the Dual Cox model is fitted to the dataset, and the fitting results are comprehensively
evaluated and analyzed. Finally, diagnostic models are developed, aiming to evaluate the performance of the Dual Cox model theory.

In Section 6, the conclusion and discussion are given.  In the conclusion subsection, the motivation of this study, the implementation of the
model algorithm, the guarantee of the asymptotic nature of the theory, the results and conclusions of the numerical simulation experiments,
and the results of the real data analysis will be summarized. In the discussion subsection, the advantages and disadvantages of the Dual Cox
model are pointed out, and possible future work is discussed.

\section{A Dual Cox Model Theory}
\label{sec:dual}

\subsection{Literature Review}

\subsubsection{Kaplan-Meier Curve}

A common question when performing survival analysis is how many individuals have already experienced the event of interest, such as death or
disease recurrence, before a particular moment. However, since some individuals may still be alive during the observation period, we cannot
determine whether they will experience the event at some point in the future. Therefore, we usually use a non-parametric method, the
Kaplan-Meier estimator\cite{kaplan_nonparametric_1958}, to calculate the proportion of individuals who survive the observation period, and the
survival rate. This method was first proposed by Edward L. Kaplan and Paul Meier in a paper in 1958 and has been widely used in the fields of
medicine and biostatistics since then, becoming one of the basic tools in survival analysis research.

We assume that there are $n$individuals in the sample, where the observed $i$individual event occurred or was censored at a time of $t_i$. We
sort k nonrepeated time points of occurrence events to obtain a sequence $t_1\leq t_2\leq \ldots \leq t_k$.

Next, we define the proportion of individuals surviving after time $t$ as the survival rate, denoted as $S(t)$, $d_{i}$ denotes the number of
individuals having an event at time $t_i$, and $n_{i}$ denotes the number of individuals still alive up to time $t_i$ plus the number of
right-censored individuals at time $t_i$. Therefore, the statistical form of the Kaplan-Meier curve can be expressed as follows:

For the survival probability ($S_i$) after surviving to the $t_i$moment

$$\hat{S}(t_i) = \prod_{t_j \le t_i} \frac{n_{j} - d_{j}}{n_{j}},$$where $\hat{S}(t_i)$ is the estimated survival probability after time
$t_i$.

Finally, the Kaplan-Meier curves were generated by line plotting the estimated rates of survival. The characteristics of this curve are that
it is applicable to right-censored data, does not require assumptions about the distribution of the data, and enables simultaneous comparisons
of survival across multiple sets of data. However, Kaplan-Meier curves do not have the power to incorporate covariates and predict survival
that Cox's proportional hazards model \cite{cox_regression_1972} do. In addition, the estimation results of this curve may not be stable
enough for small sample studies.

\subsubsection{Cox Proportional Hazards Model}

The Cox proportional hazards model \cite{cox_regression_1972}(Cox model) was introduced by the British statistician David Cox in 1972. A
widely used model for survival analysis, it can be used to study the impact of multiple covariates on survival time. The model is based on the
assumption of  proportional hazards, which states that the hazard ratio between individuals is constant and does not change over time.

Let $T_i$ denotes the survival time of the i-th individual, $C_i$ denotes the censoring time of the i-th individual, and $Z_i$ denotes the
covariate of the i-th individual, the Cox model is of the form $$h_i(t) = h_0(t)\exp(\beta Z_i),$$
where $h_i(t)$ represents the risk of the $i$-th individual at time $t$, $h_0(t)$ represents the baseline hazard, and $\beta$ represents the
effect of covariates.

The Cox model's partial likelihood function does not model the baseline hazard function $h_0(t)$, so the model is a semiparametric model.

For n observations $(T_i, \delta_i, Z_i), i=1,2,... ,n$, where $\delta_i$ is a censoring indicator variable, $\delta_i=1$ means that the
survival time $T_i$is observed, and $\delta_i=0$ means that $T_i$ is censored. The partial likelihood function for the Cox model is as
follows.
$$L(\beta) = \prod_{i=1}^n \left(\frac{\exp(\beta Z_i)}{\sum_{j\in R_i}\exp(\beta Z_j)}\right)^{\delta_i},$$ where  $R_i$ denotes the set of
individuals that are still in the risk set at time $T_i$. In the partial likelihood function, the numerator represents the contribution of the
hazard function of individual $i$, and the denominator represents the sum of the hazard functions of all individuals who are still alive.

The log-partial likelihood function is $$l(\beta) = \sum_{i=1}^n \delta_i \left(\beta Z_i - \log \sum_{j\in R_i} \exp(\beta Z_j) \right).$$

The Cox proportional hazards model can be used to estimate the value of $\beta$ by maximum partial likelihood. Maximizing the logarithm of the
partial similarity function can be achieved by numerical optimization algorithms such as Gradient descent, Newton Raphson, etc. Andersen and
Gill\cite{andersen1982cox} proved asymptotic properties of the maximum partial likelihood estimator.

The Cox proportional hazards model has good interpretability and can estimate the impact of multiple risk factors, and these effects can be
used to predict the survival time of a specific individual. It is important to note that in Cox proportional hazards model, the hazard ratio
does not change over time. However, in some cases, such as delayed clinical effects of immunotherapy, the hazard ratio may change over time.
Therefore, the Kaplan-Meier curve in Figure \ref{fig:2.1} shows that the curves intersect multiple times, indicating that the assumption of
proportional hazards is invalid, and the power of the model may be reduced.

\begin{figure}[ht]
    \centering
    \includegraphics[width=1\textwidth,height=0.8\textwidth]{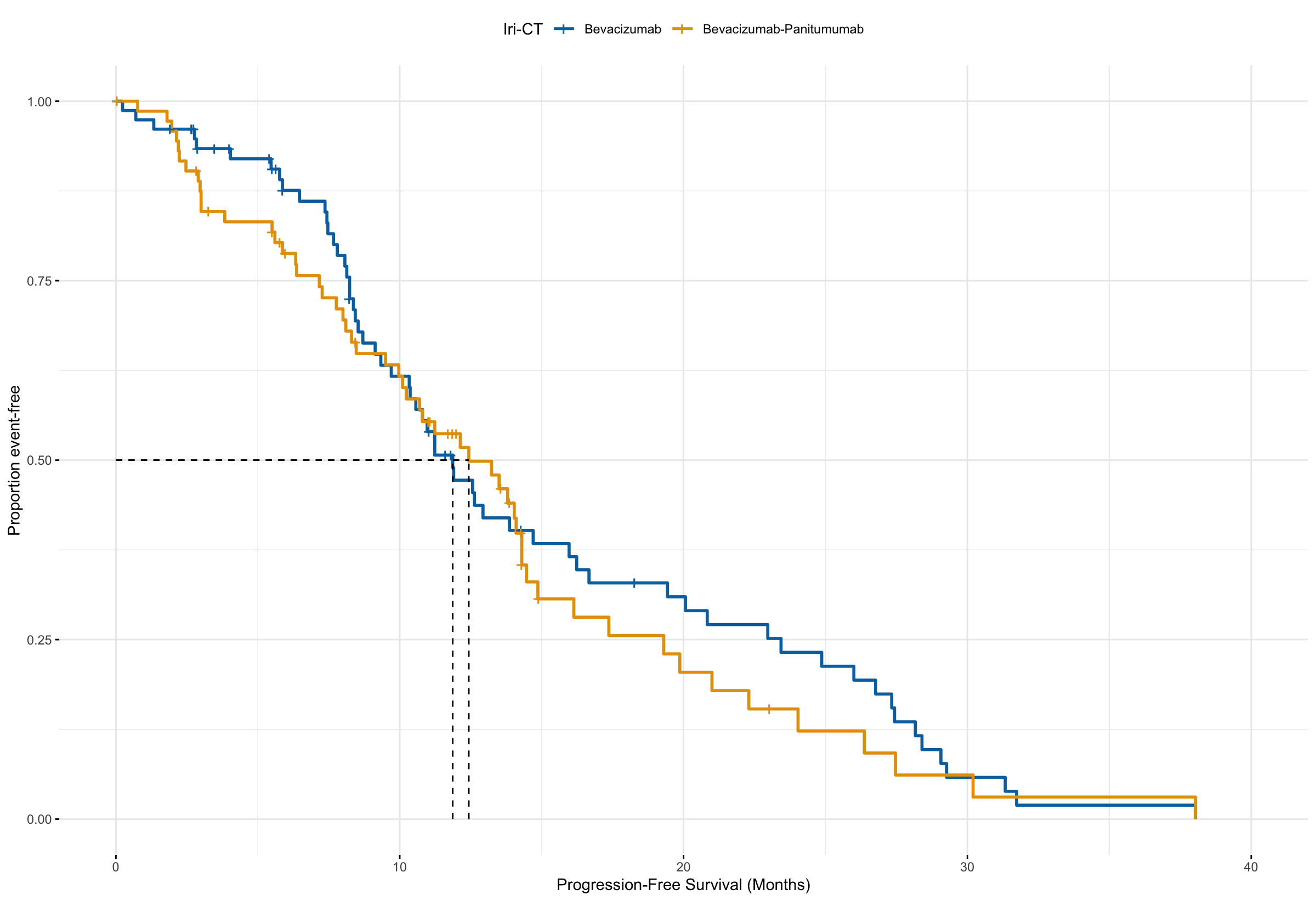}

    \caption{Multiple crossings of Kaplan-Meier curves in cancer clinical trials.
}
    \label{fig:2.1}
\end{figure}

\subsubsection{Piecewise Proportional Hazards Model}

Kalbfleisch and Prentice(1981)\cite{kalbfleisch_estimation_1981} considered the use of piecewise proportional hazards model to fit data where
Cox proportional hazards do not hold.

Zhou(2001)\cite{zhou_understanding_2001} gives the following general form of the piecewise proportional hazards model:
$$
h\left(t \mid \mathbf{z}_i\right)=h_0(t) \exp \left(\alpha z_{1 i}+\beta_1 z_{2 i} \mathbf{1}\left(t \leq c_i\right)+\beta_2 z_{2 i}
\mathbf{1}\left(t>c_i\right)\right),
$$where $z_{j i}$,$j \in\{1,2\}$ is the covariate has nothing to do with the time, $c_i$ is $i$-th observation point of division.

It can be seen that this model is a type of time-dependent Cox model\cite{therneau2000cox}, which can restrict the assumption of proportional
hazards to different time intervals. The maximum partial likelihood function for this model can be solved by an algorithm proposed by Therneau
(1999)\cite{therneau2015package}. After that, Wong et al. (2017)\cite{wong2017piecewise} proposed another simpler algorithm, however, the
likelihood function of this algorithm is not concave, so the initial value will have a great impact on the final convergence result. But they
provide methods for initial value estimation that increase the chance of convergence to the correct solution and prove that the parameter
estimates are consistent. However, due to the lack of a clear time division standard of the piecewise  proportional hazards model, it is
difficult to accurately explain the model, and the way of modeling each time period separately will make the fitting error of the model large.

\subsubsection{Cure Model}

The cure model was first proposed by Boag(1949)\cite{boag1949maximum} and Berkson $\&$ Gage (1952)\cite{berkson1952survival}.  At present,
there are mainly two types: the first one is the promotion time cure model\cite{tsodikov1996stochastic}
, which defines that the cumulative hazard will asymptote to the cure hazard and can be regarded as a special case of proportional hazards.
The second is a mixture cure model, which assumes that a subset of patients has been cured and their death rate is the same as that of the
cancer-free population.  The remaining fraction is uncured, and these subjects will eventually encounter relevant events, so their survival
function will tend to zero\cite{lambert2007modeling}.  The assumptions of the mixture cure model can be well applied to novel cancer
treatments.  Because the treatment response and survival patterns of immunotherapy and targeted therapies differ from those of chemotherapy,
which are often associated with the likelihood of long-term survival for some patients, these patients are considered "statistically cured"
and no longer susceptible to disease\cite{felizzi_mixture_2021}.  Therefore, it is reasonable to assume that some patients have been cured in
the mixture cure model, which aligns with practical applications.  In the following, we will briefly introduce the promotion time cure model
and then introduce the mixture cure model in detail.

For the promotion time cure model\cite{lambert2007modeling}, the all-cause survival function is $$S(t)=S^*(t) \exp \left\{\ln (\pi)
(1-S_z(t))\right\},$$ where $\pi$ is the cure rate of the population, $S_z(t)$ can be a distribution function commonly used for survival
functions such as Weibull or Log-normal. $S^*(t)$ is the survival function of the cured population, called the background survival function,
i.e., in cured patients, cancer no longer negatively affects the survival rate, so the survival rate is equal to that of a population in the
same age, sex, etc. The survival rate is therefore equal to the "background" level of the cancer-free population with the same age and sex
factors. Therefore, the background survival function is generally estimated from populations in countries or large regions, and we generally
assume that it is known. In addition, the promotion time cure model can be modeled with covariates in $\pi$ as well as in
$S_z(t)$\cite{lambert2007modeling}\cite{amico2018cure}.

In the mixture cure model \cite{lambert2007modeling}, the all-cause survival function is $$
S(t)=S^*(t)\left\{\pi+(1-\pi) S_u(t)\right\},
$$ where $\pi$ is the cure rate of the population, $S^*(t)$ is the background survival function, and $ S_u(t)$ is the survival time of the
uncured population, which can be expressed using some parametric distribution such as Weibull or Gompertz
distribution\cite{lambert2007modeling}\cite{amico2018cure}.

The hazard function can then be written as  $$
h(t)=h^*(t)+\frac{(1-\pi) f_u(t)}{\pi+(1-\pi) S_u(t)}.$$

Suppose there are a total of n independent samples and the observed data are shaped as $(t_i,\delta_i)$, $t_i$ is the observed event
occurrence time or censored time if $t_i$ is not censored $\delta_i=$1 otherwise it is 0, then the log-likelihood function can be written as
$$
\log L=\sum_{i=1}^n \delta_i \times \log h\left(t_i\right)+\sum_{i=1}^n \log S\left(t_i\right). $$

The mixture cure model also allows flexibility in estimating the effect of covariates on cure rate and latency, i.e., including different
covariates in the incidence ($X_i$) and latency ($Z_i$) components, as well as including the same covariate. This is important because some
factors may affect a patient's risk of having an event but not the timing of the event. For example, surgical factors may affect whether a
patient is cured or not, and their effect on the time to event may be insignificant \cite{jia_cure_2013}. The following is an example of the
method proposed by Sy and Taylor \cite{sy_estimation_2000}:

Within the curing framework, two components were jointly modeled. The first component estimates the probability of not being cured (1-cure
rate) through a logistic regression model. $$\operatorname{logit}(\operatorname{Pr}[\mathrm{Y}=1 \mid \mathrm{X}])=\mathrm{b}_0+\mathrm{X}_1
\mathrm{b}_1+\cdots+\mathrm{X}_{\ mathrm{m}} \mathrm{b}_{\mathrm{m}},$$ where it is assumed that there are m covariates to be considered and
the binary outcome indicator $\mathrm{Y}$ indicates whether the patient is uncured, uncured $\mathrm{Y}$ and cured $\mathrm{Y}$ = 0.

The second component restricts to uncured patients i.e. $(\mathrm{Y}=1)$, which models the latency period, i.e. uncured patients are at risk
$$
\lambda(\mathrm{t} \mid \mathrm{Z}, \mathrm{Y}=1)=\lambda_0(\mathrm{t} \mid \mathrm{Y}=1) \exp \left(Z_1 \beta_1+\cdots+Z_{\mathrm{n}} \
beta_{\mathrm{n}}\right),$$

Then the two components are fitted together and can be used to estimate the effect of covariates on incidence and latency.

Assuming a total of n independent samples with observations shaped as $(t_i,\delta_i,x_i,z_i)$, $t_i$ is the observed event occurrence time or
censoring time, and if $t_i$ is not censored $\delta_i=$1 otherwise it is 0, the likelihood function can be written as

$$\begin{aligned}L\left(b, \beta, \Lambda_0\right)=\prod_{i=1}^n\{ & \left.\left(1-\pi \right) \lambda_0\left(t_i \mid Y=1\right)
e^{z_i^{\prime} \beta} e^{-\Lambda_0\left(t_i \mid Y=1\right) \exp \left(z_i^{\prime} \beta\right)}\right\}^{\delta_i} \\ & \times\left\{\pi+
\left(1-\pi\right)e^{-\Lambda_0\left(t_i \mid Y=1\right) \exp \left(z_i^{\prime} \beta\right)}\right\}^{1-\delta_i},\end{aligned}$$ where
$1-\pi=\operatorname{Pr}[\mathrm{Y}=1 \mid \mathrm{X}]$,$\Lambda_0\left(t_i \mid Y=1\right)$ is the cumulative hazard function from 0 to
$t_i$.

Sy and Taylor fitted this model using the EM algorithm with consistent and asymptotic normality in the estimation of the model coefficients
and confirmed the validity of the model by simulation experiments. However, modeling incidence and latency simultaneously may be overly
dependent on parametric assumptions, leading to overparameterization of the model.

\subsubsection{Finite Mixture Model And The EM Algorithm
}

Finite Mixture Model \cite{mclachlan_role_1994}\cite{mclachlan_finite_2000}, which decompositions the population into components, each with
its own parameters and weights. The weighted sum of these components constitutes the distribution of the population. The finite mixture model
can be used to solve the situation where there are multiple groups, subgroups, patterns, etc., in the data, and it is a very powerful analysis
tool. Parameter estimation of finite mixture models can be done using maximum likelihood estimation or Bayesian methods. The EM algorithm
\cite{dempster_maximum_1977} is usually used for maximum likelihood estimation.

\noindent{\textbullet}\quad\parbox[t]{0.9\textwidth}{\textbf{Unsupervised situation
}}

For the unsupervised learning case, suppose there are $n$ observations $x_1, x_2, \dots, x_n$, where $x_i$ follows one of the $k$
distributions. The mixing probability is $\boldsymbol{\pi}=(\pi_1, \pi_2, ... , \pi_K)$, then the likelihood function of the finite mixture
model is

$$L(\boldsymbol{\theta} \mid \mathbf{y})=\prod_{i=1}^n f\left(x_i \mid \boldsymbol{\theta}\right)=\prod_{i=1}^n \left(\sum_{j=1}^k \pi_j
f_j\left(x_i \mid \boldsymbol{\theta}_j\right)\right),$$where $\boldsymbol{\theta} = (\boldsymbol{\theta_1}, \boldsymbol{\theta_2}, ...,
\boldsymbol{\theta_k}).$

The above likelihood function is not directly maximizable because the identity of the mixture components is unknown. To solve this problem,
parameter estimation can be performed using the EM algorithm, which is an iterative algorithm for solving the maximum likelihood estimates of
the parameters of probabilistic models containing latent variables. It consists of two steps: the E step and the M step. In the E step
(Expectation Step), for each data point, the probability that it belongs to each component is calculated; in the M step (Maximization Step),
the probabilities calculated in the E step are used to update the model parameters. The parameter values are estimated after several
iterations until the likelihood function converges.

Complete Likelihood Function $$L_c(\boldsymbol{\theta} \mid \mathbf{y}, \mathbf{z})=\prod_{i=1}^n \prod_{j=1}^k\left[\pi_j f_j\left(x_i \mid
\boldsymbol{\theta}_j\right)\right]^{z_{i j}},
$$where $z_{ij}$ is an indicator variable that represents the probability that the $i$-th observation belongs to the $j$th component, i.e.,
$z_{ij} = 1$ means that the $i$-th observation belongs to the $j$th component and $z_{ij}$ = 0 means that it does not.

\begin{itemize}
\item E step: In the mth iteration, the probability that each observation belongs to each component, i.e., the posterior probability of
    $z_{ij}$, is calculated $$z_{i j}^{(m)}=\frac{\pi_j^{(m-1)} f_j\left(x_i \mid \boldsymbol{\theta}_j^{(m-1)}\right)}{\sum_{l=1}^k
    \pi_l^{(m-1)} f_l\left(x_i \mid \boldsymbol{\theta}_l^{(m-1)}\right)}.$$
\item $M$ step: In the mth iteration, the parameter values are re-estimated using the posterior probabilities obtained in the $E$ step
$$
\begin{gathered}
{\pi}_j^{(m)}=\frac{\sum_{i=1}^n z_{i j}^{(m)}}{n}, \\
{\boldsymbol{\theta}}^{(m)}=\arg \max _{\boldsymbol{\theta}} Q(\boldsymbol{\theta};\boldsymbol{\theta}^{(m-1)})=\arg \max
_{\boldsymbol{\theta}}\sum_{i=1}^n z_{i j}^{(m)} \log f\left(x_i \mid \boldsymbol{\theta}\right).
\end{gathered}
$$
\end{itemize}

Dempster et al. (1977)\cite{dempster_maximum_1977} proved that the functions $Q(\boldsymbol{\theta};\boldsymbol{\theta}^{(m-1)})$ and
$L(\boldsymbol{\theta} \mid \mathbf{y})$ do not decrease after EM iterations, i.e
$$Q(\boldsymbol{\theta}^{(m)};\boldsymbol{\theta}^{(m-1)})\geq Q(\boldsymbol{\theta};\boldsymbol{\theta}^{(m-1)}),\quad
L(\boldsymbol{\theta}^{(m)} \mid \mathbf{y})\geq L(\boldsymbol{\theta}^{(m-1)} \mid \mathbf{y}).$$

\noindent{\textbullet}\quad\parbox[t]{0.9\textwidth}{\textbf{Semi-supervised situation
}}

For the case of semi-supervised learning, the finite mixture model can also be solved by the EM algorithm. In this case, the likelihood
function of the model consists of two parts: the part with known labels and the part with unknown labels.
Suppose we have $n$ samples, of which $m$ samples have labels and $l=n-m$ samples have no labels. Let the observation of the $i$-th sample be
$x_i$ and the corresponding label be $y_i$, then the likelihood function of the finite mixture model can be expressed as  $$L\left(\theta
\mid\left\{x_i, y_i\right\}\right)=\prod_{i=1}^m \prod_{k=1}^K \pi_k p \left(x_i \mid \theta_k\right)^{\left[y_i=k\right]}+\prod_{i=m+1}^n
\sum_{k=1}^K \pi_k p\left(x_i \mid \theta_k\right),$$ where $\theta$ is the model parameters, $K$ is the number of mixture components, $\pi_k$
is the prior probability of mixture probability $k$, $p(x_i|\theta_k)$ is the conditional probability density function of mixture component
$k$, and $[y_i=k]$ is the indicator function that takes the value of $1$ when $y_i=k$ and $0$ otherwise.

Complete Likelihood Function $$L_c\left(\theta \mid\left\{x_i, y_i\right\}\right)=\prod_{i=1}^m \prod_{k=1}^K \pi_k p\left(x_i \mid
\theta_k\right)^{\left[y_i=k\right]}+\prod_{i=m+1}^n \prod_{k=1}^K \left[\pi_k p\left(x_i \mid \theta_k\right)\right]^{z_{i j}}.$$

Similarly, in the case of semi-supervised learning, the EM algorithm can be used for solving the problem. In both unsupervised and
semi-supervised cases, the EM algorithm is sensitive to the initial values and tends to converge to a locally optimal solution instead of a
globally optimal one, and the convergence of the likelihood function does not guarantee the convergence of the parameters
\cite{wu_convergence_1983}. In the unsupervised case, we also encounter the arbitrary setting of the number of mixture components, which is
less interpretable.

In survival analysis, the distributions selected in the finite mixture model can be selected from various common parametric distributions in
survival analysis, such as Log-normal distribution, Weibull distribution, Gamma distribution, etc.
\cite{erisoglu2012pak}\cite{mclachlan_role_1994}, which can describe different subgroups of risk characteristics.

\subsubsection{Exponential Tilt Mixture Model}

The Exponential Tilt Mixture Model (ETMM) was first proposed by Qin\cite{qin_empirical_1999} and Zou\cite{zou_empirical_2002}, after which
Wang\cite{wang_exponential_2014} extended their results to censored, time-to-event data, and gave their theoretical asymptotic properties,
using the EM algorithm to iteratively obtain the density function of survival time for the control group, and then using the Newton-Raphson
method to estimate the mixing probability and treatment effect, which can be assessed at a time point for the responders and non-responders.
The probability of survival for responders and non-responders at a given time point can be assessed. For the placebo group, $F_0($. $)$ is the
distribution of survival time in the placebo group. For the treatment group, the distribution of survival time for the non-responder
population after treatment is the same as the distribution of survival time for the placebo group. To describe the different treatment effects
in the responding and non-responder populations, we assume that the survival times in the treatment group follow a mixture distribution:
$$\lambda F_0(t)+(1-\lambda) F_1(t),$$ where $1-\lambda$ is the proportion of the responder population and $F_1($. $)$ is the distribution of
survival time after treatment for the responder population.

We assume that $d F_1(t) =d F_0(t) \exp \{h(t, \beta)\}$, and $h(t, \beta)$ is a pre specified parametric equation. Because $F_0(t)$ is
completely unspecified, this model is semi-parametric. ETMM can be viewed as a semi-parametric generalization of the parametric mixture model.
Survival time and censoring time are assumed to be independent, and censoring time is further assumed to be discrete. For the placebo group
the censoring times are $c_{01},\ldots, c_{0 d_0},c_{0 d_0+1}$ and for the treatment group the censoring times are $\ c_{11}\ldots, c_{1 d_1},
c_{1 d_1+1}$. Let $c_z=\left(c_{z 1}, \ldots c_{z d_z}\right)^T, z=0,1$. Assume that the observations consist of $n\left(=n_0+n_1\right)$
uncensored independent observations $x_1, \ldots, x_n$, where the first $n_0$ observation is from the placebo group with probability density
$d F_0(.) $, $n_1$ comes from the treatment group with probability density $\lambda d F_0()+(1-\lambda)d F_1(.)$, we assume the presence of a
treatment effect, i.e., the need for $\lambda<1$. The data also include $m_0$ independent, censored observations from the placebo group, and
$m_1$ independent, censored observations from the treatment group, where $m_{z j}$ observations are at moments $c_{z j}\left(j=1, \ldots, d_z,
z=0,1, m_{z  1}+\ldots+m_{z d_z}=m_z\right)$. Let $N_0\left(= n_0+ m_0\right)$ and $N_1\left(= n_1 +m_1\right)$ be the number of observations
in the two groups and $N\left(=N_0+N_1\right)$ be the total number of observations. Now consider the discrete distribution at the point in
time when the observed event occurs only, and let $p_i=F_0\left(\left[x_i\right]\right)$ The non-parametric log-likelihood function
\cite{owen2001empirical} is the

\begin{align*}l\left(\theta, F_0\right)&=\sum_{i=n_0+1}^n \log \left(\lambda+(1-\lambda) \frac{\exp \left\{h\left(x_i,
\beta\right)\right\}}{\sum p_i \exp \left\{h\left(x_i, \beta\right)\right\}}\right)+  \sum_{j=1}^{d_0} m_{0 j} \log \sum_{i=1}^n p_i
I\left\{x_i>c_{0 j}\right\} \\&+\sum_{j=1}^{d_1} m_{1 j} \log \sum_{i=1}^n p_i\left(\lambda+(1-\lambda) \frac{\exp \left\{h\left(x_i,
\beta\right)\right\}}{\sum p_i \exp \left\{h\left(x_i, \beta\right)\right\}}\right) I\left\{x_i>c_{1 j}\right\}+\sum_{i=1}^n \log
p_i,\end{align*}where the constraint $\sum_{i=1}^n p_i=1$. Let $\theta=\left(\lambda, \beta^T\right)^T$, and define the profile log-likelihood
(profile log-likelihood) of $\theta$ as $p l_N(\theta)=\max _{F_0} l(\theta$, $\left.F_0\right)$. The nonparametric great likelihood
estimation (NPMLE) of $\theta$ is obtained by maximizing $p l_{N}(\theta)$, i.e., $\hat{\theta}=\operatorname{argmax} p l_N(\theta)$. $p
l_{N}(\theta)$ can be obtained by the EM algorithm, and $\hat{\theta}$ can be obtained by the Newton-Raphson method, i.e., the mixing
probability $\lambda$ and the treatment effect $\beta$ can be obtained.

To evaluate the survival probability of the responder and non-responder populations at time point $t_0$, naturally, we can
obtain:$$\left(\begin{array}{c}
\widehat{S}_0\left(t_0\right) \\
\widehat{S}_1\left(t_0\right)
\end{array}\right)=\left(\begin{array}{c}
\sum_{i=1}^n \hat{p}_i I\left\{x_i>t_0\right\} \\
\sum_{i=1}^n \hat{p}_i \exp \left\{h\left(x_i, \hat{\beta}\right)\right\} I\left\{x_i>t_0\right\} / \sum_{i=1}^n \widehat{p}_i \exp
\left\{h\left(x_i, \hat{\beta}\right)\right\}
\end{array}\right),
$$
where $\widehat{S}_0\left(t_0\right)$,$\widehat{S}_1\left(t_0\right)$ represent the survival functions in the responder and non-responder
populations in the treatment group, respectively.

ETMM has the advantage of studying the heterogeneity of treatment effects in randomized clinical trials. However, ETMM assumes a sufficiently
long follow-up period. A shorter follow-up period would affect the efficacy of the model. Another limitation of the model is the assumption
that the censoring time is discrete. Due to the technical complexity of studying the asymptotic nature of the model parameters, it is
challenging to extend the approach to continuous censoring times\cite{wang_exponential_2014}.

\subsubsection{Mixture Weibull Model}

Liao et al.\cite{liao_flexible_2019}\cite{jia_inferring_2021} propose to extend the mixture Weibull model by modeling the mixing probabilities
of potential subgroups depending on covariates and using a Bayesian information criterion to select the number of potential subgroups. In
addition, they consider the inclusion of an adaptive LASSO penalty term for variable selection and demonstrate that the estimator has oracle
properties.

Weibull distribution of the k-th subgroup  $S\left(t, \eta_{\mathrm{k}}\right)=\exp
\left\{-\left(\frac{t}{\lambda_k}\right)^{\kappa_k}\right\}$, where
$\eta_{\mathrm{k}}=\left(\kappa_{\mathrm{k}}, \lambda_{\mathrm{k}}\right)^{\mathrm{T}}, \kappa_{\mathrm{k}}$ is the shape parameter and
$\lambda_k$ is the scale parameter.

Assume $T$ is the survival time and $X$ is the baseline covariance (the first column of X is constant at 1), introducing a potential subgroup
$B$. Assume that
$$
P(T>t \mid B=k, X)=S\left(t, \eta_k\right),$$
$$P(B=k \mid X)=\frac{\exp \left\{\beta_k^{\mathrm{T}} \mathbf{X}\right\}}{\sum_{k=1}^K \exp \left\{\beta_{\mathrm{k}}^{\mathrm{T}}
\mathbf{X}\right\}}=\pi_k(X, \beta),$$
$k=1, \ldots, K$, where $\beta_1=0$ and $\beta_2, \ldots, \beta_{\mathrm{K}}$ are unknown parameters.

Suppose we have $n$ independently and identically distributed right-censored samples, denoted as:
$$
\left\{Y_i=T_i \wedge C_i, \Delta_i=I\left(T_i \leq C_i\right), X_i, i=1, \ldots, n\right\},
$$where $C_i$ is the censoring time. Assuming that the survival time $T_i$ is independent of the censoring time $C_i$ for a given $X_i$, and
the censoring time $C_i$ is also independent of the potential subgroup $B_i$, we can obtain the observed log-likelihood function as
$$
\begin{aligned}
    &l_{n, \text { obs }}(\theta)\\&=\sum_{i=1}^n\left[\Delta_i \log \left\{\sum_{k=1}^K f\left(Y_i, \eta_{\mathrm{k}}\right)
    \pi_{\mathrm{k}}\left(\mathrm{X}_{\mathrm{i}} ; \beta\right)\right\}+\left(1-\Delta_i\right) \log \left\{\sum_{k=1}^K S\left(Y_i,
    \eta_{\mathrm{k}}\right) \pi_{\mathrm{k}}\left(\mathrm{X}_{\mathrm{i}} ; \beta\right)\right\}\right],
\end{aligned}
$$where $\theta=\left(\eta^{\mathrm{T}}, \beta^{\mathrm{T}}\right)^{\mathrm{T}}, \eta^{\mathrm{T}}=\left(\eta_1^{\mathrm{T}}, \ldots,
\eta_{\mathrm{K}}^{\mathrm{T}}\right)^{\mathrm{T}}$, and $f\left(t, \eta_{\mathrm{k}}\right)=-\mathrm{S}^{\prime}\left(\mathrm{t},
\eta_{\mathrm{k}}\right)$.

By EM algorithm in E-step, we can get the posterior probability of the i-th sample in the k-th subgroup $$q_{i k}=\frac{f\left(Y_i,
\boldsymbol{\eta}_{\mathbf{k}}\right) \boldsymbol{\pi}_{\mathbf{k}}\left(\mathbf{X}_{\mathbf{i}} ; \boldsymbol{\beta}\right)}{\sum_{k=1}^K
f\left(Y_i, \boldsymbol{\eta}_{\mathbf{k}}\right) \pi_{\mathbf{k}}\left(\mathbf{X}_{\mathbf{i}} ; \boldsymbol{\beta}\right)},$$
$$l_n(\eta, \boldsymbol{\beta})=\sum_{\mathbf{i}=\mathbf{1}}^{\mathbf{n}} \sum_{\mathrm{k}=\mathbf{1}}^{\mathrm{K}}
\mathbf{q}_{\mathrm{ik}}\left[\boldsymbol{\Delta}_{\mathbf{i}} \log \left(\mathbf{f}\left(\mathbf{Y}_{\mathbf{i}},
\eta_{\mathbf{k}}\right)\right)+\left(\mathbf{1}-\boldsymbol{\Delta}_{\mathbf{i}}\right) \log \left(\mathbf{S}\left(\mathbf{Y}_{\mathbf{i}},
\eta_{\mathrm{k}}\right)\right)+\log \left(\pi_{\mathrm{k}}\left(\mathbf{X}_{\mathbf{i}} ; \boldsymbol{\beta}\right)\right)\right].$$

In the M-step, $\eta$ and $\boldsymbol{\beta}$ are computed by the Newton-Raphson method. Denote $\theta=(\eta, \boldsymbol{\beta})$. Denote $\nabla
l_n(\theta)=\frac{\partial l_n(\theta)}{\partial \theta}$ and $\nabla^2 l_n(\theta)=\frac{\partial^2 l_n(\theta)}{\partial \theta \partial
\theta^T}$. Then the (t+1)th iteration Newton-Raphson update is

$$\theta^{(t+1)}=\theta^{(t)}-\left.\left(\left.\nabla^2 l_n(\theta)\right|_{\theta=\theta^{(t)}}\right)^{-1} \nabla
l_n(\theta)\right|_{\theta=\theta^{(t)}}.$$

$l_n(\eta, \boldsymbol{\beta})$ increases at each iteration of the M-step until $l_n(\eta, \boldsymbol{\beta})$ converges to the end of the
iteration of the algorithm. For the setting of the number of potential subgroups, the Bayesian information criterion (BIC) can be used to
select the number of potential subgroups.

After adding the adaptive LASSO penalty term, the objective function becomes    $$
-l_n(\tilde{\eta}, \boldsymbol{\beta})+\lambda \sum_{k=1}^K \sum_{j=1}^d \frac{\left|\beta_{k j}\right|}{\left|\tilde{\beta}_{k j}\right|^{\gamma}},
$$
where $\gamma$ is the number set in advance to be greater than 0. It is usually set to $\boldsymbol{\beta}=\left(\beta_{11}, \beta_{12}, \ldots, \beta_{1
d}, \beta_{21}, \ldots, \beta_{K d}\right)^T$
,$\gamma=1$,$\tilde{\eta}$ is the maximum likelihood estimate and is not used as a penalty term.

To obtain the estimate $\hat{\boldsymbol{\beta}}$ of the adaptive LASSO, a two-step minimization of the objective function is used. The first step uses
$l_n(\eta, \boldsymbol{\beta})$ to compute the maximum likelihood estimate $(\tilde{\eta}, \tilde{\boldsymbol{\beta}})$ by the Newton-Raphson method. The
second step is performed by the coordinate descent algorithm to obtain the $\hat{\boldsymbol{\beta}}$ that minimizes the objective function.

This model can come to detect potential subgroups of individuals with different survival characteristics and identify important covariates
associated with the assignment of potential subgroup members. And the estimators are consistent and oracle in nature. Data with a high number
of covariates can be well handled by a penalized objective function. The method may serve as an exploratory step in clinical trials before
conducting subgroup analyses to study treatment effects. The important covariates selected may help to clearly identify subgroups and to
discover how patients in different subgroups respond differently to the treatment. Specific treatments may then be developed for specific
patient groups.

\subsubsection{Logistic-Cox Mixture Model}

Wu et al. (2016)\cite{wu_subgroup_2016} proposed a statistical method based on a semiparametric Logistic-Cox mixed model for analyzing
right-censored, time-to-event data.
Assuming n samples receive one of the two treatments, suppose $T_i, i=1, \ldots, n$ is the event time of interest, $\xi_i \in\{0,1\}$ is an
unobservable subgroup indicator variable, $\mathrm{Z}_i$ is a $p_1$ dimensionally observable covariate related to subgroup effects, $\mathrm
{X}_i$ is a $p_2$-dimensional observable covariate related to subgroup classification. Suppose $\lambda_{T_i \mid \xi_i, \mathbf{Z}_i,
\mathbf{X}_i}\left(t \mid \xi_i, \mathbf{Z}_i, \mathbf{X}_i\right)$ is $T_i$ given $\xi_i, \mathrm{Z}_i$ and $\mathrm{X}_i$ as a conditional
hazard function.
$$\lambda_{T_i \mid \xi_i, \mathbf{Z}_i, \mathbf{X}_i}\left(t \mid \xi_i, \mathbf{Z}_i, \mathbf{X}_i\right)=\lambda(t) \exp
\left\{\left(\beta_1+\beta_2 \xi_i\right)^{\top} \mathbf{Z}_i\right\},$$ where the conditional hazard function satisfies $\lambda$ is the
baseline hazard function and $\beta_1, \beta_2 \in \mathbb{R}^{p_1}$ is the unknown regression coefficient.

Assume
$$
\mathrm{P}\left(\xi_i=1 \mid \mathbf{X}_i, \mathbf{Z}_i\right)=\frac{\exp \left(\gamma^{\top} \mathbf{X}_i\right)}{1+\exp \left(\gamma^{\top}
\mathbf{X}_i\right)},
$$
where $\gamma \in \mathbb{R}^p$ is the unknown parameter.  The first element of $X_i$ is set to 1 , and the first element of $\mathbf{Z}_i$ is
the oscillatory variable indicating whether or not to receive treatment.
For each patient $i=1, \ldots, n$, let $C_i$ be the right censoring time, given $\left(\mathrm{Z}_i, \mathrm{X}_i\right)$ conditionally
independent of $\left(T_i, \xi_i\right)$. When right-censored data exist, we can only observe $Y_i=\min \left\{T_i, C_i\right\}$ and
$\Delta_i=I\left\{T_i \leqslant C_i\right\}$, where $I\{\}$ is the indicator function. The observed data are then $\left(Y_i, \Delta_i,
\mathbf{Z}_i, \mathbf{X}_i\right), i=1, \ldots, n$. Let $\Lambda(y)=\int_0^y \lambda(t) d t$,  $i=1, \ldots, n$

Denote
$$
\begin{gathered}
p\left(\gamma^{\top} \mathbf{X}_i\right)=\frac{\exp \left(\gamma^{\top} \mathbf{X}_i\right)}{1+{\exp \left(\gamma^{\top}
\mathbf{X}_i\right)}}, \\
g_1\left(Y_i, \Delta_i, \mathbf{Z}_i ; \beta_1, \Lambda\right)=\left\{\lambda\left(Y_i\right) \exp \left(\beta_1^{\top}
\mathbf{Z}_i\right)\right\}^{\Delta_i} \exp \left\{-\Lambda\left(Y_i\right) \exp \left(\beta_1^{\top} \mathbf{Z}_i\right)\right\},
\end{gathered}
$$
$$
g_2\left(Y_i, \Delta_i, \mathbf{Z}_i ; \beta_1, \beta_2, \Lambda\right)=\left[\lambda\left(Y_i\right) \exp
\left\{\left(\beta_1+\beta_2\right)^{\top} Z_i\right\}\right]^{\Delta_i} \exp [-\Lambda\left(Y_i\right) \exp
\left\{\left(\beta_1+\beta_2\right)^{\top} Z_i\right\}].
$$

Then the likelihood function of the observed data is a mixture of two subgroups  $$
\begin{aligned}
    &L\left(\beta_1, \beta_2, \Lambda, \gamma\right)\\&=\prod_{i=1}^n\left[p\left(\gamma^{\top} \mathbf{X}_i\right) g_2\left(Y_i, \Delta_i,
    \mathbf{Z}_i ; \beta_1, \beta_2, \Lambda\right)+\left\{1-p\left(\gamma^{\top} \mathbf{X}_i\right)\right\} g_1\left(Y_i, \Delta_i,
    \mathbf{Z}_i ; \beta_1, \Lambda\right)\right].
\end{aligned}
$$

Notice that $L\left(\beta_1, \beta_2, \Lambda, \gamma\right)$ involves the non-parameter $\Lambda$. To solve it we discretize it.
Specifically, restrict $\Lambda$ to be a step function with non-negative jumps only on $Y_i$'s, $i=1, \ldots, n$. Replace
$\lambda\left(Y_i\right)$ and $\Lambda\left(Y_i\right)$ by $\lambda_i$ and $\sum_{j: Y_j \ leqslant Y_i} \lambda_j$. So the unknown parameters
can be summarized as $\vartheta=\left(\beta_1, \beta_2, \lambda_1, \ldots, \lambda_n\right)$.
The authors use the EM algorithm. Let $\vartheta_{(0)}$ be the initial value of $\vartheta$.

E-step
$$\begin{aligned}
\mathrm{w}_i(\gamma, \vartheta)  & =\mathrm{P}\left(\xi_i=1 \mid Y_i, \Delta_i, \mathbf{Z}_i, \mathbf{X}_i ; \gamma, \vartheta\right) \\
& =\frac{p\left(\gamma^{\top} \mathbf{X}_i\right) g_2\left(Y_i, \Delta_i, \mathbf{Z}_i ; \beta_1, \beta_2,
\Lambda\right)}{p\left(\gamma^{\top} \mathbf{X}_i\right) g_2\left(Y_i, \Delta_i, \mathbf{Z}_i ; \beta_1, \beta_2,
\Lambda\right)+\left\{1-p\left(\gamma^{\top} \mathbf{X}_i\right)\right\} g_1\left(Y_i, \Delta_i, \mathbf{Z}_i ; \beta_1, \Lambda\right)}.
\end{aligned}$$

M-step
$$\begin{aligned}
&l_1^E\left(\vartheta_{;} \gamma, \vartheta_{(0)}\right) \\& =\sum_{i=1}^n\left(w_i\left(\gamma,
\vartheta_{(0)}\right)\left[\Delta_i\left\{\log \lambda_i+\left(\beta_1+\beta_2\right)^{\top} Z_i\right\}-\sum_{j: Y_j \leqslant Y_i}
\lambda_j \exp \left\{\left(\beta_1+\beta_2\right)^{\top} Z_i\right\}\right]\right. \\
& \left.+\left\{1-w_i\left(\gamma, \vartheta_{(0)}\right)\right\}\left[\Delta_i\left\{\log \lambda_i+\beta_1^{\top} Z_i\right\}-\sum_{j: Y_j
\leqslant Y_i} \lambda_j \exp \left(\beta_1^{\top} Z_i\right)\right]\right).
\end{aligned}$$

Given $\gamma$, we maximize $L(\vartheta, \gamma)$ to obtain an estimate of $\vartheta$ and constrain $\lambda_i \geqslant 0, i=1, \ldots, n$.
For the $\lambda$ estimate then the profile likelihood function can be maximized by fixing $\boldsymbol{\beta}$, so that $\partial l_1^E\left(\vartheta ;
\gamma, \vartheta_{(0)}\right) / \partial \lambda_i=0$ can be obtained:
$$
\hat{\lambda}_i=\frac{\Delta_i}{\sum_{Y_j \geq Y_i}\left[w_j\left(\gamma, \vartheta_{(0)}\right) \exp
\left\{\left(\beta_1+\beta_2\right)^{\top} Z_j\right\}+\left(1-w_j\left(\gamma, \vartheta_{(0)}\right)\right) \exp \left(\beta_1^{\top}
Z_j\right)\right]}.
$$

Substituting this expression back into $l_1^E$ reduces the problem to  maximization
$$\begin{aligned}
&p l_1^E\left(\beta_1, \beta_2 ; \gamma, \vartheta_{(0)}\right)\\= & \sum_{i=1}^n \Delta_i\left(\left(\beta_1+w_i\left(\gamma,
\vartheta_{(0)}\right) \beta_2\right)^{\top} \mathbf{Z}_i\right. \\
& \left.-\log \left\{\sum_{Y_j \geq Y_i}\left[w_j\left(\gamma, \vartheta_{(0)}\right) \exp \left\{\left(\beta_1+\beta_2\right)^{\top}
Z_j\right\}+\left(1-w_j\left(\gamma, \vartheta_{(0)}\right)\right) \exp \left(\beta_1^{\top} Z_j\right)\right]\right\}\right).
\end{aligned}$$

For the estimate of $\boldsymbol{\beta}$, it can be considered as the log-partial likelihood of the weighted Cox model, then the estimate of $\boldsymbol{\beta}$ can be
easily obtained.

The Logistic-Cox mixture model is an extension of the Logistic-normal mixture model proposed by Shen and He (2015)\cite{shen_inference_2015}
from a parametric model to a semi-parametric model, and the theoretical properties have been proven to be effective. In addition, in order to
adjust the bias caused by right-censored data and improve the performance of the model in limited samples, Wu et al.
(2016)\cite{wu_subgroup_2016} designed a bootstrap method for adjustment. Simulation experiments show that the method is feasible at the
finite sample size.

\subsubsection{Finite Mixture Cox Model and its Fitting Algorithm}

Although the Cox proportional hazards model has been the most commonly used regression model for censored data, its nonparametric nature makes
the finite mixture Cox model more difficult to be built. Eng and Hanlon (2014)\cite{eng_discrete_2014} proposed a Cox-assisted clustering
algorithm using the EM algorithm to fit a finite mixture Cox model, which efficiently clusters different groups for censored data. You et al.
(2018)\cite{you2018subtype}
extended it by adding the penalty term of the adaptive LASSO and provided the asymptotic nature of the theory, indicating that it has oracle
properties and noting that the estimator has a convergence rate of $\sqrt{n}$.

Suppose $T$ is the survival time or censoring time of a person, and $\boldsymbol{x}$ is a baseline covariate of dimension p. Denote whether T
is censored by $\delta=0$ or 1. Suppose there are $K \geq 2$ subgroups , and the corresponding mixing probability is $\boldsymbol{\Pi}=\left(\pi_1, \pi_2,
\ldots, \pi_K\right)$, satisfying the condition $\sum_{i=1}^K \pi_k=1$.  Then the joint density function of $(T, \delta)$ can be written as
$f(t, \delta \mid x)=\sum_{k=1}^K \pi_k f_k(t, \delta \mid \boldsymbol{x})$, where $f_k(t, \delta \mid \boldsymbol{x})$ is the density
function of the $ k$-th subgroup, $k=1,2, \ldots, K$.

Let each subgroup satisfy the proportional hazard assumption, and the hazard of the $k$th subgroup is $h_k(t \mid \boldsymbol{x})=h_{0 k}(t)
\exp \left(\boldsymbol{\beta}_k^{\mathrm{T}} \boldsymbol{x}\right)$, $h_{ 0 k}(t)$ is the baseline hazard function for the $k$-th subgroup,
and $\boldsymbol{\beta}_k$ is the effect parameter corresponding to $\boldsymbol{x}$ in the $k$-th subgroup. Then
$$f_k(t, \delta \mid \textbf{x})=h_{0 k}^\delta(t) \exp \left(\delta \boldsymbol{\beta}_k^{\mathrm{T}} \textbf{x}\right) \exp \left\{-\exp
\left(\boldsymbol{\beta}_k^{\mathrm{T}} \textbf{x}\right) \int_0^t h_{0 k}(u) \mathrm{d} u\right\}.$$

Given independent samples $\left(t_i, \delta_i, \boldsymbol{x}_i\right), i=1,2, \ldots, n$,  $\boldsymbol{Y}=\left(t_1, t_2, \ldots, t_n\right),
\boldsymbol{\Delta}=\left(\delta_1, \delta_2, \ldots, \delta_n\right)$, $  \boldsymbol{X}=\left(x_1, \boldsymbol{x}_2, \ldots, x_n\right)$, $\boldsymbol{H(t)}(t)=\left(h_{01}(t),
h_{02}(t), \ldots, h_{0 K}(t)\right)$,  $B=\left(\boldsymbol{\beta}_1^T, \boldsymbol{\beta}_2^T, \ldots \beta_K^T\right)^T$.
The log-likelihood function is then
$$
l_{o b s}(\boldsymbol{\Pi}, \boldsymbol{H}, B ; \boldsymbol{Y}, \boldsymbol{\Delta} \mid \boldsymbol{X})=\sum_{i=1}^n \log \sum_{k=1}^K \pi_k f_k\left(t_i, \delta_i \mid x_i\right).
$$

To facilitate the computation, the EM algorithm can be used for iteration. Depending on whether the $i$-th sample is from the $k$-th subgroup,
let $u_{i k}$ be denoted as the latent indicator variable, and let $\boldsymbol{U}=\left(u_{i k}\right)_{n \times K}$.

The complete log-likelihood function is then $$
l_c(\boldsymbol{\Pi}, \boldsymbol{H}, B ; \boldsymbol{Y}, \boldsymbol{\Delta}, \boldsymbol{U} \mid \boldsymbol{X})=\sum_{i=1}^n \sum_{k=1}^K u_{i k} \log \pi_k+\sum_{i=1}^n \sum_{k=1}^K u_{i k} \log f_k\left(t_i,
\delta_i \mid x_i\right).
$$

For the Eth step of the mth iteration, calculate the posterior probability $$
u_{i k}^{(m+1)}=\frac{\pi_k^{(m)} f_k\left(t_i, \delta_i \mid \boldsymbol{x}_i, h_{0 k}^{(m)}(\cdot),
\boldsymbol{\beta}_k^{(m)}\right)}{\sum_{k^{\prime}=1}^K \pi_{k^{\prime}}^{(m)} f_{k^{\prime}}\left(t_i, \delta_i \mid \boldsymbol{x}_i, h_{0
k^{\prime}}^{(m)}(\cdot), \boldsymbol{\beta}_{k^{\prime}}^{(m)}\right)}.
$$

For the mth iteration of the Mth step update $$\pi_k^{(m+1)}=\frac{\sum_{i=1}^n u_{i k}^{(m+1)}}{n}.$$

denote $\tilde{l}_{c, k}=\sum_{i=1}^n u_{i k}^{(m+1)} \log f_k\left(t_i, \delta_i \mid x_i\right).$
According to Breslow\cite{breslow1972contribution}, the update formula for the baseline hazard function can be obtained by deriving $\partial
\tilde{l}_{c, k} / \partial h_i^{\prime}=0$:
$$h_{0 k}^{(m+1)}\left(t_i\right)=\frac{u_{i k}^{(m+1)}}{\sum_{j: t_i \leq t_j} u_{j k}^{(m+1)} \exp
\left(\boldsymbol{\beta}_k^{(m+1)^{\mathrm{T}}} \boldsymbol{x}_j\right)}.$$

Finally, the update of $\boldsymbol{\beta}_k$ for the k-th subgroup can be obtained by fitting a weighted Cox model using common software such
as the "survivor" package in R.

When the assumption of proportional hazard does not hold, the finite mixture Cox model is able to relax the assumption of  proportional hazard
for each component, thus solving the heterogeneity problem. The theoretical properties of this model have been proven and its excellent
performance has been verified by a large number of simulation experiments. However, as with the finite mixture model, the setting of the
number of groups remains a challenging issue, while the selection of initial values may have an impact on the convergence of the model.

\subsection{A Dual Cox Model And SPIRLS-EM Algorithm
}

\subsubsection{Motivation}

Traditionally, chemotherapeutic agents have been evaluated using WHO\cite{world1979handbook} or RECIST guidelines
\cite{therasse_new_2000}\cite{eisenhauer2009new}, which assume that the drugs inhibits early tumor growth and recommend discontinuation of the
therapy when tumor progression is detected. Although RECIST provides a practical approach to tumor response assessment and the method is
widely accepted as a standardized measure, its limitations in targeted therapeutic agents are well recognized. The major limitations that
commonly affect the assessment results regardless of tumor type or pathogen include variability in tumor size measurement and tumor
heterogeneity. To accurately evaluate the efficacy of various tumor-targeted therapeutic agents, guidelines for the evaluation of targeted
therapies have been proposed successively, including the mRECIST criteria\cite{lencioni_modified_2010}, Choi criteria
\cite{choi2007correlation}, SACT criteria \cite{ smith2010assessing}, etc. Research in immuno-oncology in recent years has shown that
immunotherapy may have clinical benefits that differ from those of cytotoxic drugs and that, in some cases, it may be inappropriate to stop
immunotherapy as soon as disease progression is seen. A new set of evaluation criteria has been proposed, called
irRC\cite{wolchok2009guidelines}, which considers measurable total tumor load and allows for the possibility of delayed benefit and durable
stable disease\cite{chen_statistical_2013}. Although different oncology treatment guidelines have different criteria for response, they all
classify treatment outcomes into four levels: complete response, partial response, stable disease, and progressive disease. In clinical
practice, the objective response rate is usually used as the main criterion to assess the treatment effect and is calculated as (complete
response + partial response)/total number of patients treated.

Although targeted therapies and immunotherapies have different mechanisms of action, they share a common feather that there are different
efficacy profiles for different people. With targeted therapies, only a subset of patients will respond to any particular new drug because of
the genetic diversity of human tumors across tissue types and the heterogeneity of each patient. For immunotherapy, immunotherapeutic drugs
may not stimulate the patient's immune system to recognize cancer cells, and the patient may not respond well to the drug. In specific
clinical trials, targeted therapy or immunotherapy is often used as the experimental group, while conventional treatment modalities serve as
the control group. However, because the mechanism and principle of action of the experimental group are different from that of the control
group, there is a lack of criteria to assess objective response rates in both groups simultaneously. In this issue, the individual response
can be observed in the experimental group but not in the control group. To address this issue, appropriate semi-supervised statistical models
are needed to reduce the sample size and time required to complete clinical trials, thereby reducing development costs and increasing
development speed.

As discussed in Section 1, the challenges posed by targeted and immunotherapy include heterogeneity among and within patients, long-term
survival, and delayed clinical effects. Multiple statistical modeling approaches have been developed to address these challenges. The
nonparametric Kaplan-Meier curve does not need to make assumptions about the distribution of data and can compare the survival of multiple
groups of data at the same time. However, Kaplan-Meier curves cannot incorporate covariates and predict survival as the Cox proportional
hazards model does. In addition, the estimation results of this curve may not be stable enough for small sample studies. The Cox proportional
hazards model has good interpretability and can estimate the impact of multiple risk factors, and these effects can be used to predict the
survival time of a specific individual. However, when the assumption of proportional hazards does not hold, such as in the case of crossing
Kaplan-Meier curves, the predictive power of the model may be reduced. The piecewise proportional hazards model is able to restrict the
proportional hazards assumption to different time intervals, thus adapting to situations when the proportional hazards assumption does not
hold. However, the piecewise proportional hazards model lacks a clear time division standard, which makes it challenging to accurately
interpret the model, and the way of modeling each time period separately makes the fitting error of the model large. The cure model can be
used to analyze long-term survivors but may rely too heavily on parametric assumptions for the way incidence and latency are modeled
separately. Multiple different mixture models work well to cluster different subgroups and study the efficacy of treatments. However, there is
still a lack of research on the mixture model in the semi-supervised learning scenario in the field of survival analysis. Therefore, a
semi-supervised finite mixture Cox model is proposed in this paper, and the number of subgroups is set to two, which we call the Dual Cox
model. The Dual Cox model relaxes the assumption of proportional hazards to the responder and non-responder populations, which has an
excellent biological interpretation. The EM algorithm is used to predict and classify the patients who have not been observed to respond to
the drug so as to find the best population division. After dividing the population, the Dual Cox model can analyze the drug efficacy of the
responder and non-responder populations, which is of great clinical significance. It can better evaluate the difference in drug efficacy and
provide practical guidance for the optimization of treatment plans.

The basic form of the Dual Cox model is as follows

$$S(t,\delta\mid\boldsymbol{x})=\pi S_1(t,\delta\mid\boldsymbol{x})+(1-\pi)S_2(t,\delta\mid\boldsymbol{x})$$

where $\pi$ is the objective response rate, $S_1(t,\delta\mid\boldsymbol{x})$ and $S_2(t,\delta\mid\boldsymbol{x})$ are the survival functions
corresponding to responders and non-responders, respectively.

It is noted that in a clinical trial, disease progression and patient's responses to the treatment are always recorded. Therefore, there is no
missing data issue for such subgroup division based on patient's response profiles.

\subsubsection{A Dual Cox Model}
\noindent{\textbullet}\quad\parbox[t]{0.9\textwidth}{\textbf{Assumptions and Notations
}}

Let T denote the survival or censoring time of an individual with p-dimensional covariates $\boldsymbol{x}$, and let $\delta$=0, 1 be an
indicator function indicating whether T is censored, with 0 being censored and 1 is not censored. We denote the vector of survival times,
censoring indicators, and covariate vectors as $\boldsymbol{Y}=\left(t_1, t_2, \ldots, t_n\right), \boldsymbol{\Delta}=\left(\delta_1, \delta_2, \ldots,
\delta_n\right),  \boldsymbol{X}=\left (\boldsymbol{x}_1, \boldsymbol{x}_2, \ldots, \boldsymbol{x}_n\right)^\mathrm{T}$. There are K=2 subgroups with the
first subgroup  being the responders, and the second subgroup being the non-responsders. The mixing probability $\boldsymbol{\Pi} =(\pi_{1},\pi_{2})$
satisfying $\sum_{k=1}^{2} \pi_k=1$, where $\pi_{1}$ represents the objective response rate and $\pi_{2}$ represents the objective
non-response rate. Let the joint density function of $(T, \delta)$ be
$f(t,\delta\mid\boldsymbol{x})=\pi_{1}f_1(t,\delta\mid\boldsymbol{x})+\pi_{2}f_2(t,\delta\mid \boldsymbol{x})$, where $f_k(t,\delta \mid
\boldsymbol{x})$ denotes the density function of the k-th subgroup, k=1, 2.
Assume that the proportional hazards assumption is satisfied in each subgroup. The hazard of the $k$-th group is then $h_k\left(t, \delta \mid
\boldsymbol{x}\right)=h_{0 k}(t) \exp \left(\boldsymbol{\beta}_k^{\mathrm{T}} \boldsymbol{x}\right)$, where $h_ {0 k}(t)$ is the baseline
hazard function for the $k$-th subgroup, and $\boldsymbol{\beta}_k$ is the  effect
parameter corresponding to $\boldsymbol{x}$ in the $k$-th subgroup. Let the vector of regression coefficients, the vector of baseline hazards,
be $\boldsymbol{\beta}=\left(\boldsymbol{\beta}_1^T, \boldsymbol{\beta}_2^T\right)^T$,$\boldsymbol{H(t)}=\left(h_{01}(t), h_{02 }(t)\right). $

In our medical context, the experimental group samples have been observed whether they respond to the drug or not (labeled), while the control
group is not observed (unlabeled). Suppose that $\left(t_i, \delta_i, \boldsymbol{x}_i\right), i=1,2, \ldots, n$ are independently and
identically distributed right-censored samples, and that $t_i$ and $\delta_i$ are independent of each other given $\boldsymbol{x}_i$. Denote
the set of the experimental group (labeled) samples is $R_{l}=\{i \mid i \in $ experiment group$\}$ and the set of the control group
(unlabeled) samples is $R_{u}=\{i \mid i \in $ control  group $\}$. Denote $n_l$, $n_u$ as the number of samples in the experimental group and
the control group respectively, and the total number of samples is $n=n_u+n_l$.

\noindent{\textbullet}\quad\parbox[t]{0.9\textwidth}{\textbf{ SPIRLS-EM Algorithm}}

Denote the Cox proportional hazards density function belonging to the k-th subgroup (k=1 is response group and k=2 is non-response group)
as$$\begin{aligned}
f_k(t, \delta \mid \boldsymbol{x})&=\left\{h_k\left(t, \delta \mid \boldsymbol{x}\right)\right\}^{\delta} \cdot S_k\left(t, \delta \mid
\boldsymbol{x}\right)\\
&=h_{0 k}^\delta(t) \exp \left(\delta \boldsymbol{\beta}_k^{\mathrm{T}} \boldsymbol{x}\right) \exp \left\{-\exp
\left(\boldsymbol{\beta}_k^{\mathrm{T}} \boldsymbol{x}\right) \int_0^t h_{0 k}(u) \mathrm{d} u\right\}.\end{aligned}$$

Because the experimental group has been labeled, we can write the density function for the experimental group known to be in the k-th subgroup
as $\pi_kf_k(t, \delta \mid \boldsymbol{x})$ and the density function for the control group can be written as $f(t, \delta \mid
\boldsymbol{x})=\sum_{k=1}^2 \pi_k f_k( t, \delta \mid \boldsymbol{x}).$ Assume that $z_{i k}$ is based on whether the experimental group
sample is from the k-th subgroup and is 1 if yes, otherwise 0, where $i \in R_{l}$.

So the log-likelihood function based on the observed data is then  \begin{equation}
l_{o b s}(\boldsymbol{\Pi}, \boldsymbol{H}, \boldsymbol{\beta} ; \boldsymbol{Y}, \boldsymbol{\Delta} \mid \boldsymbol{X})=\sum_{i \in R_{l}} \sum_{k=1}^2z_{ik}\log \pi_k f_k(t_i, \delta_i \mid
\boldsymbol{x}_i)+\sum_{i \in R_{u}}  \log
 \sum_{k=1}^2 \pi_k f_k\left(t_i, \delta_i \mid \boldsymbol{x}_i\right).
\end{equation}

Depending on whether the $i$-th sample is from the $k$-th subgroup, let $u_{i k}$ be denoted as the latent indicator variable, $i \in R_{u}$,
k=1, 2. Let $\boldsymbol{U}=\left(u_{i k}\right)_{n_u \times 2}$
, then the complete log-likelihood function is
$$\begin{aligned}
\footnotesize l_{c}(\boldsymbol{\Pi}, \boldsymbol{H}, \boldsymbol{\beta} ; \boldsymbol{Y}, \boldsymbol{\Delta} , \boldsymbol{U} \mid \boldsymbol{X})=& \sum_{i \in R_{l}} \sum_{k=1}^2z_{ik}\log \pi_k +\sum_{i \in R_{u}}
\sum_{k=1}^2 u_{ik}\log  \pi_k \\
&+\sum_{i \in R_{l}} \sum_{k=1}^2z_{ik}\log  f_k(t_i, \delta_i \mid \boldsymbol{x}_i)+\sum_{i \in R_{u}} \sum_{k=1}^2 u_{ik}\log f_k\left(t_i,
\delta_i \mid \boldsymbol{x}_i\right) .
\end{aligned}$$

The (m + 1)-th iteration of the E-step, $i \in R_{u}$, k=1,2, we can obtain by calculating the posterior probability that   \begin{equation}
\begin{aligned}
 u_{ik}^{(m+1)}=E\left(u_{ik} \mid \boldsymbol{X}, \boldsymbol{\beta}^{(m)}, \boldsymbol{\Pi^{(m)}}, \boldsymbol{H^{(m)}} \right) &=P\left(u_{i k}=1 \mid \boldsymbol{X},
 \boldsymbol{\beta}^{(m)}, \boldsymbol{\Pi^{(m)}}, \boldsymbol{H^{(m)}}\right) \\
&=\frac{\pi_k^{(m)} f_k\left(t_i, \delta_i \mid \boldsymbol{x}_i, h_{0 k}^{(m)}(\cdot),
\boldsymbol{\beta}_k^{(m)}\right)}{\sum_{k^{\prime}=1}^2 \pi_{k^{\prime}}^{(m)} f_{k^{\prime}}\left(t_i, \delta_i \mid \boldsymbol{x}_i, h_{0
k^{\prime}}^{(m)}(\cdot), \boldsymbol{\beta}_{k^{\prime}}^{(m)}\right)}.
\end{aligned}
\end{equation}

Let the Q-function of the (m + 1)-th iteration be $E_{\boldsymbol{U} \mid \boldsymbol{X}}[l_{c}(\boldsymbol{\Pi}, \boldsymbol{H}, \boldsymbol{\beta} ; \boldsymbol{Y}, \boldsymbol{\Delta} , \boldsymbol{U} \mid \boldsymbol{X})\mid
\boldsymbol{X},\boldsymbol{\Pi^{(m)}},\boldsymbol{\beta}^{(m)},\boldsymbol{U^{(m+1)}}]$, and substitute (3) into the Q-function to obtain
\begin{equation}
\begin{aligned}
Q(\boldsymbol{\beta},\boldsymbol{H};\boldsymbol{\beta}^{(m)},\boldsymbol{H^{(m)}})&=E_{\boldsymbol{U^{(m+1)}} \mid \boldsymbol{X}}[l_{c}(\boldsymbol{\Pi}, \boldsymbol{H}, \boldsymbol{\beta} ; \boldsymbol{Y}, \boldsymbol{\Delta} , \boldsymbol{U} \mid \boldsymbol{X})\mid
\boldsymbol{X},\boldsymbol{\Pi^{(m)}},\boldsymbol{\beta}^{(m)}, \boldsymbol{U^{(m+1)}}]\\
&=\sum_{i \in R_{l}} \sum_{k=1}^2z_{ik}\log \pi_k +\sum_{i \in R_{u}} \sum_{k=1}^2 u_{ik}^{(m+1)}\log  \pi_k \\
&+\sum_{i \in R_{l}} \sum_{k=1}^2z_{ik}\log  f_k(t_i, \delta_i \mid \boldsymbol{x}_i)+\sum_{i \in R_{u}} \sum_{k=1}^2 u_{ik}^{(m+1)}\log
f_k\left(t_i, \delta_i \mid \boldsymbol{x}_i\right).
\end{aligned}
\end{equation}

The M-step is derived for $\pi_k$ in the function $\{Q(\boldsymbol{\beta},\boldsymbol{H};\boldsymbol{\beta}^{(m)},\boldsymbol{H^{(m)}})+\lambda(\sum_{k=1}^{2}
\pi_k-1)\}$, where $\lambda$ is the Lagrange multiplier. By taking the derivative with respect to $\pi_k$, we can get

\begin{equation}
\begin{aligned}
\pi_k^{(m+1)}=\frac{\sum_{i \in R_{u}} u_{i k}^{(m+1)}+\sum_{i \in R_{l}}z_{ik}}{n_{u}+n_{l}}.
\end{aligned}
\end{equation}

Next, we derive the updated formula for the following parameters

$$({\boldsymbol{\beta}}^{(m+1)},\boldsymbol{H^{(m+1)}})=\arg \max _{\boldsymbol{\beta}, \boldsymbol{H}} Q(\boldsymbol{\beta},\boldsymbol{H};\boldsymbol{\beta}^{(m)},\boldsymbol{H^{(m)}}).$$

Since the update of $\boldsymbol{\beta}$ only requires the last two items of (4) to be considered, we denote that $$\tilde{l}_{c,
k}\left(\boldsymbol{\beta}_k ; h_{0k}; \boldsymbol{Y}, \boldsymbol{\Delta} \mid \boldsymbol{X}, \boldsymbol{U^{(m+1)}}\right)= \sum_{i \in R_{u}} u_{i k}^{(m+1)} \log f_k\left(t_i, \delta_i \mid
\boldsymbol{x}_i\right)+\sum_{i \in R_{l}} z_{i k} \log f_k\left(t_i, \delta_i \mid \boldsymbol{x}_i\right). $$

Assuming that $t_i$ is not tied (all $t_i$ values are not the same), define $D_u$, $D_l$ ,  $D\subseteq\{1, \ldots, n\}$, where $D_u=\left\{I|
\delta_I=1, I \in R_u \right\}$, $D_l=\left\{I| \delta_I=1, I \in R_l \right\}$. $D=\left\{I| \delta_I=1 \right\}$. For the nonparametric part
$h_{0 k}(\cdot)$, we denote $h_{0 k}(t)$ at time $t_i$ as $h_i^{\prime}$ and $h_i^{\prime}=0$ when $\delta_i=0$, first we rewrite
$\tilde{l}_{c, k}\left(\boldsymbol{\beta}_k ; h_{0k}; \boldsymbol{Y}, \boldsymbol{\Delta} \mid \boldsymbol{X}, \boldsymbol{U^{(m+1)}}\right)$ and then according to
Breslow\cite{breslow1972contribution} to calculate $ \partial \tilde{l}_{c, k} / \partial h_i^{\prime}=0.$

$$\begin{aligned}
& \tilde{l}_{c, k}\left(\boldsymbol{\beta}_k ; h_{0k}; \boldsymbol{Y}, \boldsymbol{\Delta} \mid \boldsymbol{X},  \boldsymbol{U^{(m+1)}}\right)=\sum_{i \in R_u} u_{i k}^{(m+1)} \log f_k\left(t_i,
\delta_i \mid \boldsymbol{x}_i\right) +\sum_{i \in R_{l}} z_{i k} \log f_k\left(t_i, \delta_i \mid \boldsymbol{x}_i\right)\\ & =\left\{\sum_{i
\in D_u} u_{i k}^{(m+1)}\left[\left(\log \left(h_i^{\prime}\right)+\boldsymbol{\beta}_k^T \boldsymbol{x}_{i}\right)\right]-\sum_{i \in R_u}
u_{i k}^{(m+1)}\left[\sum_{j: t_j \leq t_i} h_j^{\prime} \exp \left(\boldsymbol{\beta}_k^{\top}
\boldsymbol{x}_{i}\right)\right]\right\}\\&+\left\{\sum_{i \in D_l} z_{i k}\left[\left(\log \left(h_i^{\prime}\right)+\boldsymbol{\beta}_k^T
\boldsymbol{x}_{i}\right)\right]-\sum_{i \in R_l} z_{i k}\left[\sum_{j: t_j \leq t_i} h_j^{\prime} \exp \left(\boldsymbol{\beta}_k^{\top}
\boldsymbol{x}_{i}\right)\right]\right\}  \\
& =\left\{\sum_{i \in D_u} u_{i k}^{(m+1)}\left[\left(\log \left(h_i^{\prime}\right)+\boldsymbol{\beta}_k^T \boldsymbol{x}_{i}\right)\right]-
\sum_{i \in D} h_i^{\prime} \left[ \sum_{\substack{j: t_j \geq t_i \\ j \in R_u}}u_{j k}^{(m+1)} \exp \left(\boldsymbol{\beta}_k^{\top}
\boldsymbol{x}_{j}\right)\right]\right\}\\&+\left\{\sum_{i \in D_l} z_{i k}\left[\left(\log \left(h_i^{\prime}\right)+\boldsymbol{\beta}_k^T
\boldsymbol{x}_{i}\right)\right]-\sum_{i \in D} h_i^{\prime}\left[ \sum_{\substack{j: t_j \geq t_i \\ j \in R_l}}z_{j k} \exp
\left(\boldsymbol{\beta}_k^{\top} \boldsymbol{x}_{j}\right)\right]\right\}.
\end{aligned}$$
$$
\frac{\partial \tilde{l}_{c, k}}{\partial h_i^{\prime}}=\left(\frac{ u_{i k}^{(m+1)}}{h_i^{\prime}}-\sum_{\substack{j: t_j \geq t_i \\ j \in
R_u}} u_{j k}^{(m+1)} \exp \left(\boldsymbol{\beta}_k^{T} \boldsymbol{x}_i\right)-\sum_{\substack{j: t_j \geq t_i \\ j \in R_l}} z_{j k} \exp
\left(\boldsymbol{\beta}_k^{T} \boldsymbol{x}_i\right)\right)=0 , \quad i \in D_u .$$

$$\frac{\partial \tilde{l}_{c, k}}{\partial h_i^{\prime}}=\left(\frac{ z_{ik}}{h_i^{\prime}}-\sum_{\substack{j: t_j \geq t_i \\ j \in R_u}}
u_{j k}^{(m+1)} \exp \left(\boldsymbol{\beta}_k^{T} \boldsymbol{x}_i\right)-\sum_{\substack{j: t_j \geq t_i \\ j \in R_l}} z_{j k} \exp
\left(\boldsymbol{\beta}_k^{T} \boldsymbol{x}_i\right)\right)=0 , \quad i \in D_l .$$

Thus, $h_{0 k}$  is then updated by
\begin{equation}
h_{0 k}^{(m+1)}\left(t_i\right)=\frac{ u_{i k}^{(m+1)}}{\sum_{j: t_i \leq t_j,j\in R_{u}} u_{j k}^{(m+1)} \exp
\left(\boldsymbol{\beta}_k^{(m+1)^{\mathrm{T}}} \boldsymbol{x}_j\right)+\sum_{j: t_i \leq t_j,j\in R_{l}}z_{jk}\exp
\left(\boldsymbol{\beta}_k^{(m+1)^{\mathrm{T}}} \boldsymbol{x}_j\right)}, i \in D_u.
\end{equation}

\begin{equation}
h_{0 k}^{(m+1)}\left(t_i\right)=\frac{z_{ik}}{\sum_{j: t_i \leq t_j,j\in R_{u}} u_{j k}^{(m+1)} \exp
\left(\boldsymbol{\beta}_k^{(m+1)^{\mathrm{T}}} \boldsymbol{x}_j\right)+\sum_{j: t_i \leq t_j,j\in R_{l}}z_{jk}\exp
\left(\boldsymbol{\beta}_k^{(m+1)^{\mathrm{T}}} \boldsymbol{x}_j\right)},  i\in D_l.
\end{equation}

For the update of $\boldsymbol{\beta}^{(m+1)}$, the profile likelihood can be obtained by taking $h_{0 k}^{(m+1)}$ back to $\tilde{l}_{c, k}$,
which gives
$$\begin{aligned}&\tilde{l}_{c, k}\left(\boldsymbol{\beta}_k ; \boldsymbol{Y}, \boldsymbol{\Delta} \mid \boldsymbol{X}, \boldsymbol{U^{(m+1)}}^{(m+1)}\right)\\&=\sum_{i \in R_u} \delta_i u_{i
k}^{(m+1)}\left[\beta_k^{\mathrm{T}} \boldsymbol{x}_i-\log \left\{\sum_{\substack{j: t_j \geq t_i \\ j \in R_u}} u_{j k}^{(m+1)} \exp
\left(\beta_k^{\mathrm{T}} \boldsymbol{x}_j\right)+\sum_{\substack{j: t_j \geq t_i \\ j \in R_l}} z_{j k}\exp
\left(\boldsymbol{\beta}_k^{\mathrm{T}} \boldsymbol{x}_j\right)\right\}\right]\\
&+\sum_{i\in R_l} \delta_i z_{i k}\left[\boldsymbol{\beta}_k^{\mathrm{T}} \boldsymbol{x}_i-\log \left\{\sum_{\substack{j: t_j \geq t_i \\ j
\in R_u}} u_{j k}^{(m+1)} \exp \left(\beta_k^{\mathrm{T}} \boldsymbol{x}_j\right)+\sum_{\substack{j: t_j \geq t_i \\ j \in R_l}} z_{j k}\exp
\left(\beta_k^{\mathrm{T}} \boldsymbol{x}_j\right)\right\}\right]\\
&+ \sum_{i\in D_u} u_{ik}^{(m+1)}+\sum_{i\in D_l} z_{ik}+\sum_{i \in R_u} \delta_i u_{i k}^{(m+1)} \log u_{i k}^{(m+1)}+\sum_{i\in R_l}
\delta_i z_{i k} \log z_{ik},
\end{aligned}$$where if $u_{ik}^{(m+1)}=0$ then $\delta_i u_{i k}^{(m+1)} \log u_{i k}^{(m+1)}=0$, and if $z_{ik}=0$ then $\delta_i z_{i k}
\log z_{ik}=0$.

Removing the terms unrelated to $\boldsymbol{\beta}$, it is known that $\tilde{l}_{c, k}$ is the log-partial likelihood. Therefore, the update
of $\boldsymbol{\beta}^{(m+1)}$ can be regarded as a weighted Cox proportional hazards model, i.e., $\boldsymbol{\beta}^{(m+1)}$ can be
obtained by the iteratively
reweighted least squares (IRLS). As iterations proceed in the EM algorithm, the observed log-likelihood function is greater than or equal to
the value of the previous iteration each time \cite{dempster_maximum_1977}, meaning $ \boldsymbol{\Pi^{(m+1)}}, u_{ik}^{(m+1)} $ , $
{\boldsymbol{\beta}}^{(m+1) }, \boldsymbol{H^{(m+1)}}$ is a better estimate than $ \boldsymbol{\Pi^{(m)}}, u_{ik}^{(m)} $, $ {\boldsymbol{\beta}}^{(m)}, \boldsymbol{H^{(m)}}$.

\noindent{\textbullet}\quad\parbox[t]{0.9\textwidth}{\textbf{Convergence criteria}}

Finally, the algorithm ends the iterations according to the following two criteria that hold simultaneously:
\begin{itemize}
    \item Absolute convergence criterion: The EM algorithm is considered to have converged if the difference between the current
        log-likelihood function value $l(\theta^{(t)})$ and the log-likelihood function value $l(\theta^{(t-1)})$ of the previous iteration
        is less than some threshold $\epsilon$,
    \begin{equation}|l(\theta^{(t)}) - l(\theta^{(t-1)})| < \epsilon. \end{equation}
    \item Relative convergence criterion: The EM algorithm is considered to have converged if the difference between the log-likelihood
        function value $l(\theta^{(t)})$ and the log-likelihood function value $l(\theta^{(t-1)})$ of the previous iteration divided by the
        log-likelihood function value is less than some threshold $\epsilon$,
    \begin{equation}|(l(\theta^{(t)}) - l(\theta^{(t-1)}))/l(\theta^{(t)})| < \epsilon.\end{equation}
\end{itemize}

\noindent{\textbullet}\quad\parbox[t]{0.9\textwidth}{\textbf{Classification Rule}}

Denote the values of the final convergence of the parameters as $\tilde{\boldsymbol{\Pi}}, \tilde{\boldsymbol{\beta}}$ and $\tilde{\boldsymbol{H}}$. For sample
$\left(t_m, \delta_m, \boldsymbol{x}_m\right)$, $m\in R_u$. The sample $\left(t_m, \delta_m, \boldsymbol{x}_m\right)$ belongs to which
subgroup is determined by
$$
\underset{k}{\operatorname{argmax} \hat u_{mk}},
$$where$$
\hat u_{mk}=P\left(\widetilde{\boldsymbol{\Pi}}, \widetilde{\boldsymbol{H}}, \widetilde{B} ; k_m=k \mid\left(\boldsymbol{x}_m, t_m,
\delta_m\right)\right)=\frac{\widetilde{\pi}_k f_k\left(t_m, \delta_m \mid \boldsymbol{x}_m, \tilde{h}_{0 k}(\cdot),
\tilde{\boldsymbol{\beta_k}}\right)}{\sum_{k=1}^2 \tilde{\pi}_{k} f_{k}\left(t_m, \delta_m \mid \boldsymbol{x}_m, \tilde{h}_{0
k^{\prime}}(\cdot), \widetilde{\boldsymbol{\beta_{k}}} \right)}.
$$

\noindent{\textbullet}\quad\parbox[t]{0.9\textwidth}{\textbf{Initial Values}}

One of the main advantages of the EM algorithm is that it always converges. However, the disadvantage is that depending on the initial values,
the algorithm may converge to a local maximum rather than a global one \cite{tanner1993tools}. In addition, McLachlan and Peel
(2000)\cite{mclachlan_finite_2000} mention that when the initial values are chosen close to the boundary, there may be no way to make the
parameters converge, as will be verified in the numerical simulations in Section \ref{sec:dual4}. Therefore, the initial value setting is very
important for our algorithm.

Eng and Hanlon (2014)\cite{eng_discrete_2014} contribute to the topic in the finite mixture Cox model by stating that in order to attain the global optimum as much as feasible, their algorithm randomly selects $u_{ik}^{(0)}$ as the initial values and then estimates $\boldsymbol{\beta}^ {(0)}$ and $\pi^{(0)}$. In practice, they run multiple experiments with random initial values and choose the one that maximizes the value of the log-likelihood function.

Under semi-supervised learning, since $\pi_k^{(0)}$ can be estimated from the labeled data (experimental group), $z_{ik}$ is used to estimate
$\boldsymbol{\Pi^{(0)}}$, and then the probability of $\boldsymbol{\Pi^{(0)}}$ can be taken as the prior information of each unlabeled sample, i.e.,
$\boldsymbol{\Pi^{(0)}}=(u_{i1}^{ (0)},u_{i2}^{(0)})$, and finally estimate $\boldsymbol{\beta}^{(0)}$, so that the information of the labeled data
(experimental group) can be maximally utilized. Besides, we can also use $z_{ik}$ to estimate $\boldsymbol{\Pi^{(0)}}$, and then randomly select different
initial values of $u_{ik}^{(0)}$ for multiple experiments, and the final result is chosen as the initial values corresponding to the maximum
likelihood estimators. In the simulation experiments in Section \ref{sec:dual4}, our results illustrate that both methods can achieve
relatively high values for the log-likelihood function.

\noindent{\textbullet}\quad\parbox[t]{0.9\textwidth}{\textbf{SPIRLS-EM Algorithm}}

Following the IRLS-EM algorithm in \cite{you2018subtype}, our fitting algorithm, called as SPIRLS-EM, is summarized as follows.

\begin{algorithm}[H]
    \KwIn{Observed data for the experimental group (labeled) and the control group (unlabeled)}
    \lFor{initial values}{ Let $m=0$, estimate $\boldsymbol{\Pi^{(0)}}$ based on the labeled sample $z_{ik}$. Initialize $u_{ik}^{(0)}$ and
    $\boldsymbol{\beta}^{(0)}$, and estimate $\boldsymbol{H^{(0)}}$ according to (5) and (6)}
   \While {(1) and (3) do not satisfy the convergence conditions (7) and (8)}
{
E-step: Compute (2) to update $u_{ik}^{(m+1)}$ for all unlabeled data;\\
M-step: Use $u_{ik}^{(m+1)}$ and $z_{ik}$ to compute (4) to update $\boldsymbol{\Pi^{(m+1)}}$;\\
\For{k=1,2}
{
Update $\boldsymbol{\beta}^{(m+1)},$ using $u_{ik}^{(m+1)}$ and $z_{ik}$ as the weights of the Cox model;\\
Update $h_{0k}^{(m+1)}(\cdot)$ according to (5) and (6);
}
m=m+1
}
    \KwOut{$\tilde{\boldsymbol{\Pi}}, \tilde {\boldsymbol{U}}, \tilde {\boldsymbol{\beta}}, \tilde{\boldsymbol{H}}$}
    \caption{SPIRLS-EM Algorithm}
    \label{algo:sgd}
\end{algorithm}

\section{Theoretical Properties}
\label{sec:dual3}
\subsection{Regularity Conditions}

Anderson and Gill (1982) \cite{andersen1982cox} proved the theoretical asymptotic properties of partial likelihood estimators in the Cox proportional hazards model, and You et
al. (2018) \cite{you2018subtype} proved the theoretical asymptotic properties of IRSL algorithm in the finite mixture Cox model, and we combine their work to
give consistency and asymptotic normality of SPRIRL-EM estimators in the Dual Cox model.

As described in Section \ref{sec:dual}, subsection 2, the mixing probability vector $\boldsymbol{\Pi} =(\pi_{1},\pi_{2}), \boldsymbol{\beta}_k$ is the
corresponding effect parameter of $\boldsymbol{x}$ in the $k$-th subgroup, and $\boldsymbol{\beta}=\left(\boldsymbol{\beta}_1^T,
\boldsymbol{\beta}_2^T\right)^T$. Let $\boldsymbol{\beta}_0$ be the true value of $\boldsymbol{\beta}$, and $\boldsymbol{\beta}_{0k}$ be the
true value of $\boldsymbol{\beta}_{k}$, where k=1, 2. Without loss of generality, we assume that
$T \in[0,1]$.

To ensure the consistency and asymptotic normality of estimators in the Dual Cox model, we consider the following regularity conditions:

\begin{enumerate}
\item[A.] For $k=1,2$, we have $\int_0^1 h_{0 k}(t) dt<\infty$.
\item[B.] Denote $Q_1\left(t_i\right)=I\left\{T \geq t_i, i \in R_l\right\}, Q_2\left(t_i\right)=I\left\{T \geq t_i, i \in R_u\right\}.$ The
    process $Q(t)=I\{T \geq t\}$ is left-continuous with right-hand limits,  and $P\left(Q(t)=1, \forall t \in[0,1]\right)>0$;
\item[C.] For $k=1,2$, there exists a neighborhood $\mathcal{B}$ of $\boldsymbol{\beta}_0$ such that
$$\quad \mathrm{E}\left\{\sup _{t \in[0.1] \cdot \boldsymbol{\beta} \in \mathcal{B}} \left(\sum_{k=1}^2 V_k Q_2(t)
\boldsymbol{x}^{\mathrm{T}} \boldsymbol{x} \exp \left(\boldsymbol{\beta}_k^{\mathrm{T}} \boldsymbol{x}\right)+ \sum_{k=1}^2 z_k Q_1(t)
\boldsymbol{x}^{\mathrm{T}} \boldsymbol{x} \exp \left(\boldsymbol{\beta}_k^{\mathrm{T}} \boldsymbol{x}\right)\right)\right\}<\infty,$$ $V_k$
is an indicator variable for whether the sample is from the k-th component, whose follows the Bernoulli distribution with probability $
\pi_k $. $z_k$ is a labeled value, indicating whether the experimental group samples come from the k-th component;

\item[D.] Let $w_k^{(0)}\left(\boldsymbol{\beta}_k, t\right)=\mathrm{E} V_k Q_2(t) \exp \left(\boldsymbol{\beta}_k^{\mathrm{T}}
    \boldsymbol{x}\right), w_k^{(1)}\left(\boldsymbol{\beta}_k, t\right)=\mathrm{E} V_k Q_2(t) x \exp
    \left(\boldsymbol{\beta}_k^{\mathrm{T}} \boldsymbol{x}\right)$, and $w_k^{(2)}\left(\boldsymbol{\beta}_k, t\right)=\mathrm{E} V_k Q_2(t)
    \boldsymbol{x} \boldsymbol{x}^{\mathrm{T}} \exp \left(\boldsymbol{\beta}_k^{\mathrm{T}} \boldsymbol{x}\right)$, where $w_k^{(0)}(\cdot,
    t), w_k^{(1)}(\cdot, t)$, and $w_k^{(2)}(\cdot, t)$  are continuous
in $\boldsymbol{\beta}_k \in \mathcal{B}_k$,  uniformly in $t \in[0,1]$, $w_k^{(0)}(\cdot, \cdot), w_k^{(1)}(\cdot, \cdot)$, and
$w_k^{(2)}(\cdot, \cdot)$  are bounded on $\mathcal{B}_k \times[0,1]$, and $w_k^{(0)}(\cdot, \cdot)$ is bounded away from zero on
$\mathcal{B}_k \times[0,1]$. The matrix
$$
\begin{aligned}
I_{u,k}\left(\boldsymbol{\beta}_{0 k}\right)= & \int_0^1\left\{\frac{w_k^{(2)}\left(\boldsymbol{\beta}_{0 k},
t\right)}{w_k^{(0)}\left(\boldsymbol{\beta}_{0 k}, t\right)}-\left(\frac{w_k^{(1)}\left(\boldsymbol{\beta}_{0 k},
t\right)}{w_k^{(0)}\left(\boldsymbol{\beta}_{0 k}, t\right)}\right)\left(\frac{w_k^{(1)}\left(\boldsymbol{\beta}_{0 k},
t\right)}{w_k^{(0)}\left(\boldsymbol{\beta}_{0 k}, t\right)}\right)^{\mathrm{T}}\right\} \\
& \times w_k^{(0)}\left(\boldsymbol{\beta}_{0 k}, t\right) h_{0 k}(t) d t
\end{aligned}
$$

is finite positive definite;

\item[E.] Let $s_k^{(0)}\left(\boldsymbol{\beta}_k, t\right)=\mathrm{E}  z_k Q_1(t) \exp \left(\boldsymbol{\beta}_k^{\mathrm{T}}
    \boldsymbol{x}\right), s_k^{(1)}\left(\boldsymbol{\beta}_k, t\right)=\mathrm{E} z_k Q_1(t) x \exp
    \left(\boldsymbol{\beta}_k^{\mathrm{T}} \boldsymbol{x}\right)$, and $s_k^{(2)}\left(\boldsymbol{\beta}_k, t\right)=\mathrm{E} z_k Q_1(t)
    \boldsymbol{x} \boldsymbol{x}^{\mathrm{T}} \exp \left(\boldsymbol{\beta}_k^{\mathrm{T}} \boldsymbol{x}\right),$
where $s_k^{(0)}(\cdot, t), s_k^{(1)}(\cdot, t)$, and $s_k^{(2)}(\cdot, t)$ are continuous
in $s_k \in \mathcal{B}_k$, uniformly in $t \in[0,1]$, $s_k^{(0)}(\cdot, \cdot), s_k^{(1)}(\cdot, \cdot)$, and $s_k^{(2)}(\cdot, \cdot)$ are
bounded on $\mathcal{B}_k \times[0,1]$, and $s_k^{(0)}(\cdot, \cdot)$ is bounded away from zero on $\mathcal{B}_k \times[0,1]$.  The matrix
$$
\begin{aligned}
I_{l,k}\left(\boldsymbol{\beta}_{0 k}\right)= & \int_0^1\left\{\frac{s_k^{(2)}\left(\boldsymbol{\beta}_{0 k},
t\right)}{s_k^{(0)}\left(\boldsymbol{\beta}_{0 k}, t\right)}-\left(\frac{s_k^{(1)}\left(\boldsymbol{\beta}_{0 k},
t\right)}{s_k^{(0)}\left(\boldsymbol{\beta}_{0 k}, t\right)}\right)\left(\frac{s_k^{(1)}\left(\boldsymbol{\beta}_{0 k},
t\right)}{s_k^{(0)}\left(\boldsymbol{\beta}_{0 k}, t\right)}\right)^{\mathrm{T}}\right\} \\
& \times s_k^{(0)}\left(\boldsymbol{\beta}_{0 k}, t\right) h_{0 k}(t) d t
\end{aligned}
$$
is finite positive definite.

\end{enumerate}

Regularity condition A guarantees that the baseline cumulative hazard is finite, and regularity condition B ensures that the law of large
numbers can be satisfied. The regularity conditions A-E ensure that the local asymptotic quadratic (LAQ) property holds for the partial
likelihood function, allowing the  asymptotic normality of the maximum partial likelihood estimates to be valid\cite{andersen1982cox}.

\subsection{Consistency and Asymptotic Normality}

For the sake of concise description and proof, let Z be the random variable of the sample labels in the experimental group, Z=1 for the
responders, and Z=2 for the non-responders. The sample ($\boldsymbol{x},Z$) is generated from $f(t_i, \delta_i,Z_i \mid \boldsymbol{x})$, and
is a labeled sample if the value of Z is known, or an unlabeled sample if the value of Z is unknown. For the density function with labels then
it can be written as $f(t_i, \delta_i,Z_i=k \mid \boldsymbol{x})=\pi_kf_k(t_i, \delta_i \mid \boldsymbol{x}), i\in R_l, k=1,2$. For the
density function without labels then it can be written as $f(t_i, \delta_i \mid \boldsymbol{x})=\pi_1f_1(t_i, \delta_i \mid
\boldsymbol{x})+\pi_2f_2(t_i, \delta_i \mid \boldsymbol{x}), i\in R _u$ .

Then we can rewrite the likelihood function based on the observed data as
\begin{equation}
L(\boldsymbol{\Pi}, \boldsymbol{H}, \boldsymbol{\beta} ; \boldsymbol{Y}, \boldsymbol{\Delta} \mid \boldsymbol{X})=\prod_{i\in R_l} f(t_i, \delta_i,Z_i \mid \boldsymbol{x}) \prod_{i\in R_u}f(t_i,
\delta_i \mid \boldsymbol{x}).
\end{equation}

To investigate the asymptotic behavior of the estimator that maximizes (3.1), assume that the number of samples for the labeled data
(experimental group) and unlabeled data (control group) are $n_l$ and $n_u$, respectively, and assume that when $n=n_l+n_u \rightarrow
\infty$, $\lambda=\frac{n_l}{n_l+n_u}$ is the probability of a sample belonging to a labeled sample and $\lambda$ is a known constant
independent of the sample. We know that, for sufficiently large n, we can approximate that for labeled samples $n_l= n\lambda \rightarrow
\infty$ and for unlabeled samples $n_u= n(1-\lambda) \rightarrow \infty$. In this way, the asymptotic properties of the maximum likelihood
estimator $\boldsymbol{\beta}$ in (10) can be obtained\cite{dillon2010asymptotic}.

To

\begin{lemma}\label{lemma:lemma1}
    Cozman and Cohen (2003) Assume that$\quad \left(t_i, \delta_i, \boldsymbol{x}_i\right), i=1,2, \ldots, n \quad$ are
independently and identically distributed, $t_i$ and $\delta_i$ are conditionally independent given $\boldsymbol{x}_i$, and suppose
$\lambda=\frac{n_l}{n_l+n_u}$ is the probability of belonging to a labeled sample and is constant, then the maximum likelihood estimate on
$\boldsymbol{\Pi}, \boldsymbol{H}, \boldsymbol{\beta}$ in (3.1) is equivalent to obtaining
    \begin{align}
        \underset{\boldsymbol{\Pi}, \boldsymbol{H}, \boldsymbol{\beta}}{\operatorname{argmax}}\left\lbrace\lambda E_{f\left(t_i, \delta_i, {Z}_i \mid
        \boldsymbol{x}\right)}\left[\log  f\left(t_i, \delta_i, Z_i \mid \boldsymbol{x} \right)\right]+(1-\lambda)E_{f\left(t_i, \delta_i, Z_i
        \mid \boldsymbol{x}\right)}\left[\log f\left(t_i, \delta_i\mid \boldsymbol{x}\right)\right]\right\rbrace.
    \end{align}
\end{lemma}

Lemma \ref{lemma:lemma1} was proposed by Cozman and Cohen (2003)\cite{cozman2003semi}, and the resulting expression (11) is the objective
function we need, demonstrating that semi-supervised learning may be seen as a "convex" combination of unsupervised learning $E_{f\left(t_i,
\delta_i, Z_i \mid \boldsymbol{x}\right)}\left[\log f\left(t_i, \delta_i\mid \boldsymbol{x}\right)\right]\rbrace$ and supervised learning
$E_{f\left(t_i, \delta_i, {Z}_i \mid \boldsymbol{x}\right)}\left[\log  f\left(t_i, \delta_i, Z_i \mid \boldsymbol{x} \right)\right]$. The
asymptotic behavior of the labeled samples can be ensured by the work of Andersen and Gill (1982)\cite{andersen1982cox}. The unlabeled
samples, on the other hand, can be guaranteed by the work of You et al. (2018)\cite{you2018subtype}. When the asymptotic behavior of the
labeled and unlabeled samples is guaranteed,  the asymptotic nature of the finite mixture Cox model estimators then extends well to the
semi-supervised case \cite{dillon2010asymptotic}.

\begin{proof}

Denote $\tilde {Z}$ as Z "plus" the random variable of the unlabeled sample, i.e., $\tilde {Z}$ is a random variable with three values, where
$\tilde {Z}=0$ represents the value represented by the unlabeled samples (control group), $\tilde {Z}\neq0$ represents the subgroup
represented by the labeled samples ($\tilde {Z}=1$ responders group and $\tilde {Z}=2$ non-responders group). Then $P(\tilde
{Z}=0)=1-\lambda$, $P(\tilde {Z}\neq0)=\lambda$.

Thus $$f(t, \delta,\tilde {Z}=z \mid \boldsymbol{x})=(\lambda f(t, \delta,Z=z \mid \boldsymbol{x}))^{I_{\lbrace\tilde
{Z}\neq0\rbrace}(z)}((1-\lambda) f(t, \delta \mid \boldsymbol{x}))^{I_{\lbrace\tilde {Z}=0\rbrace}(z)}.$$

The estimates of $\boldsymbol{\Pi}, \boldsymbol{H}, \boldsymbol{\beta}$ are
$$\underset{\boldsymbol{\Pi}, \boldsymbol{H}, \boldsymbol{\beta}}{\operatorname{argmax}} E_ {f\left(t, \delta, \tilde{Z} \mid \boldsymbol{x}\right)}\left(\log f\left(t,
\delta, \tilde{Z} \mid \boldsymbol{x}\right)\right)$$
$$\Downarrow$$
$$\underset{\boldsymbol{\Pi}, \boldsymbol{H}, \boldsymbol{\beta}}{\operatorname{argmax}} E_{f\left(t, \delta, \tilde{Z} \mid
\boldsymbol{x}\right)}\left[I_{\left\lbrace\tilde{Z} \neq 0\right\rbrace}\left(\tilde{Z}\right)\left(\log \lambda f\left(t, \delta, Z \mid
\boldsymbol{x}\right)\right)+I_{\left\lbrace\tilde{Z} = 0\right\rbrace}\left(\tilde{Z}\right)(\log (1-\lambda)f\left(t, \delta\mid
\boldsymbol{x}\right))\right].$$

As $\lambda$ does not depend on $\boldsymbol{\Pi}, \boldsymbol{H}, \boldsymbol{\beta}$, which is equivalent to maximizing

$$\begin{aligned} &E_{f\left(t, \delta, \tilde{Z} \mid \boldsymbol{x}\right)}\left[I_{\left\lbrace\tilde{Z} \neq
0\right\rbrace}\left(\tilde{Z}\right)\left(\log  f\left(t, \delta, Z \mid \boldsymbol{x}\right)\right)+I_{\left\lbrace\tilde{Z} =
0\right\rbrace}\left(\tilde{Z}\right)(\log f\left(t, \delta\mid \boldsymbol{x}\right))\right] \\
=&\lambda E_{f\left(t, \delta, \tilde{Z} \mid \boldsymbol{x}\right)}\log  f\left(t, \delta, Z \mid \boldsymbol{x},
\tilde{Z}\neq0\right)+(1-\lambda)E_{f\left(t, \delta, \tilde{Z} \mid \boldsymbol{x}\right)}\log f\left(t, \delta\mid
\boldsymbol{x},\tilde{Z}=0\right)\\
=&\lambda E_{f\left(t, \delta, {Z} \mid \boldsymbol{x}\right)}\log  f\left(t, \delta, Z \mid \boldsymbol{x}\right)+(1-\lambda)E_{f\left(t,
\delta, Z \mid \boldsymbol{x}\right)}\log f\left(t, \delta\mid \boldsymbol{x}\right),\\
\end{aligned}$$ where the second equal sign is because

$$f\left(t, \delta, Z \mid \boldsymbol{x}, \tilde{Z}\neq0\right)=  f\left(t, \delta, Z \mid \boldsymbol{x} \right),$$
$$ f\left(t, \delta \mid \boldsymbol{x}, \tilde{Z}=0\right)=f\left(t, \delta \mid \boldsymbol{x} \right).$$\end{proof}

\begin{lemma}\label{lemma:lemma2}
 Let $E$ be an open convex subset of $\mathbb{R}^p$, and let $F_1, F_2, \cdots$ be a sequence of random concave functions on $E$ such that
 $\forall x \in E$, $F_n(x) \xrightarrow{p} f(x)$ as $n \rightarrow \infty$ where $f$ is some real function on $E$, and $f$ has a unique
 maximum at $\hat{x} \in E$. Let $\hat{X}_n$ maximize $F_n$. Then $\hat{X}_n \xrightarrow{p} \hat{x}$ as $n \rightarrow \infty.$
\end{lemma}

Lemma \ref{lemma:lemma2} is important because it ensures that $\boldsymbol{\hat\beta}\xrightarrow{p} \boldsymbol{\beta}_0$ can hold. The
details and proof can be found in Appendix II of Anderson and Gill (1982)\cite{andersen1982cox}.

\begin{theorem}[Consistency]\label{thm:thm2}
    Under the conditions of Lemma \ref{lemma:lemma1} and the regularity conditions $A-E$, then $\boldsymbol{\hat\beta}\xrightarrow{p}
    \boldsymbol{\beta}_0$.
\end{theorem}

\begin{proof}

From Section\ref{sec:dual}, subsection 2, it can be obtained that
$$\begin{aligned}
\log_c L(\boldsymbol{\Pi}, \boldsymbol{H}, \boldsymbol{\beta} ; \boldsymbol{Y}, \boldsymbol{\Delta} \mid \boldsymbol{X})
&=\sum _{i\in R_l}\sum _{k=1}^{2}z_{ik}\log\pi_k+\sum _{i\in R_u}\sum _{k=1}^{2}u_{ik}\log\pi_k\\
+\sum _{i\in R_u}\sum _{k=1}^{2}u_{ik}&\log f_k\left(t_i, \delta_i \mid \boldsymbol{x}_i\right)+\sum _{i\in R_l}\sum _{k=1}^{2}z_{ik}\log
f_k\left(t_i, \delta_i \mid \boldsymbol{x}_i\right).
\end{aligned}$$

Since we only consider $\boldsymbol{\beta}$, then by removing the first two terms we can obtain $$\begin{aligned}\tilde{l}_{c}(
\boldsymbol{\beta} ; \boldsymbol{Y}, \boldsymbol{\Delta} \mid \boldsymbol{X},\boldsymbol{U},Z)&=\sum _{i\in R_u}\sum _{k=1}^{2}u_{ik}\log f_k\left(t_i, \delta_i \mid \boldsymbol{x}_i\right)+\sum
_{i\in R_l}\sum _{k=1}^{2}z_{ik}\log f_k\left(t_i, \delta_i \mid \boldsymbol{x}_i\right)\\
&=l_{c,u}( \boldsymbol{\beta} )+l_{c,l}( \boldsymbol{\beta} ),
\end{aligned}$$where$$\sum _{i\in R_u}\sum _{k=1}^{2}u_{ik}\log f_k\left(t_i, \delta_i \mid \boldsymbol{x}_i\right)=l_{c,u}(
\boldsymbol{\beta} ),$$
$$\sum _{i\in R_l}\sum _{k=1}^{2}z_{ik}\log f_k\left(t_i, \delta_i \mid \boldsymbol{x}_i\right)=l_{c,l}( \boldsymbol{\beta} ).$$

From the Lemma \ref{lemma:lemma1} we get

$$\underset{\boldsymbol{\beta}}{\operatorname{argmax}}E\left[\tilde{l}_{c}( \boldsymbol{\beta} ; \boldsymbol{Y}, \boldsymbol{\Delta} \mid \boldsymbol{X},\boldsymbol{U},Z)\right]$$
$$\Updownarrow$$
$$\underset{\boldsymbol{\beta}}{\operatorname{argmax}}\lbrace\lambda \cdot E\left(  l_{c,l}( \boldsymbol{\beta})\right)+(1-\lambda)\cdot
E\left( l_{c,u}( \boldsymbol{\beta} )\right)\rbrace.$$

For $l_{c,u}( \boldsymbol{\beta})- l_{c,u}( \boldsymbol{\beta}_0)$, it follows from the proof of Theorem I of You et al.
(2018)\cite{you2018subtype} that, under the regularity conditions A-D, it is sufficient to show that when $n_u$ is sufficiently large, for any
given $\epsilon_1>0$, there exists a constant $M_1>0$ such that $$|\frac{1}{n_u}(l_{c,u}(\boldsymbol{\beta}_0+n_u^{-1 / 2}
v_1)-l_{c,u}\left(\boldsymbol{\beta}_0 \right))|<\frac{1}{2(1-\lambda)}\cdot\epsilon_1,$$ where for unlabeled data, there exists an estimator
in the ball of  $\left\{\boldsymbol{\beta}_0+n_u^{-1 / 2} \boldsymbol{v_1 }:\|\boldsymbol{v_1}\| \leq M_1\right\}$ that make the partial
likelihood function locally maximal.

Similarly for $l_{c,l}( \boldsymbol{\beta} )-l_{c,l}( \boldsymbol{\beta}_0 )$, the proof of Lemma 3.1 of Andersen and Gill
(1982)\cite{andersen1982cox} yields, under the regularity conditions A, B, C, E, it is sufficient to show that when $n_l$ is large enough, for
any given $\epsilon_1>0$, there exists a constant $M_2>0$ such that $$|\frac{1}{n_l}(l_{c,l}(\boldsymbol{\beta}_0+n_l^{-1 / 2}
v_2)-l_{c,l}\left( \boldsymbol{\beta}_0 \right))|<\frac{1}{2\lambda}\epsilon_1,$$ where for the labeled data, there exists an estimator in the
ball of $\left\{\boldsymbol{\beta}_0+n_l^{-1 / 2} \boldsymbol{v_2}:\|\boldsymbol{v_2}\| \leq M_2\right\}$ that make the partial likelihood
function locally maximal.

Then there exists a constant M > 0 such that $$\left\{\boldsymbol{\beta}_0+n^{-1 / 2} \boldsymbol{v}:\|\boldsymbol{v}\| \leq M\right\} \subset
\left\{\boldsymbol{\beta}_0+n_u^{-1 / 2} \boldsymbol{v_1}:\|\boldsymbol{v_1}\| \leq M_1\right\} \cap \left\{\boldsymbol{\beta}_0+n_l^{-1 / 2}
\boldsymbol{v_2}:\|\boldsymbol{v_2}\| \leq M_2\right\}.$$

So, there exists a constant $M>0$ for any given $\epsilon_1>0$ such that
$$\lbrace\frac{1}{n_l}\lbrace\lambda \cdot l_{c,l}(\boldsymbol{\beta}_0+n^{-1 / 2} v )-\lambda \cdot l_{c,l}( \boldsymbol{\beta}_0
)\rbrace+\frac{1}{n_u}\lbrace(1-\lambda)\cdot l_{c,u}( \boldsymbol{\beta}_0+n^{-1 / 2} v)-(1-\lambda)\cdot l_{c,u}( \boldsymbol{\beta}_0
)\rbrace\rbrace<\epsilon_1.$$

This means that there exists an estimate $\boldsymbol{\hat\beta}$ in the ball
of $\left\{\boldsymbol{\beta}_0+n^{-1 / 2} \boldsymbol{v}:\|\boldsymbol{v}\| \leq M\right\}$  such that $\tilde{l}_{c}( \boldsymbol{\beta} ;
\boldsymbol{Y}, \boldsymbol{\Delta} \mid \boldsymbol{X},\boldsymbol{U},Z)$ is a local maximum.

Let $\tilde{l}_{c}( \boldsymbol{\beta} ; \boldsymbol{Y}, \boldsymbol{\Delta} \mid \boldsymbol{X},\boldsymbol{U},Z)=l_n(\boldsymbol{\beta})$, $\nabla l_n(\boldsymbol{\beta})=\partial
l_n(\boldsymbol{\beta}) / \partial \boldsymbol{\beta}$,  $\nabla^2 l_n(\boldsymbol{\beta})=\partial^2 l_n(\boldsymbol{\beta}) /\left\{\partial
\boldsymbol{\beta} \partial \boldsymbol{\beta}^{\mathrm{T}}\right\}.$

Let $I_u\left( \boldsymbol{\beta}_0\right)=\left[\begin{array}{lll}I_{u,1 }\left(\boldsymbol{\beta}_{01}\right) & & \\  \\ & & I_{u,
2}\left(\boldsymbol{\beta}_{0 2}\right)\end{array}\right].$

Let $I_l\left(\boldsymbol{\beta}_0\right)=\left[\begin{array}{lll}I_{l,1 }\left(\boldsymbol{\beta}_{01}\right) & & \\  \\ & & I_{l,
2}\left(\boldsymbol{\beta}_{0 2}\right)\end{array}\right].$

Let $I\left(\boldsymbol{\beta}_{0 }\right)=\lambda I_{u}(\boldsymbol{\beta}_{0 })+(1-\lambda)I_{l}(\boldsymbol{\beta}_{0})$,

where $I_{u,k}$ is the Fisher information matrix from the k-th subgroup of unlabeled data, $I_{l,k}$ is the Fisher information matrix from the
k-th subgroup of labeled data, k=1,2.

When n is large, we can obtain\begin{align*}
    \frac{\nabla l_n\left(\boldsymbol{\beta_0}\right)^{\mathrm{T}}}{n}& =\vec{\mathbf{0}}, \\
    \frac{\nabla^2 l_n\left(\boldsymbol{\beta_0}\right)}{n} &=-\left(\lambda \cdot E(\nabla^2 l_{c,l}( \boldsymbol{\beta}_0 ))+(1-\lambda)
    \cdot E(\nabla^2 l_{c,u}( \boldsymbol{\beta}_0 ))+o_p(1)\right)=-\left(I\left(\boldsymbol{\beta}_{0 }\right)+o_p(1)\right).
\end{align*}

The second-order Taylor expansion gives
$$
\begin{aligned}
& \frac{1}{n}\left\{l_n\left(\boldsymbol{\beta_0}+n^{-1 / 2} v\right)-l_n\left(\boldsymbol{\beta_0}\right)\right\} \\
= & \frac{1}{n} \nabla l_n\left(\boldsymbol{\beta_0}\right)^{\mathrm{T}} n^{-1 / 2} v+\frac{1}{2 n} v^{\mathrm{T}}\left\{\nabla^2
l_n\left(\boldsymbol{\beta_0}\right) / n\right\} v+\frac{1}{n} v^{\mathrm{T}} o_p(1) v\\
=& \frac{1}{n} \nabla l_n\left(\boldsymbol{\beta_0}\right)^{\mathrm{T}} n^{-1 / 2} v-\frac{1}{2 n}
v^{\mathrm{T}}\left\{I\left(\boldsymbol{\beta}_{0 }\right)+o_p(1)\right\} v+\frac{1}{n} v^{\mathrm{T}} o_p(1) v.
\end{aligned}
$$

From the regularity conditions D, E, we can get that $I\left(\boldsymbol{\beta}_{0 }\right)=\lambda I_{u}(\boldsymbol{\beta}_{0
})+(1-\lambda)I_{l}(\boldsymbol{\beta}_{0})$ is a finite positive definite matrix, and $I\left(\boldsymbol{\beta}\right)$ is a semi-positive
definite matrix. Then we have $v^{\mathrm{T}}\left\{I\left(\boldsymbol{\beta_0}\right)\right\} v>0$. Therefore
$l_n\left(\boldsymbol{\beta_0}+n^{-1 / 2} v\right)-l_n\left(\boldsymbol{\beta_0}\right)<0$, and when n is large enough,
$l_n\left(\boldsymbol{\beta_0}+n^{-1 / 2} v\right)$ converges in probability to $l_n\left(\boldsymbol{\beta_0} \right)$. These mean
$l_n(\boldsymbol{\beta})$ has a unique maximum value at $\boldsymbol{\beta}$=$\boldsymbol{\beta}_0$. In addition, $\boldsymbol{\hat{\beta}}$
is the estimator that maximize $l_n(\boldsymbol{\beta})$ in the ball of $\left\{\boldsymbol{\beta}_0+n^{-1 / 2}
\boldsymbol{v}:\|\boldsymbol{v}\| \leq M\right\}$, which follows from the Lemma \ref{lemma:lemma2}, we can obtain
$\boldsymbol{\hat\beta}\xrightarrow{p}\boldsymbol{\beta}_0$.\\
\end{proof}

\begin{theorem}[Asymptotic Normality]
\label{thm:thm3}
Under the conditions of the Lemma \ref{lemma:lemma1} and the Theorem \ref{thm:thm2}, then
$$\sqrt{n}\left(\boldsymbol{\hat\beta}_{k}-\boldsymbol{\beta}_{0k}\right) \xrightarrow{d} N\left(0, \left(\lambda
I_{u,k}(\boldsymbol{\beta}_{0 k})+(1-\lambda)I_{l,k}(\boldsymbol{\beta}_{0 k})\right)^{-1}\right),$$ where $I_{u,k}$ is the Fisher information
matrix from the k-th subgroup of unlabeled data, $I_{l,k}$ is the Fisher information matrix from the k-th subgroup of labeled data, k=1, 2.
\end{theorem}

\begin{proof}
Let $l_{c,l}( \boldsymbol{\beta})+l_{c,u}( \boldsymbol{\beta} )=l_n(\boldsymbol{\beta})$, $\nabla l_n(\boldsymbol{\beta})=\partial
l_n(\boldsymbol{\beta}) / \partial \boldsymbol{\beta}$ and $\nabla^2 l_n(\boldsymbol{\beta})=\partial^2 l_n(\boldsymbol{\beta})
/\left\{\partial \boldsymbol{\beta} \partial \boldsymbol{\beta}^{\mathrm{T}}\right\}.$

Let$$I^{(1)}_u\left( \boldsymbol{\beta}_0\right)=\left[\begin{array}{lll}I_{u,1 }\left(\boldsymbol{\beta}_{01}\right) & & \\  \\ & & I_{u,
2}\left(\boldsymbol{\beta}_{0 2}\right)\end{array}\right].$$

Let
$$I^{(1)}_l\left(\boldsymbol{\beta}_0\right)=\left[\begin{array}{lll}I_{l,1 }\left(\boldsymbol{\beta}_{01}\right) & & \\  \\ & & I_{l,
2}\left(\boldsymbol{\beta}_{0 2}\right)\end{array}\right].$$

For $\eta\in(0,1)$ , $\boldsymbol{\beta}^{\prime}=\boldsymbol{\beta}_0+\eta\left(\hat{\boldsymbol{\beta}}-\boldsymbol{\beta}_0\right)$. From
the second order Taylor expansion we get $$\nabla l_n(\boldsymbol{\hat\beta})=\nabla l_n(\boldsymbol{\beta_0})+\nabla^2
l_n(\boldsymbol{\beta^{\prime}})(\hat{\boldsymbol{\beta}}-\boldsymbol{\beta}_0).$$

Since $\boldsymbol{\hat\beta}$ is at the maximum point, we have $\nabla l_n(\boldsymbol{\hat\beta})=0$ and
$$\sqrt{n}\left(\boldsymbol{\hat\beta}-\boldsymbol{\beta}\right)=\left(\frac{-\nabla^2 l_n(\boldsymbol{\beta^{\prime}})}{n}\right)^{-1} \cdot
\frac{\nabla l_n(\boldsymbol{\beta})}{\sqrt{n}}.$$

And since  $\boldsymbol{\hat\beta}\xrightarrow{p}\boldsymbol{\beta}_0$ can be obtained from  Theorem \ref{thm:thm2}, we have
$\boldsymbol{\beta}^{\prime}\xrightarrow{p}\boldsymbol{\beta}_0$, and $\nabla^2 l_n(\boldsymbol{\beta})$ is a continuous function, then we can
get
$$\begin{aligned}
\left(\frac{-\nabla^2 \ell_n \left(\boldsymbol{\beta}^{\prime}\right)}{n}\right)^{-1} & \xrightarrow{p}\left(-\frac{\nabla^2
\ell_n\left(\boldsymbol{\beta}_0\right)}{n} \right)^{-1} \\
& \xrightarrow{p}\left(\lambda \cdot I_l\left(\boldsymbol{\beta}_0\right)+(1-\lambda) \cdot I_u\left(\boldsymbol{\beta}_0\right)\right)^{-1}.
\end{aligned}$$

Let
$$\begin{aligned}
\frac{1}{\sqrt{n}} \nabla l_n(\boldsymbol{\beta_0})&=\sqrt{n} \cdot \frac{1}{n}\lbrace  \nabla l_{c,l}( \boldsymbol{\beta}_0)+  \nabla
l_{c,u}( \boldsymbol{\beta}_0 )\rbrace \\
&=\sqrt{n} \cdot \frac{1}{n} \sum_{i=1}^{n}\lbrace  S^{(i)}\cdot \nabla l_{c,l}^{(i)}( \boldsymbol{\beta}_0)+ ( 1-S^{(i)})\cdot \nabla
l_{c,u}^{(i)}( \boldsymbol{\beta}_0 )\rbrace,
\end{aligned}$$

where $S^{(i)}$ represents whether the i-th sample belongs to the labeled data and follows the Bernoulli distribution with probability
$\lambda$. $\nabla l_{c,l}^{(i)}( \boldsymbol{\beta}_0)$ and $\nabla l_{c,u}^{(i)}( \boldsymbol{\beta_0})$ denote their corresponding i-th
samples, respectively.

Denote $\frac{1}{n} \sum_{i=1}^{n}  S^{(i)}\cdot \nabla l_{c,l}^{(i)}( \boldsymbol{\beta}_0)=W$, $\frac{1}{n} \sum_{i=1}^{n}( 1-S^{(i)})\cdot
\nabla l_{c,u}^{(i)}( \boldsymbol{\beta}_0 )=Q$. According to Anderson and Gill (1982)\cite{andersen1982cox} Theorem 3.2 and You et al.
(2018)\cite{you2018subtype} Theorem II, from the central limit theorem $\frac{1}{\sqrt{n}} \nabla l_n(\boldsymbol{\beta_0})$ converges to a
normal distribution, where the expectation and variance are

$$E(W+Q)=\lambda \cdot E( \nabla l_{c,l}( \boldsymbol{\beta}_0))+(1-\lambda)\cdot E(\nabla l_{c,u}( \boldsymbol{\beta}_0))=0,$$

$$\begin{aligned}Var(W+Q)
&=E(WW^{T})+E(QQ^{T})\\
&=\lambda I_l\left(\boldsymbol{\beta}_0\right)+(1-\lambda) I_u\left(\boldsymbol{\beta}_0\right),
\end{aligned}$$

By Slutsky's theorem, we can obtain $$\sqrt{n}\left(\boldsymbol{\hat\beta}_{k}-\boldsymbol{\beta}_{0k}\right) \xrightarrow{d} N\left(0,
\left(\lambda I_{u,k}(\boldsymbol{\beta}_{0 k})+(1-\lambda)I_{l,k}(\boldsymbol{\beta}_{0 k})\right)^{-1}\right).$$
\end{proof}

\section{Simulations}
\label{sec:dual4}
\subsection{Simulation Methods}

\subsubsection{Simulation Data Generation}

Suppose there are n independent samples, and $0.3\cdot n$ samples in the responders' group, $0.7\cdot n$ samples in the non-responders' group,
and the mixing probability is $\boldsymbol{\Pi}$=(0.3,0.7). The $\frac{n}{2}$ samples were randomly selected as the labeled data (experimental group), and
the rest were unlabeled data (control group). The covariate matrix $\boldsymbol{X}$ is a matrix of $n\times4$, where two covariates are independently
generated from the standard normal distribution $N(0,1)$, and the remaining two covariates are also independently generated from the binomial
distribution with probability 0.5. The $\boldsymbol{\beta}^T$ is the true value of the regression coefficients, and the true values of the
regression coefficients for the responders' and responders' groups are $\boldsymbol{\beta}_1$=$( -1, 0.5, 3, 0.8)^T$ and
$\boldsymbol{\beta}_2$=$( 2, -0.1, -3, 0.2)^T$.

Generating survival and censoring times are based on the method used by You et al. (2018)\cite{you2018subtype}. $T_i=\{-35*
\log(\boldsymbol{U})/\exp(\boldsymbol{X}\boldsymbol{\beta}_k)\}$ is the survival time of the i-th sample, where U follows from the uniform distribution $U(0,1)$. \hypertarget{c-definition}{}
The
uniform distribution $U(0,e^{c})$ generates the censoring time $C_i$, where c is a given constant. Thus, the observed survival time is
$t_i=min(C_i,T_i)$. For example, given c=9.5, 6.5, 3.8, conducting multiple experiments yields censoring rates around 5$\%$, 20$\%$, 45$\%$.

\begin{table}[H] %voc table result
    \centering
    \caption{The setting of the simulation}
    \begin{tabular}{*{5}{c}}
        \toprule
        Parameter                            & Responders    & Non-Responders   \\
        \midrule
        $\pi$      & 0.3     & 0.7   \\
        \midrule
        $\boldsymbol{\beta}$             & $( -1, 0.5, 3, 0.8)^T$    & $( 2, -0.1, -3, 0.2)^T$                      \\
   \midrule
        $n$             & $0.3\cdot n$    & $0.7\cdot n$                      \\

           \midrule
        $t_i$             & $t_i=min(C_i,\frac{-35*\log(U)}{\exp(\boldsymbol{X}\boldsymbol{\beta}_1)})$    & $
        t_i=min(C_i,\frac{-35*\log(U)}{\exp(\boldsymbol{X}\boldsymbol{\beta}_2)})$                      \\
        \bottomrule
    \end{tabular}
    \label{tab:siftflow}
\end{table}

\subsubsection{Model Evaluation Metrics}

Suppose we have $n_u$ unlabeled samples to classify, where the number of correct predictions is $s$, then the classification accuracy can be
expressed by

$$\text{Accuracy} = \frac{s}{n_u}.$$

For the evaluation of estimators, we consider bias and relative bias:

$$\text{Bias}(\hat{\beta}_j) = \mathbb{E}(\hat{\beta}_j) - \beta_j   ,$$  where $\mathbb{E}$ denotes the expectation, and in the experiment we
use the mean of $n_k$ experiments as the estimate of $\mathbb{E}(\hat{\beta}_j)$. $\hat{\beta}_j$ denotes the estimated coefficient of the
$j$-th baseline covariate, and $\beta_j$ denotes the true coefficient.

$$\text{Relative bias} = \frac{E(\hat{\boldsymbol{\beta}}{_j}) - \boldsymbol{\beta}_j}{\boldsymbol{\beta}_j} \times 100\%.$$

For the stability of the prediction accuracy, consider calculating the standard deviation of the results of $n_k$ experiments

$$\text{SD}(\text{Accuracy})=\sqrt{\frac{\sum(\text{Accuracy} -mean(\text{Accuracy}))^2}{n_k-1}},$$where $n_k$ represents the number of
experiments.

Similarly, for the estimate of $\boldsymbol{\beta}$, we calculate its standard deviation.

At last, for the estimate of mixing probability $\boldsymbol{\Pi}$, we calculate its bias and standard deviation.

\subsubsection{Model Evaluation Methods}

To verify the reliability of the Dual Cox model theory, the following points of the theory were considered: classification accuracy and its standard
deviation, unbiasedness and standard deviation of the estimators, the convergence of the parameter estimators, and the effect of sample size
and censoring rate on the model performance.

\begin{itemize}
\item 1. For the accuracy and stability of classification prediction: The average classification accuracy and standard deviation were
    calculated by repeating the simulation experiment 1000 times.
\item 2. Unbiasedness and robustness of model parameter estimators: bias, relative bias, and standard deviation of the parameters were
    calculated by repeating the simulation experiment 1000 times.
\item 3. Convergence of the fits: whether the model parameter estimates converge to the true values and whether the likelihood function
    converges to the global maximum.

\item 4. Effect of censoring rate: Three different censoring rates (5$\%$,20$\%$, 45$\%$) were considered to observe their impacts over the fitting.

\item 5. Effect of sample size: Given n=1000, 700, and 300 to observe how different sample sizes affect the fitting.\
\end{itemize}

The selection of initial values will be discussed in Section 4.1.3, which is based on 1000 replicate experiments with n=1000, \hyperlink{c-definition}
c=6.5, resulting in a
censoring rate of 18.1$\%$. For the simulation experiments in Section 4.2 with 1000 replications, we are aiming to investigate the
potential effects of sample size and censoring rate on model fitting performance. In those experiments, we select values of 1000, 700,
and 300 for the sample size n, and set \hyperlink{c-definition}
c at 9.5, 6.5, and 3.8, respestively, for the censoring rate to investigate its effect around
5$\%, 20\%, 45\%$ respectively. In total, we conducted 3x3=9 experiments.

\subsubsection{Initial Values}

Although the theoretical properties of Section \ref{sec:dual3} have illustrated the reasonability of the $\boldsymbol{\beta}$ estimate, the
strict concavity of the likelihood function is assumed by the regularity conditions. Moreover, as mentioned in Section \ref{sec:dual},
subsection 2, the EM algorithm may cause the likelihood function to converge to a local maximum rather than a global maximum. Therefore, the
following will consider three types of initial values of $u_{ik}$ and investigate numerically whether they converge to a local maximum.

\begin{itemize}
\item [Method 1.] Initial values are at bounds: For each unlabeled data $u_{i1}^{(0)}$ follows Bernoulli distribution with probability 0.5,
    $u_{i2}^{(0)}=1-u_{i1}^{(0)}$, and the obtained values are 0/1 sequences.
\item [Method 2.] Random assignment of initial values: For each unlabeled data $u_{i1}^{(0)}$ following a uniform distribution $U(0,1)$,
    $u_{i2}^{(0)}=1-u_{i1}^{(0)}. $
\item [Method 3.] $\boldsymbol{\Pi^{(0)}}$ as $u_{ik}$ a priori information: the initial value $\boldsymbol{\Pi^{(0)}}$ is derived from the maximum likelihood
    estimation of the labeled data, and let $u_{i1}^{(0)}$=$\pi_1^{(0)}$, $u_{i2}^{(0)}$=$\pi_2^{(0)}. $
\end{itemize}

In order to compare the effects of the three initial values on the model fitting results under the same survival data, experiments with the
above three assigned initial values are conducted with a sample size of n=1000 and a censoring rate of 18.1$\%$, and the survival data is
generated by the method described in Section \ref{sec:dual4},. Specifically, the experiments are conducted 1000 times using each of the three
methods of assigning initial values, and the convergence results are compared (Note: Method 3 is fitted with the same results for 1000
repetitions because the initial values are the same).

\subsection{Experimental Results}

In this section, a total of 9 experiments were conducted, each with 1000 replications, to investigate the effects of different sample sizes and censoring rates on the precision and stability of the Dual Cox model
fitting algorithm.

\subsubsection{Effect of Censoring Rate}

From the tables \ref{tab:4.3}, \ref{tab:4.4}, \ref{tab:4.5}, with the same sample size and different censoring rates, we can see that for the
relative bias of the estimated $\boldsymbol{\beta}$, the classification accuracy and the estimated mixing probability, the convergence results
of the Dual Cox model algorithm are close to the true values, and the model algorithm fitting results can all reach a certain level of accuracy, and
the consistency of the estimators, which is proved in Section \ref{sec:dual3}, are verified. However, as the censoring rate increases, the
standard deviation increases significantly, and the standard deviation is greatest when the censoring rate is high. As shown in Tables
\ref{tab:4.3}, \ref{tab:4.4}, and \ref{tab:4.5}, for n=1000, the classification accuracy is 0.89 for low censoring rate, the estimated mixing
probability is 0.31, and the estimated $\boldsymbol{\beta}$ ($\hat{\beta}_{11 }$ to $\hat{\beta}_{24}$) have relative bias between 2$\%$ and
7$\%$. For moderate censoring, the classification accuracy is 0.89, the estimated mixing probability is 0.31, and the relative bias of
estimated $\boldsymbol{\beta}$ range from 1$\%$ to 8$\%$, while for high censoring, the classification accuracy is 0.87, the estimated mixing
probability is 0.31, and the relative bias of estimated $\boldsymbol{\beta}$ range from 0$\%$ to 11$\%$. Thus, convergence results of
$\boldsymbol{\beta}$ do not change significantly with different censoring rates, and they are close to the true value. The results can reach a
certain precision. For the standard deviation, it can be found that the standard deviation of all estimated $\boldsymbol{\beta}$ are lower at
low censoring rates than at medium censoring rates. The standard deviation of all estimated $\boldsymbol{\beta}$ is lower at medium censoring
rates than at high censoring rates, which shows that the stability in parameter estimation decreases as the censoring rate increases. This is
in line with the conventional wisdom that the performance of the model algorithm fit decreases as the data becomes more incomplete. In
addition, the standard deviations obtained for all three censoring rates are relatively small at n=1000, and there is no significant
difference, which indicates that the sample size is sufficient to maintain the stability of the Dual Cox model algorithm fitting results even
at high censoring rates.

\subsubsection{Effect of Sample Size}

From the tables \ref{tab:4.3}, \ref{tab:4.4}, \ref{tab:4.5}, we can see that the convergence results of the Dual Cox model algorithm are close
to the true values for the same censoring rate and different sample sizes, for the relative bias of the estimated $\boldsymbol{\beta}$, the
classification accuracy and the mixing probability estimates. All of them are close to the true values, the models can reach a certain
precision, and the consistency of the estimators is verified. As the sample size decreases, the standard deviation increases significantly.
For the case of medium censoring rate and n=1000 of table \ref{tab:4.3}, the classification accuracy is 0.89, the estimated mixing probability
is 0.31, and the relative bias of the estimated $\boldsymbol{\beta}$ is between 2$\%$ and 7$\%$. For n=700, the classification accuracy is
0.89, and the estimated mixing probability is 0.31. The relative bias of the estimated $\boldsymbol{\beta}$ range from 0$\%$ to 7$\%$. For
n=300, the classification accuracy is 0.88, the estimated mixing probability is 0.31, and the relative bias of the estimated
$\boldsymbol{\beta}$  range from 3 $\%$ to 8$\%$. It can be found that there is no significant change in the convergence results of the
estimates in the three sample sizes, which are all close to the true values, and the consistency of the estimators is verified. For the
standard deviation, it can be found that for n=1000, the standard deviation of all estimated $\boldsymbol{\beta}$ is lower than n=700, and
when n=700 the standard deviation of all estimated $\boldsymbol{\beta}$ is lower than n=300, and the standard deviation is largest for n=300.
Besides, for most parameters n=300 has twice the standard deviation of n=1000, i.e., n=300 has the worst stability. It can be seen that the
stability of the parameters increases as the sample size increases. This is consistent with the theoretical proof in Section \ref{sec:dual3}
and the conventional knowledge that the variance decreases when n increases.

\subsubsection{Convergence of the Fit}

The above discussions on sample size and censoring rate both illustrate that the parameter estimates are close to the true values, the
classification accuracy is high, and the fitting algorithm possesses good convergence. In addition, by discussing the selection of initial
values in 4.1.3, it is also shown that the model can converge to a relatively high likelihood function by randomly selecting any initial
values except the boundary values, which shows that the model is not so sensitive to the initial values. This may be because of our
semi-supervised case, the condition of fixed K=2 in the mixture model, and the fact that the proportion of clinical experimentally labeled
data (experimental group) to the total sample is not too low, which can make our model much more stable.

When setting up the algorithm, the log-likelihood function must converge, but the convergence of the log-likelihood function does not mean
that the parameters converge. In the following, we will show whether $\pi$ converges to the true value as the algorithm iterates, whether
$\boldsymbol{\beta}$ converges to the true value as theoretically proven, and how the classification accuracy changes during the iterations.

The sample size is set to n=1000, the censoring rate is 18.1$\%$, and the convergence condition of the log-likelihood function is that the
difference between two iterations is less than $\epsilon$=$10^{-5}$. As shown in Figure \ref{fig:4.4}, the likelihood function converges, and
the classification accuracy and mixing probability also stabilize. For the regression coefficients, it can also be found that the regression
coefficients are estimated closer and closer to the true value of the model as the number of iterations increases. It is noteworthy that the
classification accuracy is highest at the seventh iteration because the log-likelihood function has reached a relatively high value and the
mixing probability is close to the true value of 0.3, which gives rise to this phenomenon.

where the criterion for determining the convergence of $\hat{\boldsymbol{\beta}}$: $\hat{\boldsymbol{\beta}}$ is considered to have converged
if $\Delta \beta$ is less than some threshold $\epsilon'$: $$ \Delta \beta = ||\hat{\boldsymbol{\beta}} - \boldsymbol{\beta}||_2 <\epsilon',$$
where $||\cdot||_2$ denotes the $L_2$ norm and t represents the t-th iteration.

\begin{table}[H] %voc table result
 \centering
    \caption{The experiment results with around 5$\%$ censoring rates across varied sample sizes}
\begin{tabular}{*{5}{c}}
\hline $\mathrm{n}=1000$  & Mean      & SD     & BIAS &   Relative Bias   \\
\hline$\hat{\beta}_{11}$ & -1.02 & 0.17 & -0.02 & 0.02 \\
$\hat{\beta}_{12}$ & 0.53 & 0.15 & 0.03 & 0.06 \\
$\hat{\beta}_{13}$ & 3.12 & 0.19 & 0.12 & 0.04 \\
$\hat{\beta}_{14}$ & 0.83 & 0.08 & 0.03 & 0.04 \\
$\hat{\beta}_{21}$ & 2.14 & 0.12 & 0.14 & 0.07 \\
$\hat{\beta}_{22}$ & -0.11 & 0.09 & -0.01 & 0.06 \\
$\hat{\beta}_{23}$ & -3.18 & 0.12 & -0.18 & 0.06 \\
$\hat{\beta}_{24}$ & 0.21 & 0.01 & 0.05 & 0.06 \\
$\hat{\pi}_1$ & 0.31 & 0.01 & 0.01 & \\
 Classification Accuracy & 0.89 & 0.01 & & \\
\hline

\hline $\mathrm{n}=700$  \\
\hline$\hat{\beta}_{11}$ & -1.04 & 0.20 & -0.04 & 0.04 \\
$\hat{\beta}_{12}$ & 0.53 & 0.17 & 0.03 & 0.05 \\
$\hat{\beta}_{13}$ & 3.13 & 0.24 & 0.13 & 0.04 \\
$\hat{\beta}_{14}$ & 0.84 & 0.11 & 0.04 & 0.05 \\
$\hat{\beta}_{21}$ & 2.15 & 0.13 & 0.15 & 0.07 \\
$\hat{\beta}_{22}$ & -0.10 & 0.11 & 0.00 & 0 \\
$\hat{\beta}_{23}$ & -3.19 & 0.14 & -0.19 & 0.06 \\
$\hat{\beta}_{24}$ & 0.21 & 0.05 & 0.01 & 0.05 \\
$\hat{\pi}_1$ & 0.31 & 0.01 & 0.01 & \\
 Classification Accuracy & 0.89 & 0.02 & & \\
\hline

\hline $\mathrm{n}=300$   \\
\hline$\hat{\beta}_{11}$ & -1.04 & 0.33 & -0.04 & 0.04 \\
$\hat{\beta}_{12}$ & 0.52 & 0.32 & 0.02 & 0.03 \\
$\hat{\beta}_{13}$ & 3.22 & 0.39 & 0.22 & 0.07 \\
$\hat{\beta}_{14}$ & 0.85 & 0.18 & 0.05 & 0.07 \\
$\hat{\beta}_{21}$ & 2.17 & 0.21 & 0.17 & 0.08 \\
$\hat{\beta}_{22}$ & -0.11 & 0.17 & -0.01 & 0.07 \\
$\hat{\beta}_{23}$ & -3.23 & 0.23 & -0.23 & 0.07 \\
$\hat{\beta}_{24}$ & 0.21 & 0.08 & 0.01 & 0.06 \\
$\hat{\pi}_1$ & 0.31 & 0.01 & 0.01 & \\
 Classification Accuracy & 0.89 & 0.03 & & \\
\hline

\end{tabular}
    \label{tab:4.3}

\end{table}

\begin{table}[H] %voc table result
 \centering
    \caption{The experiment results with around 20$\%$ censoring rates across varied sample sizes}
\begin{tabular}{*{5}{c}}
\hline $\mathrm{n}=1000$  & Mean      & SD     & BIAS &   Relative Bias   \\
\hline$\hat{\beta}_{11}$ & -0.99 & 0.19 & 0.01 & -0.01 \\
$\hat{\beta}_{12}$ & 0.53 & 0.17 & 0.03 & 0.06 \\
$\hat{\beta}_{13}$ & 3.12 & 0.21 & 0.12 & 0.04 \\
$\hat{\beta}_{14}$ & 0.84 & 0.09 & 0.04 & 0.05 \\
$\hat{\beta}_{21}$ & 2.16 & 0.12 & 0.16 & 0.08 \\
$\hat{\beta}_{22}$ & -0.11 & 0.10 & -0.01 & 0.06 \\
$\hat{\beta}_{23}$ & -3.20 & 0.12 & -0.20 & 0.07 \\
$\hat{\beta}_{24}$ & 0.21 & 0.05 & 0.01 & 0.06 \\
$\hat{\pi}_1$ & 0.31 & 0.01 & 0.01 & \\
 Classification Accuracy & 0.89 & 0.01 & & \\
\hline

\hline $\mathrm{n}=700$  \\
\hline$\hat{\beta}_{11}$ & -1.01 & 0.23 & -0.01 & 0.01 \\
$\hat{\beta}_{12}$ & 0.53 & 0.20 & 0.03 & 0.05 \\
$\hat{\beta}_{13}$ & 3.13 & 0.28 & 0.13 & 0.04 \\
$\hat{\beta}_{14}$ & 0.84 & 0.12 & 0.04 & 0.05 \\
$\hat{\beta}_{21}$ & 2.17 & 0.14 & 0.17 & 0.08 \\
$\hat{\beta}_{22}$ & -0.10 & 0.12 & 0.00 & 0 \\
$\hat{\beta}_{23}$ & -3.21 & 0.15 & -0.21 & 0.07 \\
$\hat{\beta}_{24}$ & 0.21 & 0.06 & 0.01 & 0.06 \\
$\hat{\pi_1}$ & 0.31 & 0.01 & 0.01 & \\
 Classification Accuracy & 0.89 & 0.02 & & \\
\hline

\hline $\mathrm{n}=300$  \\
\hline$\hat{\beta}_{12}$ & -1.01 & 0.4 & -0.01 & 0.01 \\
$\hat{\beta}_{12}$ & 0.51 & 0.37 & 0.01 & 0.02 \\
$\hat{\beta}_{13}$ & 3.22 & 0.45 & 0.22 & 0.07 \\
$\hat{\beta}_{14}$ & 0.86 & 0.21 & 0.06 & 0.07 \\
$\hat{\beta}_{21}$ & 2.19 & 0.23 & 0.19 & 0.10 \\
$\hat{\beta}_{22}$ & -0.10 & 0.18 & -0.01 & 0.05 \\
$\hat{\beta}_{23}$ & -3.26 & 0.24 & -0.26 & 0.09 \\
$\hat{\beta}_{24}$ & 0.21 & 0.09 & 0.01 & 0.07 \\
$\hat{\pi}_1$ & 0.31 & 0.01 & 0.01 & \\
 Classification Accuracy & 0.88 & 0.03 & & \\
\hline

\end{tabular}
    \label{tab:4.4}

\end{table}

\begin{table}[H] %voc table result
 \centering
    \caption{The experiment results with around 45$\%$ censoring rates across varied sample sizes}
\begin{tabular}{*{5}{c}}
\hline $\mathrm{n}=1000$  & Mean      & SD     & BIAS &   Relative Bias   \\
\hline$\hat{\beta}_{11}$ & -0.95 & 0.25 & 0.05 & -0.05 \\
$\hat{\beta}_{12}$ & 0.51 & 0.22 & 0.01 & 0.02 \\
$\hat{\beta}_{13}$ & 3.08 & 0.27 & 0.08 & 0.03 \\
$\hat{\beta}_{14}$ & 0.80 & 0.13 & 0 & 0 \\
$\hat{\beta}_{21}$ & 2.18 & 0.15 & 0.18 & 0.09 \\
$\hat{\beta}_{22}$ & -0.11 & 0.12 & -0.01 & 0.11 \\
$\hat{\beta}_{23}$ & -3.20 & 0.14 & -0.20 & 0.07 \\
$\hat{\beta}_{24}$ & 0.20 & 0.06 & 0 & 0.01 \\
$\hat{\pi}_1$ & 0.31 & 0.01 & 0.01 & \\
Classification Accuracy & 0.87 & 0.02 & & \\
\hline

\hline $\mathrm{n}=700$   \\
\hline$\hat{\beta}_{11}$ & -0.97 & 0.30 & -0.03 & -0.03 \\
$\hat{\beta}_{12}$ & 0.51 & 0.27 & 0.01 & 0.02 \\
$\hat{\beta}_{13}$ & 3.10 & 0.35 & 0.10 & 0.03 \\
$\hat{\beta}_{14}$ & 0.80 & 0.16 & 0 & 0 \\
$\hat{\beta}_{21}$ & 2.19 & 0.17 & 0.19 & 0.09 \\
$\hat{\beta}_{22}$ & -0.10 & 0.14 & 0 & 0.02 \\
$\hat{\beta}_{23}$ & -3.21 & 0.16 & -0.21 & 0.07 \\
$\hat{\beta}_{24}$ & 0.20 & 0.07 & 0 & 0.02 \\
$\hat{\pi}_1$ & 0.31 & 0.01 & 0.01 & \\
Classification Accuracy & 0.87 & 0.02 & & \\
\hline

\hline $\mathrm{n}=300$   \\
\hline$\hat{\beta}_{11}$ & -0.96 & 0.51 & 0.04 & -0.04 \\
$\hat{\beta}_{12}$ & 0.50 & 0.50 & 0 & 0 \\
$\hat{\beta}_{13}$ & 3.19 & 0.58 & 0.19 & 0.06 \\
$\hat{\beta}_{14}$ & 0.82 & 0.27 & 0.02 & 0.02 \\
$\hat{\beta}_{21}$ & 2.23 & 0.28 & 0.23 & 0.11 \\
$\hat{\beta}_{22}$ & -0.10 & 0.21 & -0 & 0.05 \\
$\hat{\beta}_{23}$ & -3.27 & 0.27 & -0.27 & 0.09 \\
$\hat{\beta}_{24}$ & 0.20 & 0.11 & 0 & 0.02 \\
$\hat{\pi}_1$ & 0.31 & 0.01 & 0.01 & \\
Classification Accuracy & 0.87 & 0.03 & & \\
\hline
\end{tabular}
    \label{tab:4.5}

\end{table}

\begin{figure}[H]
    \centering
    \includegraphics[width=.47\textwidth,height=.65\textwidth]{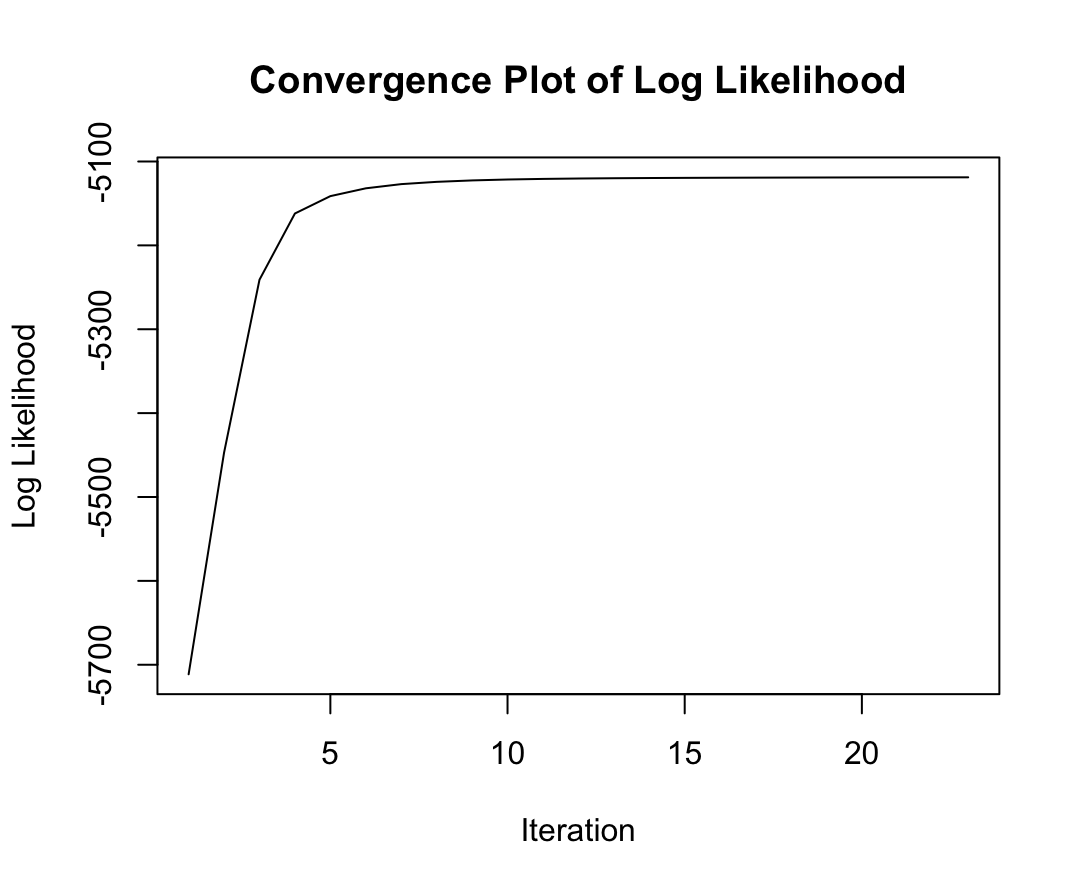}
    \includegraphics[width=.47\textwidth,height=.65\textwidth]{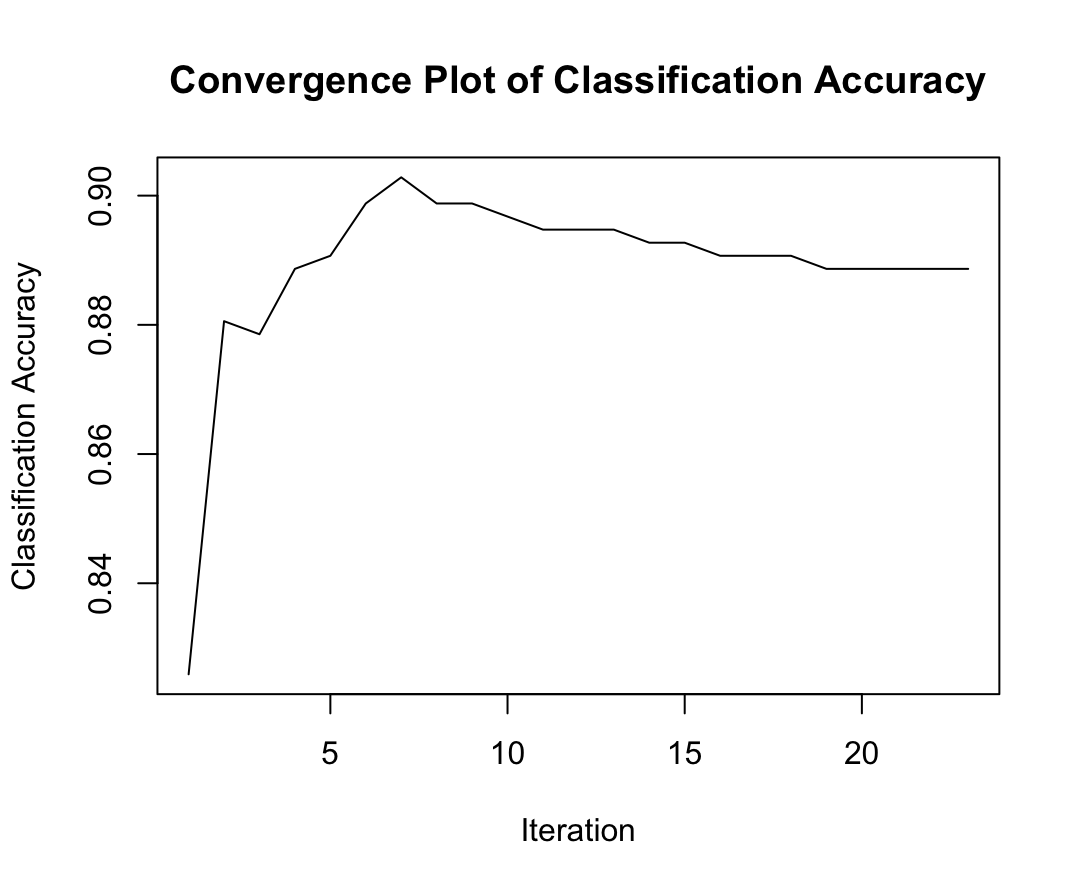}
    \includegraphics[width=.47\textwidth,height=.65\textwidth]{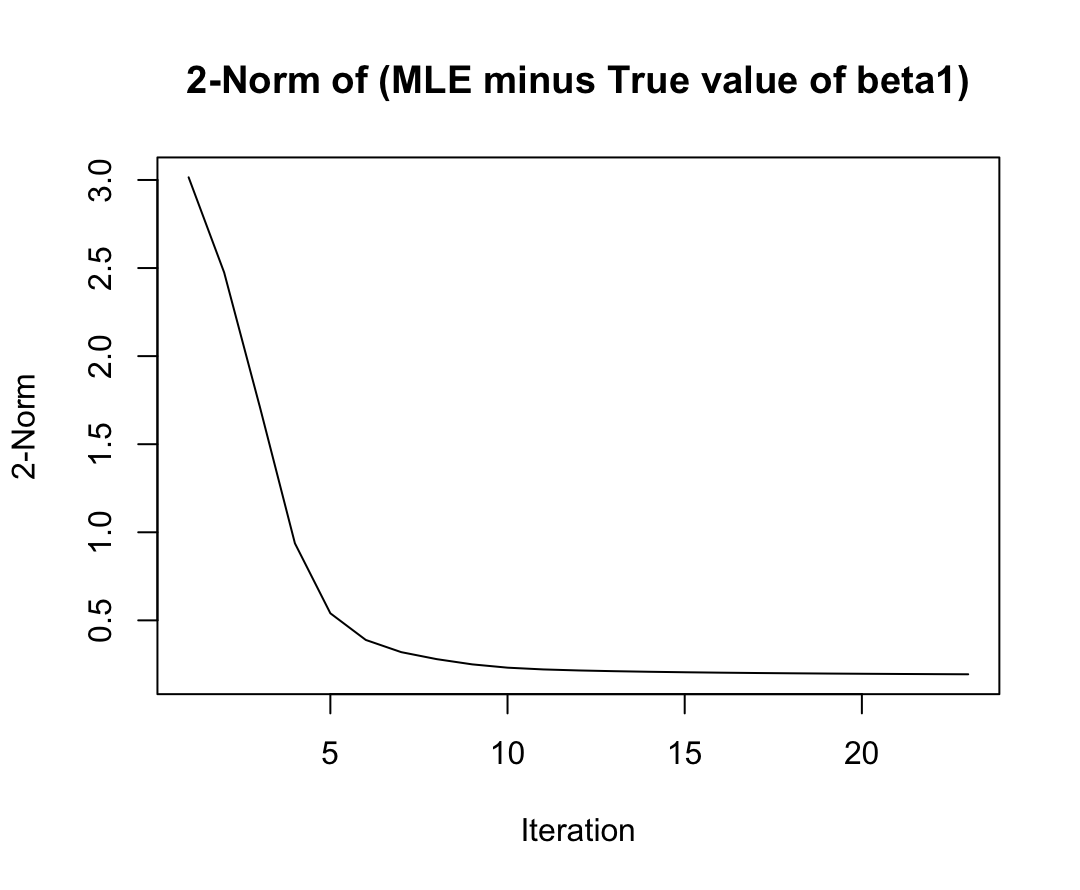}
    \includegraphics[width=.47\textwidth,height=.65\textwidth]{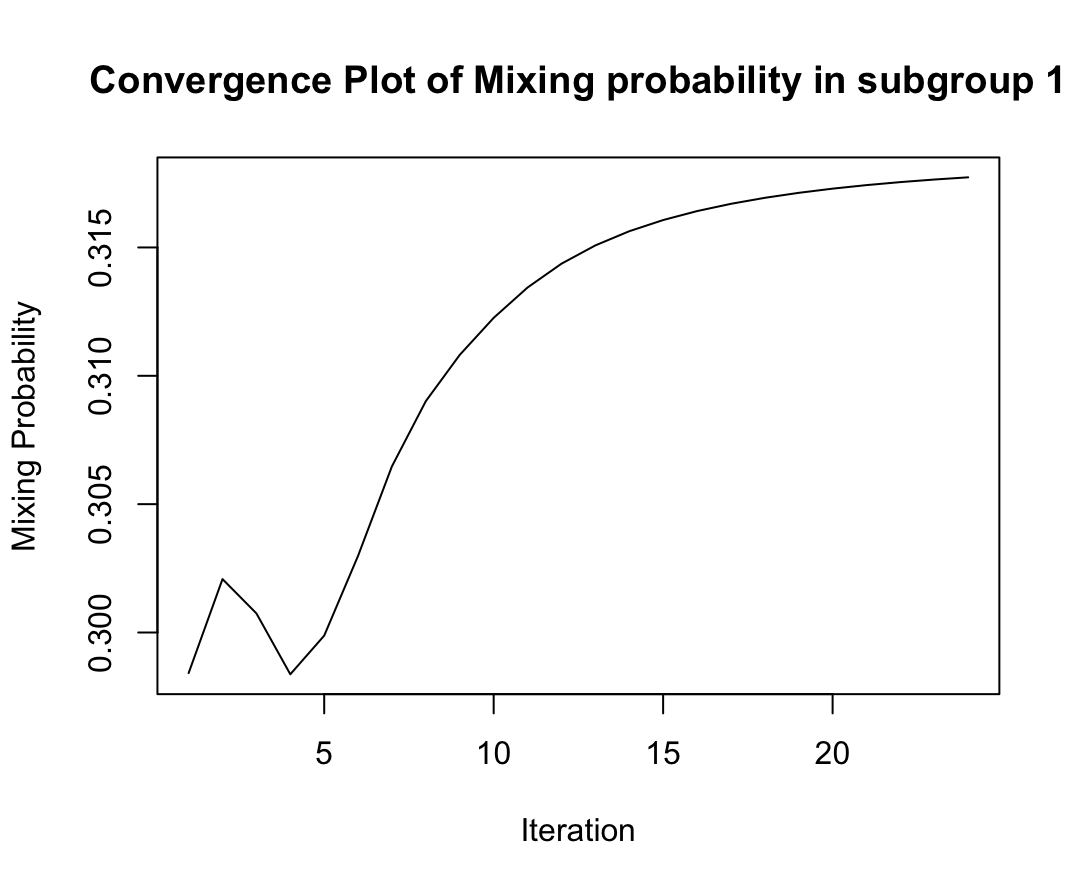}

    \caption{Convergence plot for the Dual Cox model's algorithm}
    \label{fig:4.4}
\end{figure}

\subsubsection{Initial Values}
According to the results in Figure \ref{fig:4.1}, it can be found that the observed
log-likelihood function is small when using Method 1 as the initial values method, and there is a large difference in the convergence results
for each of the 1000 experiments. In contrast, when using Method 2 as the initial values method, although the values obtained in each
experiment are slightly different, the difference between the values is not large, the convergence of the model fit is more stable, and the
observed log-likelihood function can converge to a higher value. And when using Method 3, the values of the observed log-likelihood function
obtained are higher than most of the experimental results of Method 2. This is because by using $\boldsymbol{\Pi^{(0)}}$ as $u_{ik}^{(0)}$ a priori, the
model can obtain as much information as possible from the labeled data at the very beginning.

From Table \ref{tab:4.1}, it can be seen that the classification accuracy using Method 1 is significantly lower, with an average of 0.64, and
the relative bias of the estimators is also higher, with a rate of change between 22.9$\%$ and 132$\%$. The standard deviation is also much
higher compared to Method 2. For Method 2 and Method 3, both methods can obtain high classification accuracy with a mean value of 0.87, and
the relative deviations of the estimated coefficients for both are also relatively small, ranging from 0.6$\%$ to 11.2$\%$ and 2.8$\%$ to
11.3$\%$, respectively. In terms of relative bias and classification accuracy, there is no significant difference between these two methods.
In addition, for Method 1, the number of iterations is only 4$\sim$6, while for Methods 2 and 3, 14$\sim$24 and 17 iterations can be obtained.
This indicates that the initial values selection of Method 1 can quickly fall into a local optimum.

Therefore, it is suggested that the choice of using $\boldsymbol{\Pi^{(0)}}$ as a priori in the selection of the initial value of $u_{ik}$, or using
multiple initial values and selecting from them the initial values that make the maximum of the likelihood function, are to try to avoid the
influence of local maxima.

\begin{figure}[H]
    \centering
    \includegraphics[width=1\textwidth,height=0.8\textwidth]{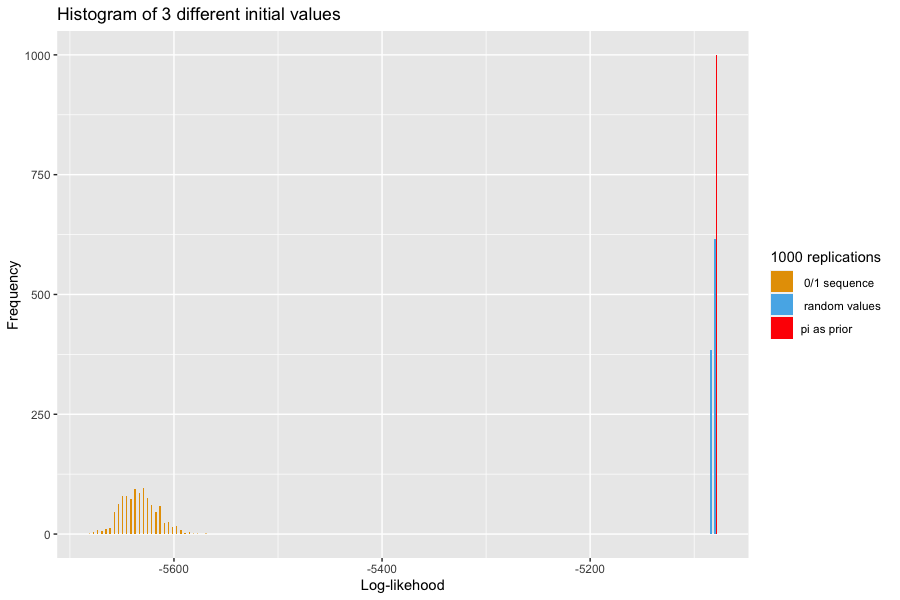}
    \caption{The log-likelihood function of the observed data obtained using different initial values of $\boldsymbol{U}$.}
    \label{fig:4.1}
\end{figure}

\begin{table}[H] %voc table result
    \centering
    \caption{The sample size is n = 1000, the censoring rate is 18.1$\%$, and the experiment is repeated 1000 times for each of the three initial values using our method}
    \begin{tabular}{*{5}{c}}
        \toprule
        Method 1                               & Mean      & SD     & BIAS &   Relative Bias   \\
        \midrule
        $\hat{\beta}_{11}$              & 0.32  & 0.093 & 1.32  & \textbf{-1.320}            \\
        $\hat{\beta}_{12}$      &  0.26  & 0.097 & -0.24 & \textbf{-0.485}             \\
        $\hat{\beta}_{13}$           &  0.44  & 0.082 & -2.56 & \textbf{-0.855}          \\
        $\hat{\beta}_{14}$ &0.14  & 0.044 & -0.66 & \textbf{-0.824}        \\
         $\hat{\beta}_{21}$              &1.54  & 0.050 & -0.46 & \textbf{-0.229}       \\
        $\hat{\beta}_{22}$      &  0.03  & 0.056 & 0.13  & \textbf{-1.259}              \\
        $\hat{\beta}_{23}$           & -1.70 & 0.058 & 1.31  & \textbf{-0.435}           \\
        $\hat{\beta}_{24}$ & 0.04  & 0.031 & -0.16 & \textbf{-0.777}          \\
        $\hat{\pi}_1$           & 0.315 & 0.010 & 0.015           \\
        Classification Accuracy &  \textbf{0.64}  & 0.020         \\

  \toprule
        Method 2                               & Mean      & SD     & BIAS &   Relative Bias   \\
        \midrule
        $\hat{\beta}_{11}$              & -0.96 & 0.018  & 0.04   & -0.042            \\
        $\hat{\beta}_{12}$      &  0.50  & 0.014  & 0      & 0.006             \\
        $\hat{\beta}_{13}$           &  3.04  & 0.021  & 0.04   & 0.013          \\
        $\hat{\beta}_{14}$ &0.76  & 0.010  & -0.04  & -0.049      \\
         $\hat{\beta}_{21}$              &2.20  & 0.010  & 0.20   & 0.101       \\
        $\hat{\beta}_{22}$      &  -0.11 & 0.006  & -0.009 & 0.091             \\
        $\hat{\beta}_{23}$           & -3.29 & 0.007 & -0.29  & 0.096         \\
        $\hat{\beta}_{24}$ & 0.22  & 0.003  & 0.022  & 0.112            \\
        $\hat{\pi}_1$           & 0.309 & 0.001  & 0.009  &           \\
         Classification Accuracy & 0.87  & 0.004          \\

        \midrule

       \toprule
        Method 3                              & Mean      & SD     & BIAS &   Relative Bias   \\
        \midrule
        $\hat{\beta}_{11}$              & -0.95 &  --  & 0.05  & -0.049           \\
        $\hat{\beta}_{12}$      &  0.51  &  --  & 0.01  & 0.028              \\
        $\hat{\beta}_{13}$           & 3.09  &  --  & 0.09  & -0.030           \\
        $\hat{\beta}_{14}$ &0.77  &  --  & -0.03 & -0.037        \\
         $\hat{\beta}_{21}$              &2.18  & --  & 0.18  & 0.092       \\
        $\hat{\beta}_{22}$      &  -0.11 &  --  & -0.01 & 0.091              \\
        $\hat{\beta}_{23}$           & -3.27 &  --  & -0.27 & 0.090           \\
        $\hat{\beta}_{24}$ & 0.22  &  --  & 0.02  & 0.113        \\
        $\hat{\pi}_1$           & 0.308 &  --  & 0.08             \\
         Classification Accuracy & 0.87  &   --         \\

        \midrule
    \end{tabular}
    \label{tab:4.1}
\end{table}

\section{ Real Data Analyses}

\subsection{PRIME Trial}

PRIME\cite{douillard_randomized_2010}\cite{douillard_panitumumabfolfox4_2013}\cite{douillard2014final}  is a randomized, multicenter Phase III
trial, The objective of this study is to evaluate the efficacy and safety of the targeted agent panitumumab (a monoclonal antibody against
epidermal growth factor receptor) in combination with chemotherapy agent FOLFOX4 (folinic acid, 5-fluorouracil, and oxaliplatin) in patients
with metastatic colorectal cancer (mCRC).

The trial was conducted between 2006 and 2013 and enrolled 1183 previously untreated mCRC patients. As shown in Figure \ref{fig:5.1}, 593
patients in the experimental group received panitumumab plus FOLFOX4, and 590 patients in the control group received FOLFOX4 alone. The trial
also included an analysis of the relationship between treatment outcomes in two subgroups of patients with wild-type KRAS tumors and patients
with KRAS-mutated tumors, respectively.KRAS is a gene that encodes a KRAS protein, which KRAS mutations are common in various types of cancer,
including colorectal cancer. KRAS mutations and wild-type refer to the status of the KRAS gene in tumor cells of patients with metastatic
colorectal cancer. KRAS mutations are genetic alterations that result in changes in the DNA sequence of the KRAS gene, resulting in the
production of abnormal KRAS proteins. In contrast, wild-type KRAS refers to the normal, unaltered form of the KRAS gene. The primary endpoint
of the trial is progression-free survival (PFS), with secondary endpoints including overall survival (OS), objective response rate (ORR),
duration of response (DOR), and safety.

As shown in Figures \ref{fig:5.2},\ref{fig:5.3} and tables \ref{tab:5.1} (Note: The public data set contains only $80\%$ of subjects, and the results
are slightly different from those in the original paper), the results of the test taking KRAS status into account show that, progression-free
survival is improved with the addition of panitumumab to chemotherapy versus chemotherapy alone in patients with wild-type KRAS tumors.
Specifically, median progression-free survival is 9.73 months ($95\%$ CI 9.33 to 11.37) in the group receiving panitumumab plus chemotherapy
versus 9.33 months (95$\%$ confidence interval (CI) 7.80 to 11.10) in the group receiving chemotherapy alone. The hazard ratio (HR) is 0.84
(95$\%$ CI 0.69-1.02, p=0.07) with a log-rank test p-value of 0.06. However, the addition of panitumumab to chemotherapy does not significantly
improve progression-free survival, as compared with chemotherapy alone, in patients with KRAS-mutated tumors. Median PFS is 7.47 months
(95$\%$ CI 6.40 to 8.47) in the group receiving panitumumab plus chemotherapy versus 9.1 months (95$\%$ CI 7.80 to 10.03) in the group receiving
chemotherapy alone, with a hazard ratio (HR) of 1.19 (95$\%$ CI 0.95 to 1.50, p=0.13), log-rank test p-value is 0.075. Thus, the difference
between KRAS mutant and wild-type in the PRIME trial is related to the difference in response to panitumumab plus chemotherapy between the two
groups. Patients with KRAS wild-type tumors have an improved risk with the addition of panitumumab, whereas patients with KRAS mutant tumors
have an increased risk.

In summary, this trial helps establish panitumumab as an effective treatment option for patients with metastatic colorectal cancer, especially
for patients with KRAS wild-type tumors, highlighting the potential benefit of targeted therapy in this patient population and that KRAS
status may be a useful biomarker for predicting response to panitumumab in patients with metastatic colorectal cancer. Next, we, therefore, will
consider 514 patients only with wild-type KRAS.

\begin{figure}[H]
    \centering
    \includegraphics[width=0.65\textwidth,height=0.4\textheight]{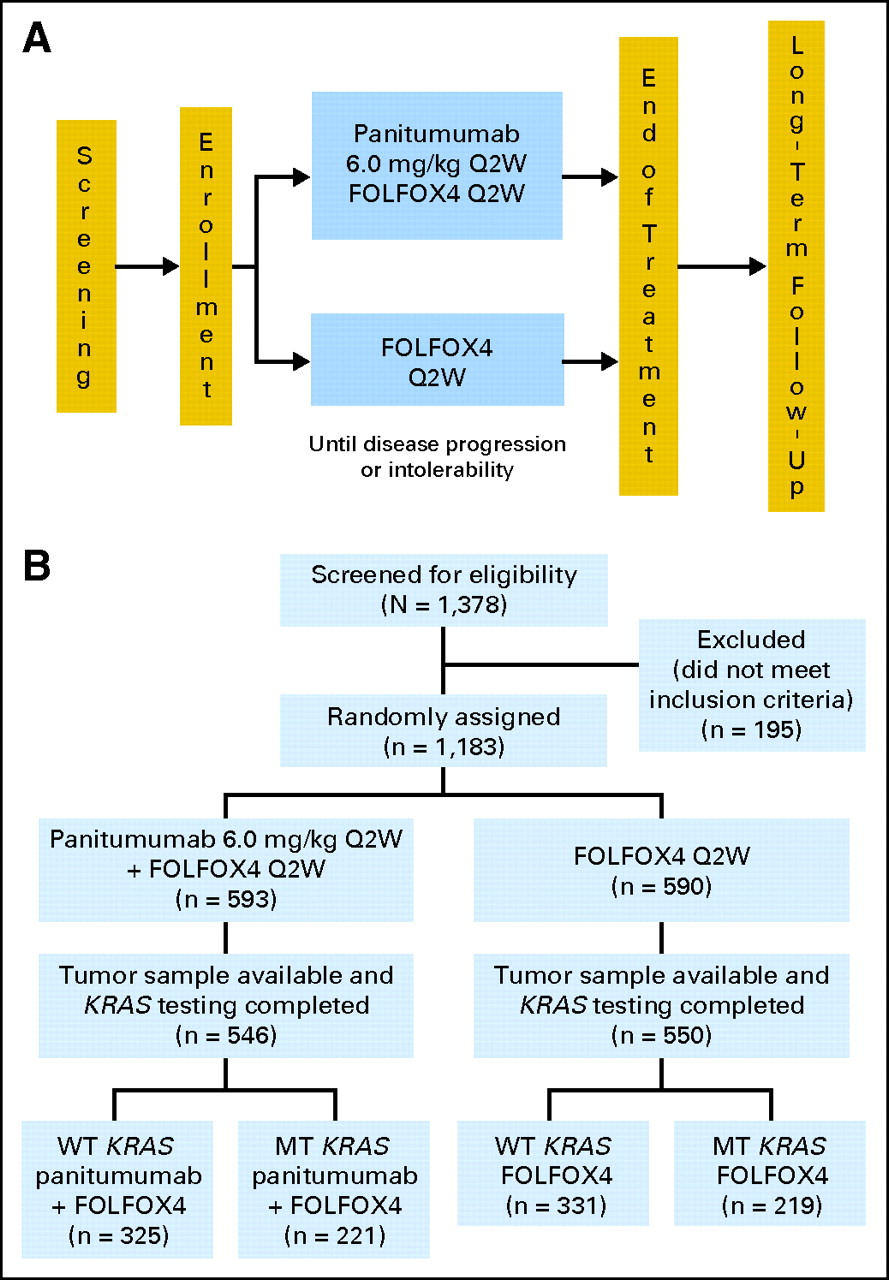}
    \caption{PRIME Trial\cite{douillard_randomized_2010}}
    \label{fig:5.1}
\end{figure}

\begin{figure}[H]
    \centering
    \includegraphics[width=1\textwidth,height=0.6\textwidth]{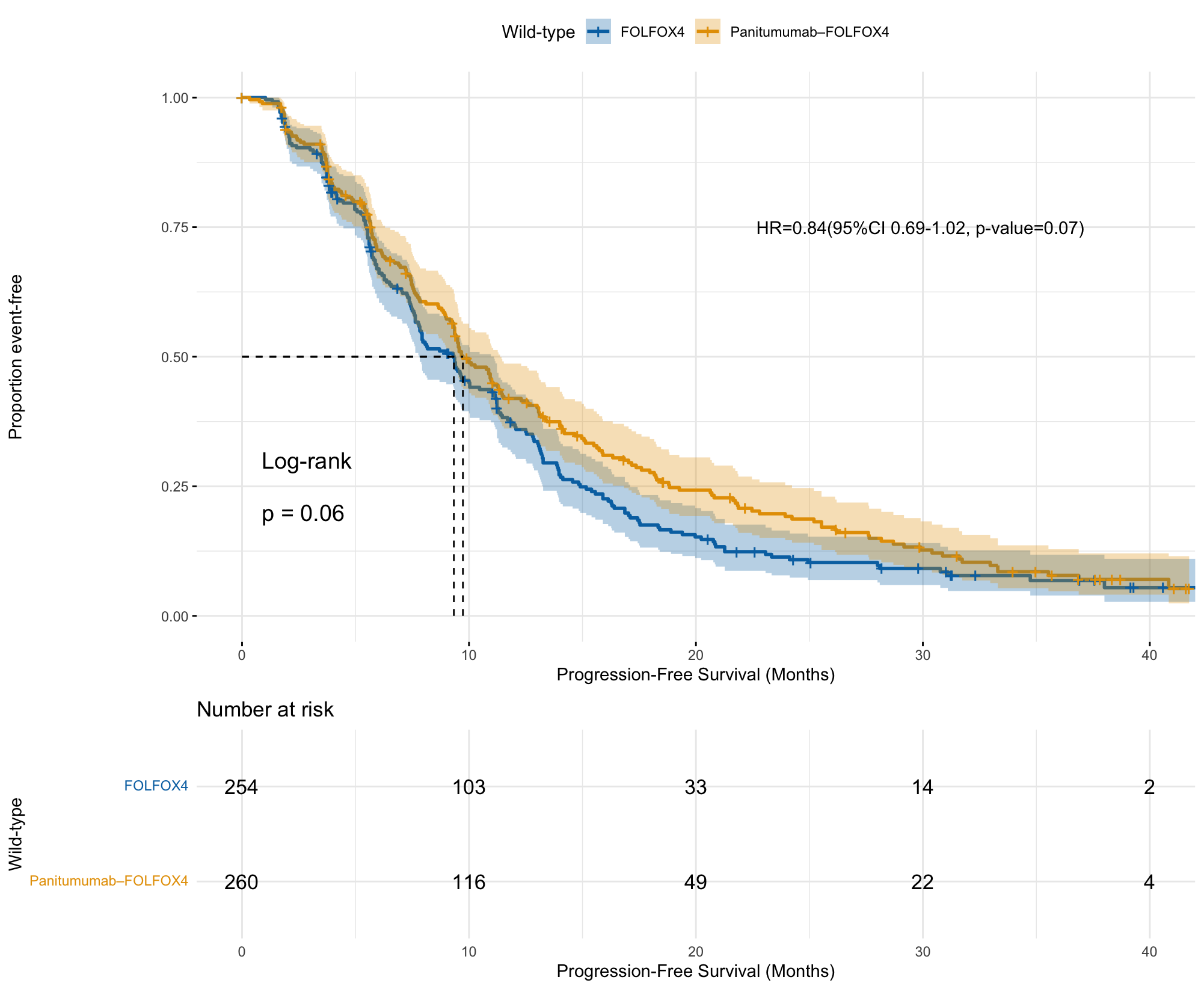}

    \caption{Kaplan-Meier curves, p-values from the log-rank test, hazard ratios with confidence intervals, and p-values for panitumumab plus
    FOLFOX4 and FOLFOX4 alone in patients with wild-type KRAS.}
    \label{fig:5.2}
\end{figure}

\begin{figure}[H]
    \centering
    \includegraphics[width=1\textwidth,height=0.6\textwidth]{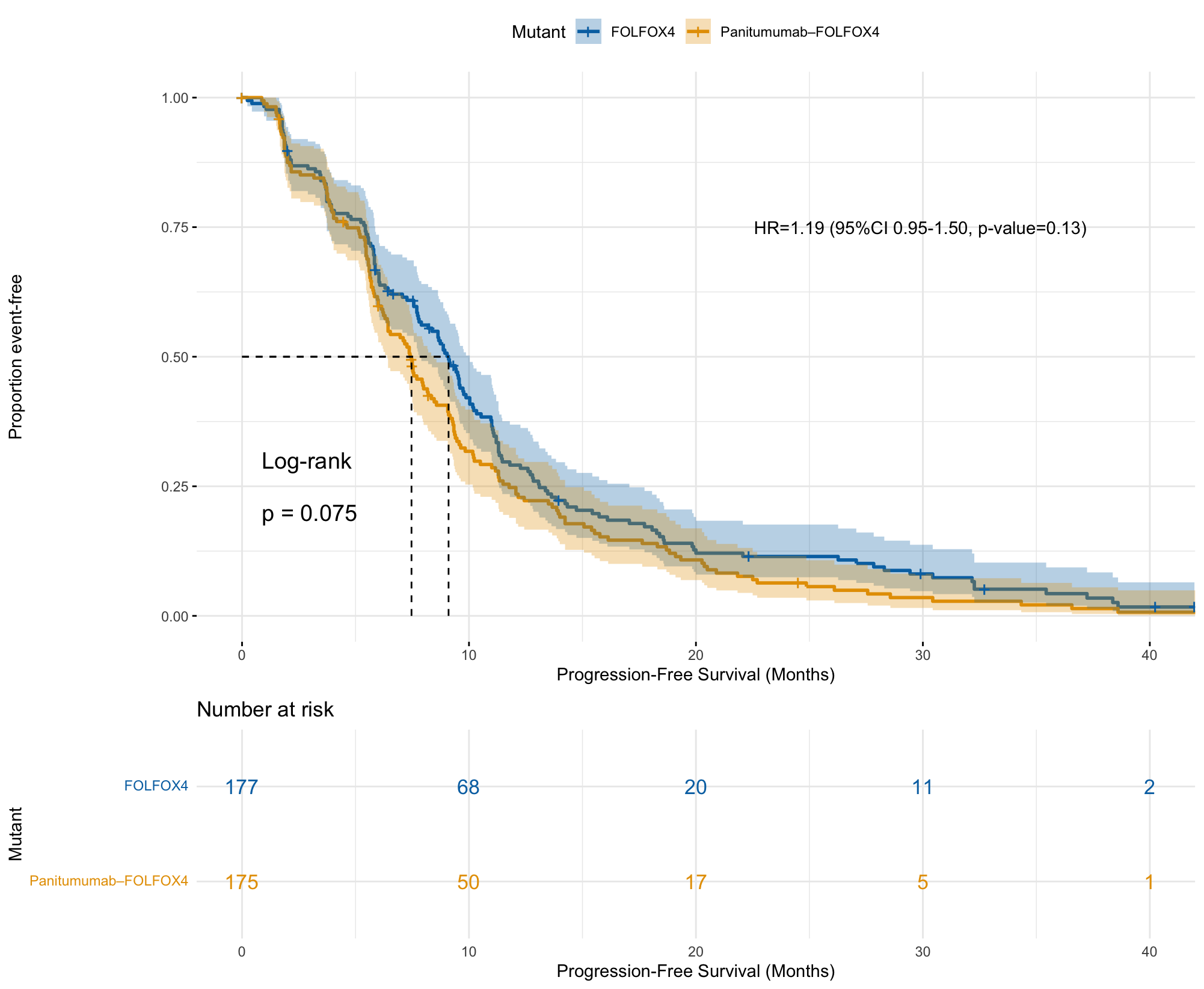}

    \caption{Kaplan-Meier curves, p-values from the log-rank test, hazard ratios with confidence intervals, and p-values for panitumumab plus
    FOLFOX4 and FOLFOX4 alone in patients with KRAS mutations.}
    \label{fig:5.3}
\end{figure}

\begin{table}[H]
\centering
\caption{Median progression-free survival and its 95$\%$ confidence interval with panitumumab plus FOLFOX4 and FOLFOX4 alone in patients with
wild-type KRAS}
\begin{tabular}{rlrrr}

  \hline
  Treatment & PFS & 95\% CI Lower & 95\% CI Upper \\
  \hline
FOLFOX4  & 9.33 & 7.80 & 11.10 \\
  Panitumumab–FOLFOX4  & 9.73 & 9.33 & 11.37 \\
   \hline
\end{tabular}
    \label{tab:5.1}
\end{table}

\begin{table}[H]
\centering
\caption{Median progression-free survival and its 95$\%$ confidence interval with panitumumab plus FOLFOX4 and FOLFOX4 alone in patients with
KRAS mutations}
\begin{tabular}{rlrrr}
  \hline
Treatment & PFS & 95\% CI Lower & 95\% CI Upper \\
  \hline
FOLFOX4 & 9.10 & 7.80 & 10.03 \\
Panitumumab–FOLFOX4 & 7.47 & 6.40 & 8.47 \\
   \hline
\end{tabular}
\label{tab:5.2}
\end{table}

 As shown in Table \ref{tab:5.3}, the names, meanings, and taking values of baseline covariates after preprocessing of 514 cases of wild-type
 KRAS are given. To assess the impact of various baseline covariates on patient risk, a total of 9 baseline covariates were included: ATRT,
 ECOG, $SITES\_ NUM \_ STRATA$, SEX, $AGE\_65$, $IS \_ WHITE$, DIAGTYPE, LIVERMET, LDH. The details are described below:
\begin{enumerate}
\item Treatment modality (ATRT): Panitumumab + FOLFOX4 in the experimental group and FOLFOX4 in the control group.
\item ECOG performance status (ECOG): The Eastern Cooperative Oncology Group (ECOG) performance status was used to assess the functional
    status of patients. Patients scoring 0-2 were included in the trial.ECOG performance status is a measure of a patient's functional
    status and ability to perform daily activities. \
\item Metastatic burden ($SITES\_ NUM \_STRATA$): The number of metastatic sites was analyzed as a potential prognostic factor. \
\item Gender (SEX): The patient's gender was recorded in the study entry. \
\item Age ($AGE\_65$): The patient's age was recorded at the time of study entry.\
\item Race ($IS \_ WHITE$): The patient's race was recorded at the time of study entry.

\item Primary tumor location (DIAGTYPE): The location of the primary tumor in the colon or rectum was recorded. \
\item Presence or absence of metastasis to the liver at study entry (LIVERMET): The presence of liver metastasis is an important baseline
    covariate because it is a known prognostic factor in metastatic colorectal cancer and can affect treatment outcome. \
\item Baseline LDH (LDH): As a marker of tumor burden and prognosis. High LDH levels are associated with poorer prognosis in patients with
    metastatic colorectal cancer. \
\end{enumerate}

\begin{table}[H] %voc table result
 \centering
    \caption{Name, meaning, take value, and percentage of baseline covariates after PRIME data preprocessing}
\begin{tabular}{|c|c|c| }
\hline Baseline covariates name & Meaning & Value ( percentage $\%$)   \\

\hline

 &  &1=Panitumumab + FOLFOX4\\
ATRT&experimental/control group  & (experimental group 50.58$\%$)   \\
& &0=FOLFOX4(control group)\\
\hline

  &  &1=Less than 50$\%$ of time in bed ($5.45\%$)   \\
ECOG &Baseline ECOG status & 0=Fully active or \\ & &Symptomatic but movable people\\
\hline

SITES$\_$NUM$\_$STRATA & Metastasis sites &1=metastasis sites$\geq$ 3(44.16$\%$),0=others\\
\hline

SEX & Gender &1=Male(64.01$\%$),0=Female   \\
\hline

$AGE\_65$ & Age &1=Age$\geq$ 65 (41.8$\%$),0=
others   \\
\hline
$IS \_ WHITE$ &White or not &1= White or Caucasian(93.39$\%$),0=others   \\
\hline
DIAGTYPE & Primary tumor type &1=Rectal(36.38$\%$),
0=Colon   \\
\hline
LIVERMET & initial metastasis to the liver &1=Yes($88.33\%$),0=NO \\
\hline
&  &1=LDH \\
LDH & Lactate dehydrogenase levels &$\geq$2ULN(29.18$\%$) \\
&  &0=LDH<2ULN \\
\hline
\end{tabular}
\label{tab:5.3}

\end{table}

\noindent{\textbullet}\quad\parbox[t]{0.9\textwidth}{\textbf{Model Algorithm Fitting}}

In the PRIME trial, the response of patients in the experimental group to the drugs was assessed according to the RECIST 1.1 criteria
\cite{eisenhauer2009new}, and in the experimental group, we can estimate an objective response rate at $61\%$. For the control group, we can not
observe the drug's response conditions in the control group because of the different mechanisms of action of the targeted drugs and
chemotherapy, so in this subsection, we fit the Dual Cox model algorithm to 514 wild-type KRAS patients to classify the control group and
explore the difference in the efficacy of panitumumab plus FOLFOX4 and FOLFOX4 alone on patients in different subgroups.

According to the Dual Cox model described in Section\ref{sec:dual}, subsection 2, we randomly choose 1000 different initial values to fit the
Dual Cox model, from which we select the initial values that make the maximum value of the log-likelihood function.

Finally, we obtain the form of the Dual Cox model as

$$S(t,\delta\mid\textbf{x})=\pi S_1(t,\delta\mid\textbf{x})+(1-\pi)S_2(t,\delta\mid\textbf{x}),$$

the hazard function in the response group
$$h_1(t \mid \boldsymbol{x})=h_{0 1}(t) \exp \left(\boldsymbol{\beta}_1^{\mathrm{T}} \boldsymbol{x}\right),$$the hazard function in the
non-response group
$$h_2(t \mid \boldsymbol{x})=h_{0 2}(t) \exp \left(\boldsymbol{\beta}_2^{\mathrm{T}} \boldsymbol{x}\right),$$
where\begin{align*}
  \boldsymbol{x} &= (ATRT, ECOG, SITES\_NUM\_STRATA, SEX, \\
  &\qquad AGE\_65, IS\_WHITE, DIAGTYPE, LIVERMET, LDH)^{\mathrm{T}},\\
    \boldsymbol{\Pi}&=(0.53,0.47),\\
    \boldsymbol{\beta}_1 &= (\num{-2.27}, \num{0.99}, \num{0.38}, \num{-0.18}, \num{0.34}, \num{0.06}, \num{-0.21}, \num{-0.54},
    \num{0.53})^{\mathrm{T}}, \\
\boldsymbol{\beta}_2 &= (\num{1.26}, \num{0.53}, \num{0.13}, \num{-0.41}, \num{0.06}, \num{0.30}, \num{-0.43}, \num{-0.29},
\num{0.53})^{\mathrm{T}}.
\end{align*}

\noindent{\textbullet}\quad\parbox[t]{0.9\textwidth}{\textbf{Subgroup Classification}}

The Dual Cox model divides the unobserved population into the response and non-response groups, and the results are shown in Table
\ref{tab:5.4}. Among them, the response group contains 302 individuals with a censoring rate of 16.6$\%$, and the non-response group contains
212 individuals with a censoring rate of 18.4$\%$. The mean values of individual convergence probabilities $\hat{u}_{ik}$ in different
populations are 0.95 and 0.98, respectively, which are close to 1, [[[indicating that the model can well converge individuals into the same
subgroup, and there is little difference between individuals within the group]]]. In addition, as shown in Figure \ref{fig:5.4}, the Dual Cox
model algorithm is able to group the non-censored data well, and the maximum posterior probability $\max(\hat{u}_{i1},\hat{u}_{i2})$ is 1 for
all non-censored data for the unobserved control population. However, for the censored data, the model has some uncertainty, but the certainty
of classification increases with increasing progression-free survival, indicating that our algorithm has difficulty in classifying patients
with censored, short progression-free survival, which is consistent with the conclusion of Eng and Hanlo\cite{eng_discrete_2014}’s discussion
about the unsupervised case. Based on the results of the subgroup classification, we can obtain an estimate of the objective response rate,
which is 0.51 in the control group, compared to an estimated response rate of 0.61 in the experimental group, and the convergence results were
in line with the real situation, that is, panitumumab plus FOLFOX4 can improve the objective response rate of patients in wild-type KRAS
patients.

\begin{table}[H]
\centering
\caption{The Dual Cox model algorithm classifies wild-type KRAS patients with unobserved response conditions to drugs}
\begin{tabular}{|c|c|c|c|c|}
\hline
& Total Patients & Censored Rate($\%$) & Mean of $\hat{u}_{ik}$  \\
\hline
Responders & 302 & 16.6$\%$ & 0.95 \\
\hline
Non-Responders & 212 & 18.4$\%$ & 0.98  \\
\hline
Overall Population & 514 & 17.3$\%$ & --\\
\hline
\end{tabular}
\label{tab:5.4}

\end{table}

\begin{figure}[H]
    \centering
    \includegraphics[width=1\textwidth,height=0.5\textwidth]{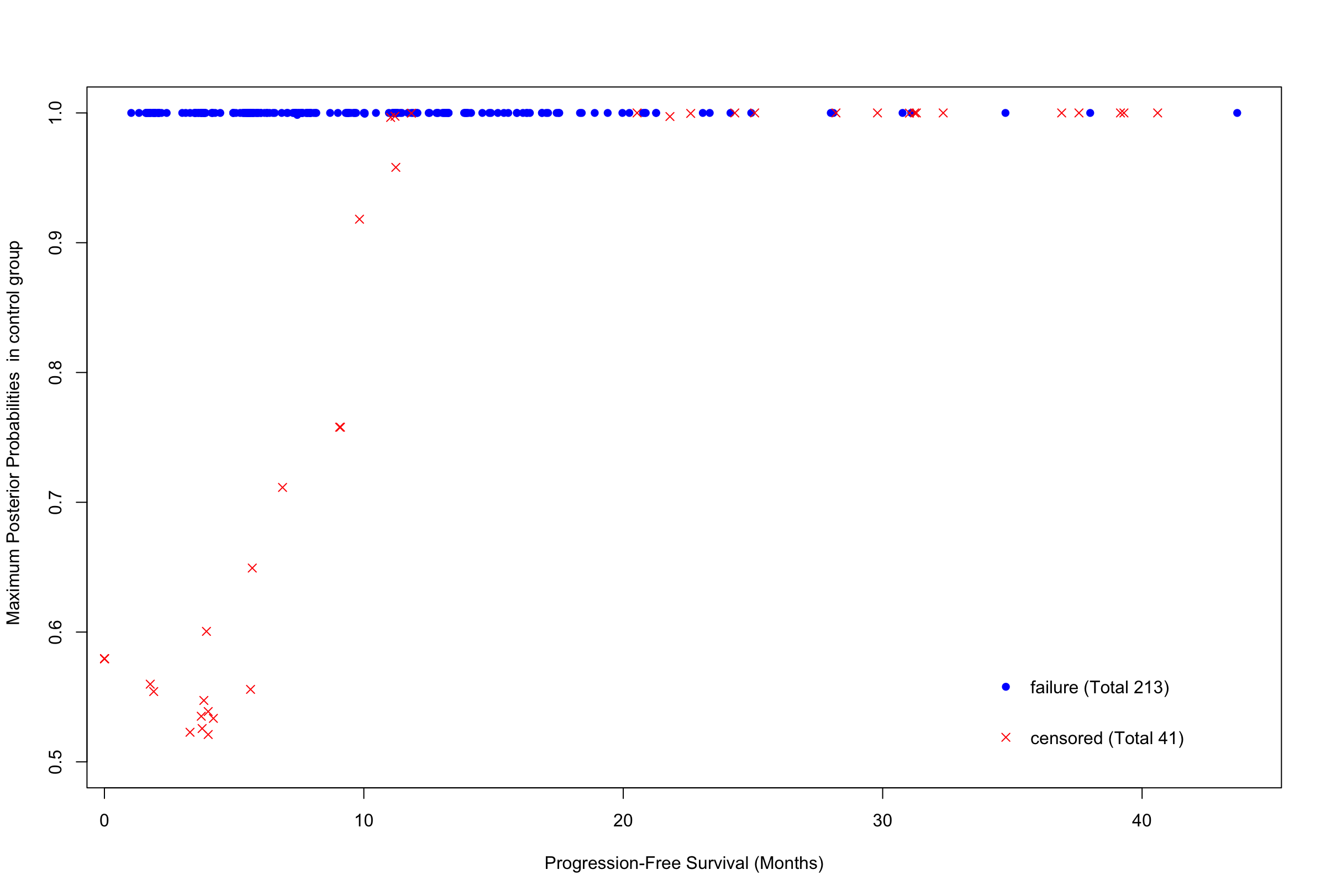}
    \caption{The maximum posterior probability $\max(\hat{u}_{i1},\hat{u}_{i2})$  fitted by the Dual Cox model algorithm for wild-type KRAS
    patients with unobserved response conditions.}
    \label{fig:5.4}
\end{figure}

\newpage

\noindent{\textbullet}\quad\parbox[t]{0.9\textwidth}{\textbf{Parameter Estimation}}

We fit the Cox proportional hazards model to the overall population and compare the results with the results of the Dual Cox Model. As shown
in Table \ref{tab:5.5} and Figure \ref{fig:5.5}, the estimated coefficients for the overall population, response, and non-response groups are
all of the same sign for the eight baseline covariates except for treatment efficacy ATRT, which illustrates the homogeneity of the response
and non-response groups for the eight baseline covariates except for ATRT and the heterogeneity for treatment efficacy ATRT.

For the overall population, ATRT is not significant at the level of 0.05 (p=0.07), and in the response group HR=0.1 with p-value <0.05, indicating
better efficacy in the experimental group than in the control group for patients who have responded. In contrast, HR=3.52 with p-value<0.05 in the
non-response group, indicating better efficacy in the control group than in the experimental group for non-response patients; also, the
estimated coefficients of baseline covariates ECOG and SITES$\_$NUM$\_$STRATA, which are significant in both the responder group and overall
population, show ECOG=1 (less than 50$\%$ of time in bed) and metastatic site$\geq$3 have a significantly higher risk, but do not raise risk
in the non-response group; also, for patients with age$\geq$65 in the response group, the risk is significantly higher for age$\geq$65,
indicating that age$\geq$65 has a higher risk for the response group, but not for the overall population, non-response group. For the overall
population, non-responders group, there is no significant increase in risk. For the covariate LDH, LDH=1 shows a significantly higher risk
than LDH=0 for the overall population, responders, and non-responders groups.

\begin{table}[H] %voc table result
 \centering
    \caption{Estimated coefficients, HR, and p-values are obtained from the Dual Cox model and the Cox proportional hazards model for
    wild-type KRAS patients}

        \begin{threeparttable}
\begin{tabular}{p{2cm}ccccccccc}
\hline & \multicolumn{3}{c}{ Responders } & \multicolumn{3}{c}{ Non-Responders } & \multicolumn{3}{c}{ Overall Population } \\
\hline Covariates & $\hat{\boldsymbol{\beta}_1}$ & HR &p-values & $\hat{\boldsymbol{\beta}_2}$ & HR &p-values  & $\hat{\boldsymbol{\beta}}$ &
HR & p-values \\
\hline
ATRT & -2.27 & 0.10& $\pmb{2\cdot 10^{-16}}$
 & 1.26 &3.52 & $\pmb{2\times10^{-11}}$ & -0.18&  0.84 & 0.07 \\
\hline
ECOG & 0.99 & 2.68& $\pmb{4\cdot 10^{-5}}$ & 0.53 & 1.69& 0.16 & 0.50 &1.64 & $\textbf{0.02}$\\
\hline
SITES &  0.38 & 1.50& $\textbf{0.004}$  & 0.13 &1.14 & 0.45 & 0.27 &1.31 & $\textbf{0.007}$\\
\hline
SEX & -0.18& 0.83 & 0.15 & -0.41 &0.66 & $\textbf{0.01}$ & -0.09 & 0.91& 0.39 \\
\hline
AGE\_65 & 0.34 &1.40 & $\textbf{0.01}$ & 0.06 & 1.06& 0.74 & 0.11 &1.11& 0.30  \\
\hline
IS\_WHITE & 0.06 & 1.06& 0.81 & 0.30 & 1.35& 0.31 & 0.08  &1.08&0.70\\
\hline
DIAGTYPE & -0.21 & 0.81& 0.11 & -0.43 & 0.65 & $\textbf{0.03}$ & -0.03 &0.97& 0.75\\
\hline
LIVERMET & -0.54 & 0.58 & $\textbf{0.004}$ & -0.29 & 0.75& 0.19 & -0.24 &0.78 & 0.12\\
\hline
LDH & 0.53 &1.70 &  $\pmb{5\cdot10^{-4}}$ & 0.53 &1.70 & $\textbf{0.004}$ & 0.34 &1.40&$\textbf{0.001}$\\
\hline
\end{tabular}
    \label{tab:5.5}
    \begin{tablenotes}
        \footnotesize
        \item[1] $SITES\_NUM\_STRATA$ is
    indicated here by SITES
        \item[2]  p-values below the significance level of 0.05 are indicated in bold
      \end{tablenotes}
          \end{threeparttable}

\end{table}

\begin{figure}[H]
    \centering
    \includegraphics[width=1\textwidth,height=0.65\textwidth]{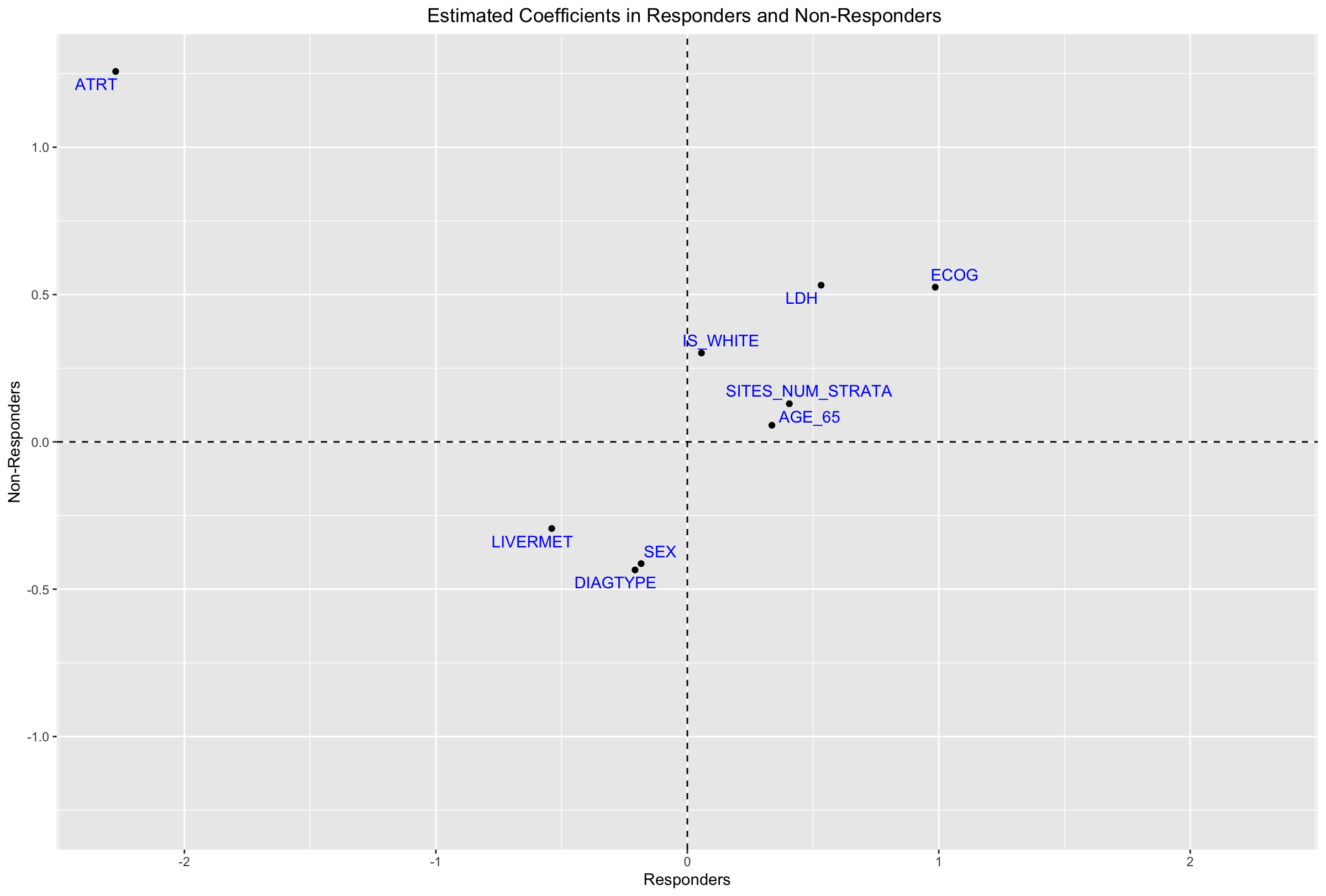}
    \caption{Estimated coefficients for the response and non-response groups separately for patients in the wild-type KRAS patients.}
    \label{fig:5.5}
\end{figure}

\noindent{\textbullet}\quad\parbox[t]{0.9\textwidth}{\textbf{Diagnostics}}

After subgroup classifying, we fit Cox proportional hazards model in the response and non-response groups, respectively.
After that, time-dependent ROC curve analysis is performed on the two subgroups and the overall population. Time-dependent ROC curve analysis
is widely used in biomedical research to assess the performance of models \cite{kamarudin_time-dependent_2017}.

Various definitions and estimation methods have been proposed in several literatures on time-dependent ROC curves
\cite{kamarudin_time-dependent_2017}\cite{Blanche2012TimedependentAW}. Here, based on the method proposed by Heagerty and
Zheng(2005)\cite{heagerty_survival_2005}, sensitivity is extended to incident sensitivity and specificity is extended to dynamic specificity,
and then ROC curves are drawn. Incident sensitivity refers to the probability that an individual has a marker value greater than c among
individuals experiencing an event at time t. Dynamic specificity refers to the probability that an individual has a marker value less than or
equal to c among individuals without an event at time t. It is generally accepted that higher marker values are more indicative of disease.
The performance of markers is assessed by the area under the ROC curve (AUC); the higher the AUC value, the better the model performance.

At each time point t, each individual is categorized as a case or control. For each patient $i (i=1, \ldots, n)$, let $T_i$ be the time of
disease event and $M_i$ be the marker value ($\boldsymbol{\beta}^{\mathrm{T}} \boldsymbol{x}$ under the Cox  proportional hazards model). The
observed survival time is $Y_i=\min \left(T_i\right.$, $\left.C_i\right)$, where $C_i$ is the censoring time, and let $\delta_i$ be the
censoring indicator variable that takes the value 1 when the event (disease) occurs and 0 otherwise.

Therefore, the sensitivity, specificity, and the resulting $A U C(t)$ are defined as
$$
\begin{aligned}
 S e^{I}(c,t)&=P\left(M_i>c \mid T_i=t\right), \\
 S p^{D}(c,t)&=P\left(M_i \leq c \mid T_i>t\right), \\
 A U C^{I, D}(t)&=P\left(M_i>M_j \mid T_i=t, T_j>t\right), i \neq j.
\end{aligned}
$$

Denote the risk set as $R(t)=\left\{i: Y_i \geq t\right\}$.
$S(t)=\left\{i: Y_i > t\right\}$, $n_t$ is the size of S(t). Hagerty and Zheng\cite{heagerty_survival_2005} proposed to make the following
estimates under the assumption of proportional hazards.
$$
\widehat{\mathrm{S e}}^{\mathbb{I}}(c, t)=\frac{\sum_{i \in R(t)} \mathbb{I}\left(M_i>c\right) \exp \left\{\boldsymbol{\beta}^{\mathrm{T}}
\boldsymbol{x_i}\right\}}{\sum_{i \in R(t)} \exp \left\{\boldsymbol{\beta}^{\mathrm{T}} \boldsymbol{x_i}\right\}},
$$

$$
\widehat{\operatorname{S p}}^D(c,t)=1-\frac{1}{n_t} \sum_{i \in S(t)} \mathbb{I}\left (M_i>c \right).
$$

Depending on the value of different c, the sensitivity, and specificity sequences at the same time point can be obtained, then the ROC curve
can be drawn. Xu and O'Quigley \cite{xu_proportional_2000} proved the consistency of sensitivity estimation. Since the specificity is an
empirical distribution on S(t), the specificity is also consistent if the sample is unbiased \cite{heagerty_survival_2005}.

From Figure \ref{fig:5.6}, it can be seen that from 0 to 1000 days, for the response group, the AUC is approximately between 0.8 and 0.6. For
the non-response group, the AUC is between 0.7 and 0.65, and for the overall population, it is consistently below 0.6. These illustrate that
the Cox model fitted for both the response and non-response groups performs better than the Cox model without classifying the subgroup.
Similarly, as seen in Figure \ref{fig:5.7}, the AUC at day 100 was 0.8 for the response group, 0.7 for the non-response group, and 0.58 for
the overall population, showing that our model can successfully split the population into two subgroups. In this dataset, the Dual Cox model
outperforms the Cox proportional hazards model.

\begin{figure}[H]
    \centering
    \includegraphics[width=1\textwidth,height=0.5\textwidth]{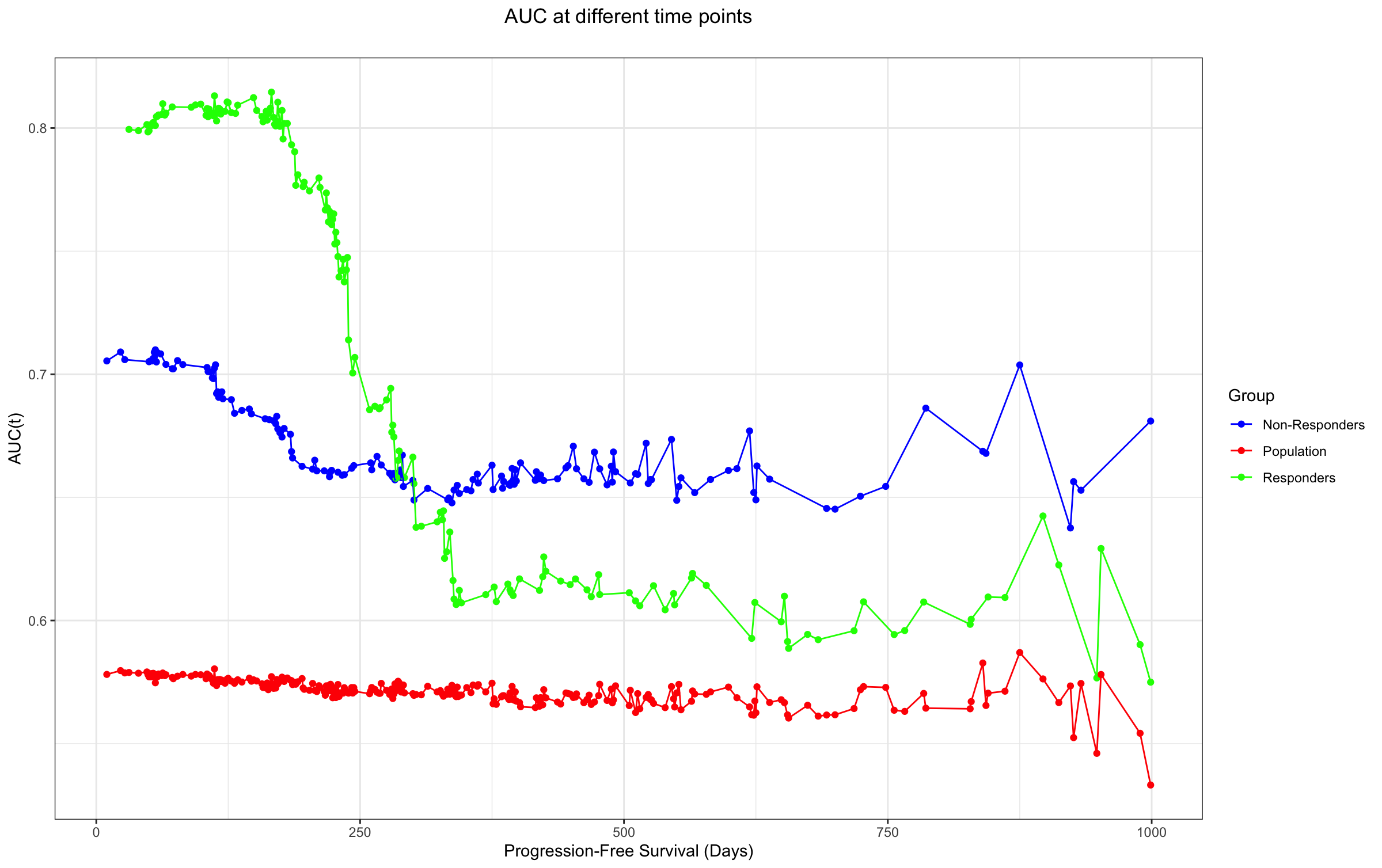}
    \caption{Values of AUC at different time points by fitting separate Cox proportional hazards models for the  responders, non-responders,
    and overall population in the wild-type KRAS patients.}
    \label{fig:5.6}
\end{figure}

\begin{figure}[H]
    \centering
    \includegraphics[width=1\textwidth,height=0.5\textwidth]{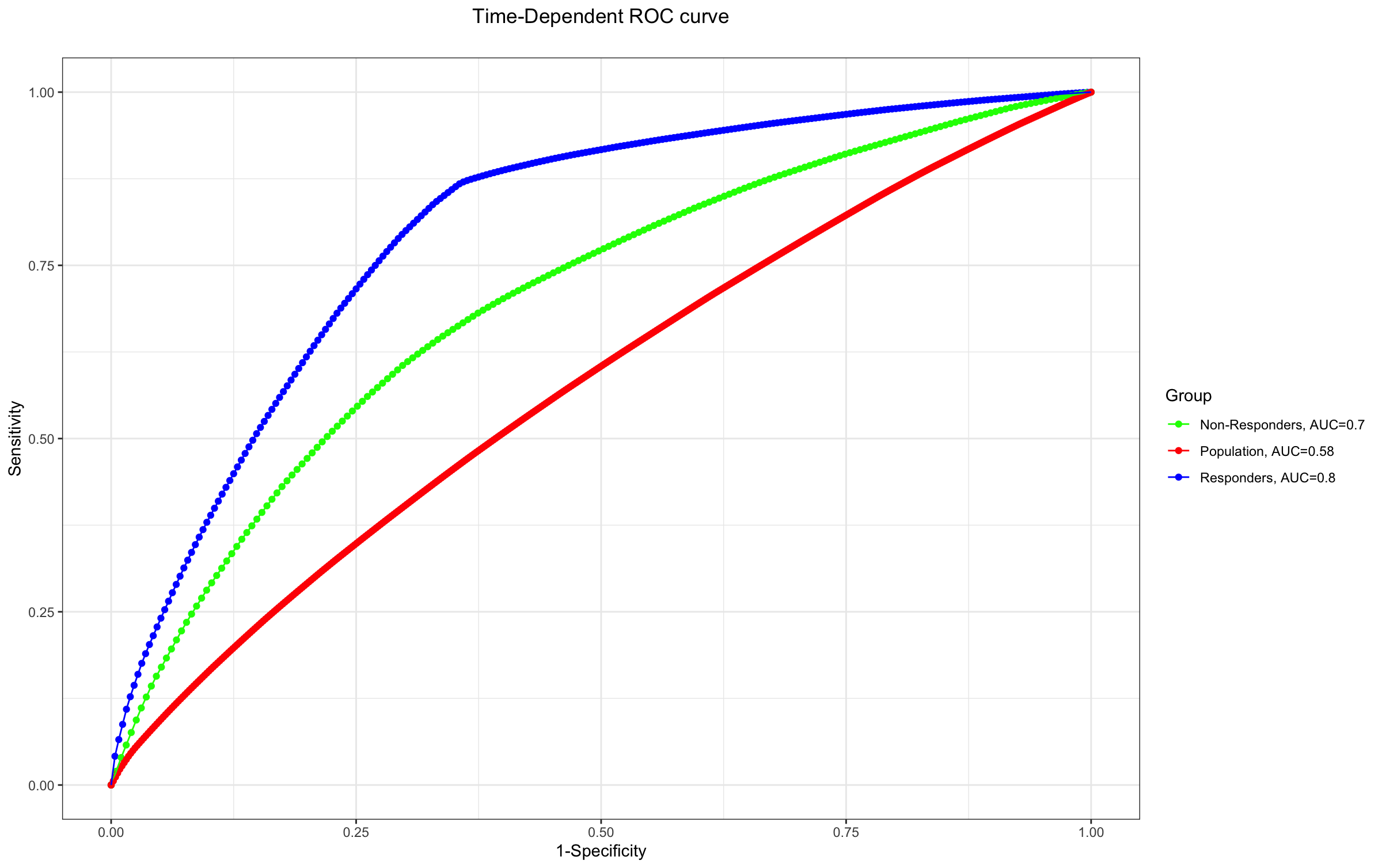}
    \caption{The ROC curves of the Cox proportional hazards model at day 100 for responders, non-responders, and the overall population in the
    wild-type KRAS patients.}
    \label{fig:5.7}
\end{figure}

\subsection{PACCE Trial}

PACCE\cite{hecht2009randomized} is a randomized, open-label, multicenter phase IIIB clinical trial designed to evaluate the efficacy and
safety of combining chemotherapy, bevacizumab (a humanized monoclonal antibody against tumor angiogenesis) and panitumumab (a monoclonal
antibody against the epidermal growth factor receptor) and chemotherapy in combination with bevacizumab for the treatment of patients with
metastatic colorectal cancer (mCRC). Based on investigator selection, patients were enrolled in one of two cohorts: a fluorouracil, folinic
acid, and oxaliplatin-based chemotherapy (Ox-CT) cohort or a fluorouracil, folinic acid, and irinotecan-based chemotherapy (Iri-CT) cohort.

The study was conducted between 2005 and 2007, and as shown in Figure \ref{fig:5.8}, 1053 patients were randomized to the study at 200 centers
in the U.S. 823 were enrolled in the Ox-CT cohort and 230 were enrolled in the Iri-CT cohort. In both cohorts, patients were assigned in a 1:1
ratio to the experimental (bevacizumab and panitumumab) and control (bevacizumab) groups. The primary endpoint is progression-free survival
(PFS), and secondary endpoints included overall survival (OS), objective response rate (ORR), and safety.

As shown in Figures \ref{fig:5.9}, \ref{fig:5.10} and Tables \ref{tab:5.6}, \ref{tab:5.7} (Note: The public data set contains only 80$\%$ of subjects,
and the results are slightly different from those in the original paper), the trial results in the Ox-CT cohort shows that, compare to
bevacizumab plus panitumumab, the use of bevacizumab alone improves progression-free survival in patients who receives bevacizumab
compares to bevacizumab plus panitumumab.

Specifically, median progression-free survival is 10.93 months (95$\%$ CI 9.23-11.83) in the
group receiving bevacizumab plus panitumumab compared with 12.10 months (95$\%$ CI 11.10-13.37) in the group receiving
bevacizumab, with a hazard ratio (HR) of 1.18 (95$\%$ CI 0.99-1.40. P=0.07), log-rank test p-value of 0.096. However, the trial results at
Iri-CT show no significant improvement in progression-free survival in patients on bevacizumab alone compared with bevacizumab added to
panitumumab. The median progression-free survival is 12.43 months (95$\%$ CI 10.10-14.47) in the group receiving bevacizumab added to
panitumumab compared to 11.87 months (95$\%$ CI 10.33-16.23) in the group receiving bevacizumab, with a hazard ratio (HR) of 1.09 (95$\%$ CI
0.74-1.61, p=0.66). The log-rank test p-value is 0.48. Thus, the difference between the Ox-CT cohort and the Iri-CT in the PACCE trial is
related to the difference in response to the addition of panitumumab between the two groups. Patients with tumors in the Iri-CT cohort have no
improvement in risk with the addition of panitumumab, whereas patients with tumors in the Ox-CT cohort have an increased risk.

In summary, the trial does not support the use of panitumumab in combination with bevacizumab and oxaliplatin or irinotecan-based chemotherapy
for the treatment of metastatic colorectal cancer. There is no difference between the panitumumab in combination with bevacizumab group and
the bevacizumab group in the Iri-CT cohort, and patients with panitumumab in combination with bevacizumab in the Ox-CT cohort perform worse in
terms of progression-free survival. We therefore next consider only 653 patients in the Ox-CT cohort.

\begin{figure}[H]
    \centering
    \includegraphics[width=1\textwidth,height=0.4\textheight]{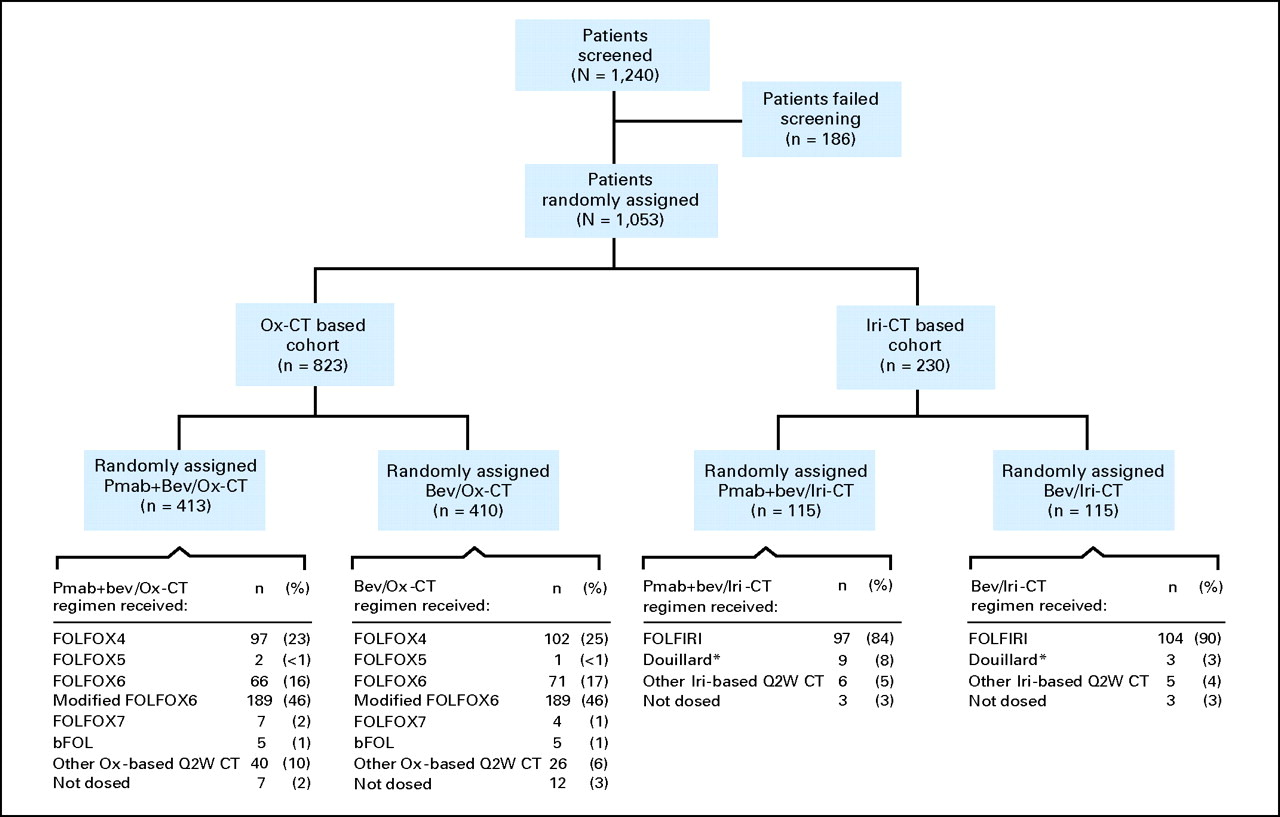}
    \caption{PACCE Trial\cite{hecht2009randomized}}
    \label{fig:5.8}
\end{figure}

\begin{figure}[H]
    \centering
    \includegraphics[width=1\textwidth,height=0.5\textwidth]{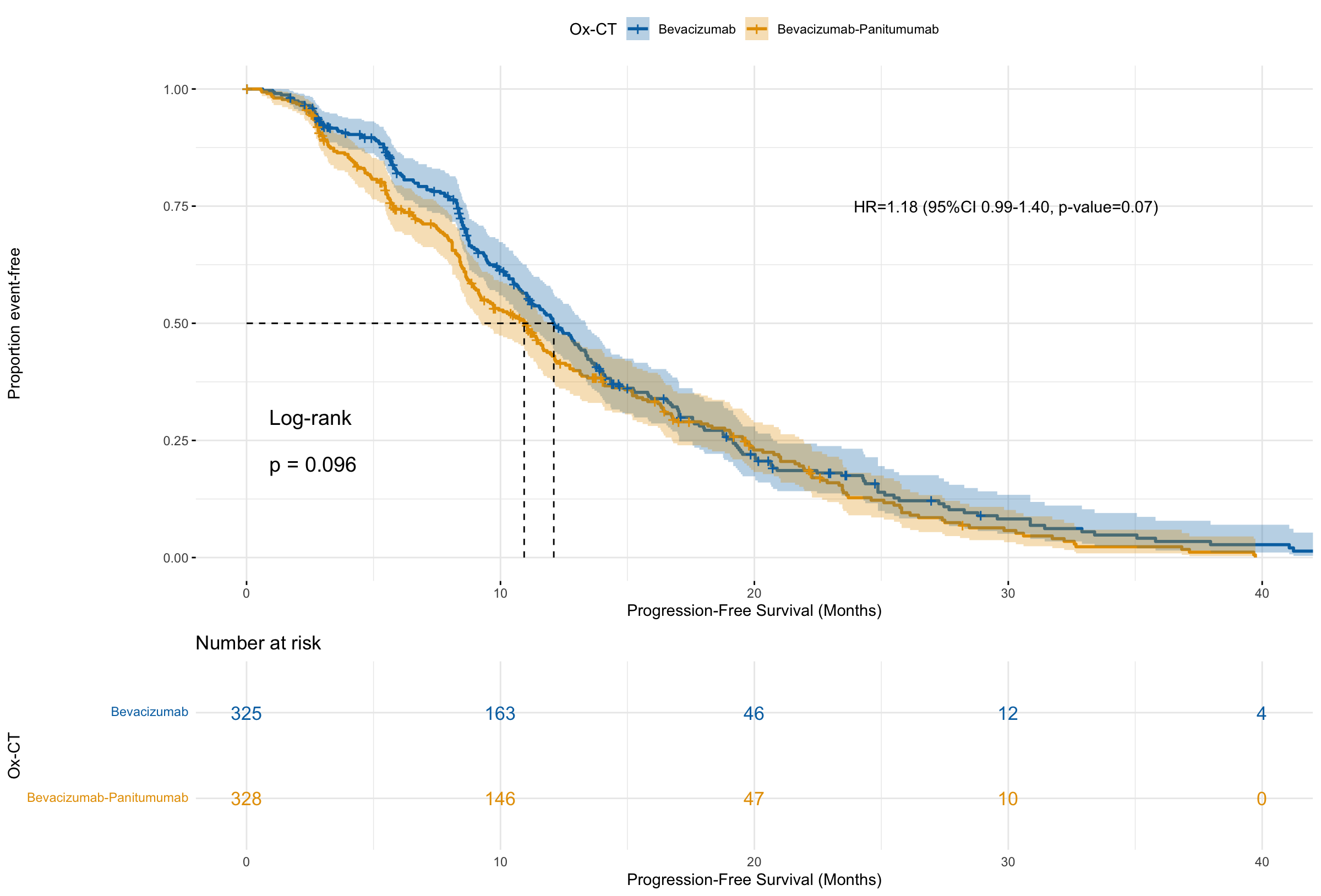}

    \caption{Kaplan-Meier curves, p-values from the log-rank test, hazard ratios with confidence intervals, and p-values for bevacizumab plus
    panitumumab and bevacizumab in the Ox-CT cohort of patients.}
    \label{fig:5.9}
\end{figure}

\begin{figure}[H]
    \centering
    \includegraphics[width=1\textwidth,height=0.5\textwidth]{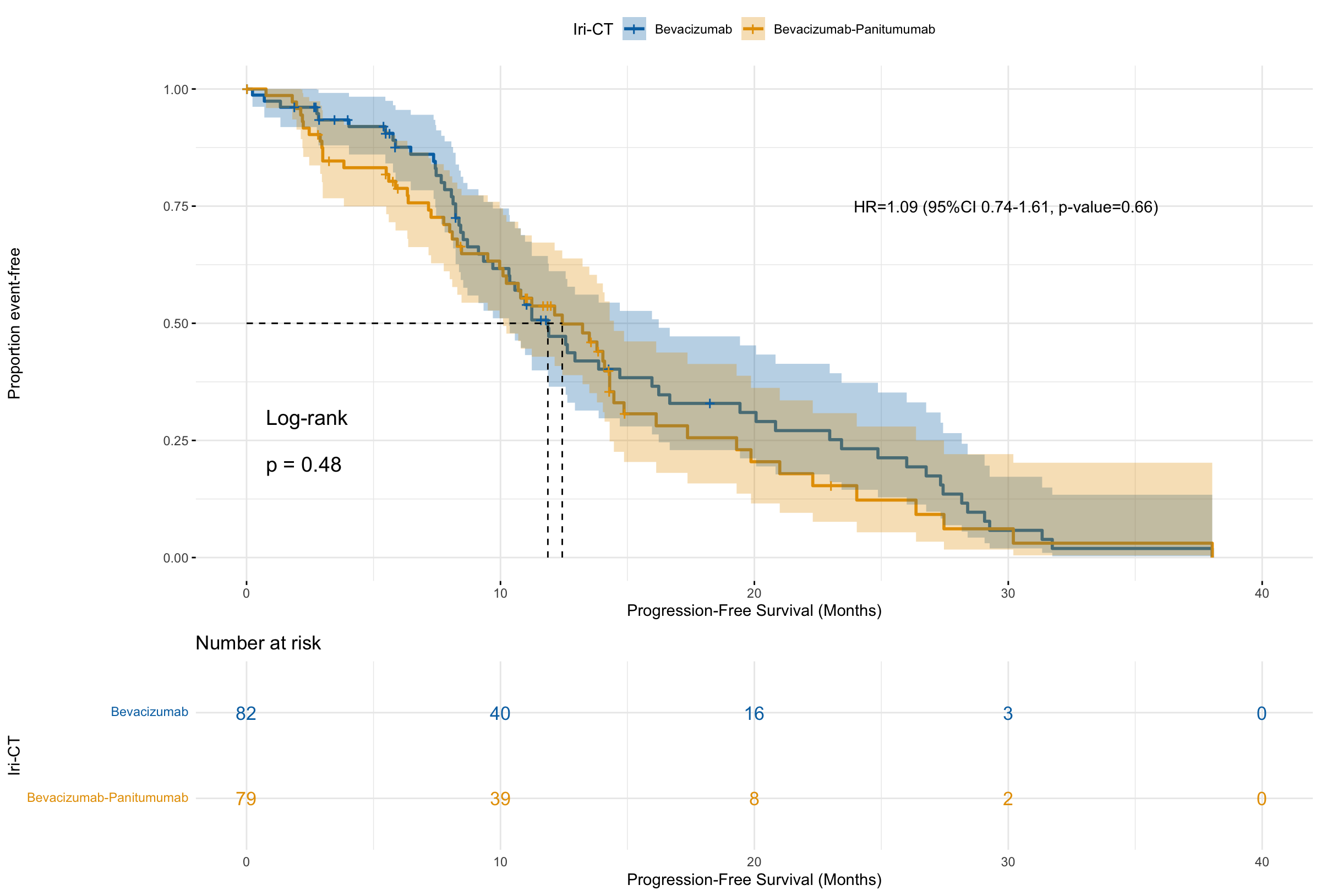}

    \caption{Kaplan-Meier curves, p-values from the log-rank test, hazard ratios with confidence intervals, and p-values for bevacizumab plus
    panitumumab and bevacizumab in the Iri-CT cohort of patients}
    \label{fig:5.10}
\end{figure}

\begin{table}[H]
\centering
\caption{Median progression-free survival and its 95$\%$ confidence interval with bevacizumab plus panitumumab and with bevacizumab in the
Ox-CT cohort of patients}
\begin{tabular}{rlrrr}
  \hline
 Treatment & Progression-free Survival & 95\% CI Lower & 95\% CI Upper \\
  \hline
Bevacizumab  & 12.10 & 11.10 & 13.37 \\
  Bevacizumab-Panitumumab  & 10.93 & 9.23 & 11.83 \\
   \hline
\end{tabular}
    \label{tab:5.6}

\end{table}

\begin{table}[H]
\centering
\caption{Median progression-free survival and its 95$\%$ confidence interval with bevacizumab plus panitumumab and with bevacizumab in the
Iri-CT cohort of patients.}
\begin{tabular}{rlrrr}
  \hline
  Treatment & Progression-free Survival & 95\% CI Lower & 95\% CI Upper \\
  \hline
Bevacizumab  & 11.87 & 10.33 & 16.23 \\
  Bevacizumab-Panitumumab  & 12.43 & 10.10 & 14.47 \\
   \hline
\end{tabular}
    \label{tab:5.7}

\end{table}

 As shown in Table \ref{tab:5.8}, the names, meanings, and values of the baseline covariates after preprocessing are given for 653
 patients in the Ox-CT cohort. To assess the impact of various baseline covariates on patient risk, a total of 10 baseline covariates were
 included: $ATRT, ECOG $,$SITES\_ NUM \_ STRATA$, SEX, $AGE\_65$, $IS \_ WHITE$, DIAGTYPE, LDH, KRASCD, and PRADJYN. The details are
 described below:
\begin{enumerate}
\item Treatment modality (ATRT): Bevacizumab + FOLFOX4 in the experimental group and bevacizumab in the control group.
\item ECOG performance status (ECOG): The Eastern Cooperative Oncology Group (ECOG) performance status was used to assess the functional
    status of patients. Patients scoring 0-2 were included in the trial. ECOG performance status is a measure of a patient's functional
    status and ability to perform daily activities. \
\item Metastatic burden ($SITES\_ NUM \_ STRATA$): The number of metastatic sites was analyzed as a potential prognostic factor. \

\item Gender (SEX): The patient's gender was recorded in the study entry. \

\item Age ($AGE\_65$): The patient's age was recorded at the time of study entry.\

\item Race ($IS \_ WHITE$): The patient's race was recorded at the time of study entry.

\item Primary tumor location (DIAGTYPE): The location of the primary tumor in the colon or rectum was recorded. \

\item Baseline LDH (LDH): As a marker of tumor burden and prognosis. High LDH levels are associated with poorer prognosis in patients with
    metastatic colorectal cancer. \

\item KRAS gene (KRASCD): KRAS gene is a useful biomarker for predicting response to panitumumab. \

\item Adjuvant therapy (PRADJYN): Records whether the patient has had adjuvant therapy. \
\end{enumerate}

\begin{table}[H] %voc table result
 \centering
    \caption{Name, meaning, take value, and percentage of baseline covariates after PRIME data preprocessing}
\begin{tabular}{|c|c|c| }
\hline Baseline covariates name & Meaning & Value ( percentage $\%$)   \\

\hline

 &  &1 = Bevacizumab +Panitumumab\\
ATRT&Experimental/control group & (experimental group 49$\%$)   \\
& &0=FOLFOX4(control group)\\
\hline

  &  &1=Less than 50$\%$ of time in bed ($41\%$)   \\
ECOG &Baseline ECOG status & 0=Fully active or \\ & &Symptomatic but movable people\\
\hline

SITES$\_$NUM$\_$STRATA & Metastasis sites &1=metastasis sites$\geq$ 3(21.7$\%$),0=others\\
\hline

SEX & Gender &1=Male(55.3$\%$),0=Female   \\
\hline

$AGE\_65$ & Age &1=Age$\geq$ 65 (34.8$\%$),0=
others   \\
\hline
$IS \_ WHITE$ &White or not &1= White or Caucasian(72.7$\%$),0=others \\
\hline
DIAGTYPE & Primary tumor type &1=Rectal(21.7$\%$),
0=Colon   \\

\hline
 &  &1=LDH \\
LDH & Lactate dehydrogenase levels &$\geq$1.5ULN(18.6$\%$)\\
 &  &0=LDH<1.5ULN \\
\hline

\hline
 &  &1=Wild type($49.31\%$)\\

KRASCD & KRAS Gene Results & 2=Mutant type($32.92\%$) \\
 &  & 88=Failure\\

\hline
PRADJYN & Adjuvant treatment  &1=Yes($33.54\%$),0=NO \\
\hline

\end{tabular}
\label{tab:5.8}
\end{table}

\noindent{\textbullet}\quad\parbox[t]{0.9\textwidth}{\textbf{Model Fitting}}

In the PACCE trial, we estimate an objective response rate of 0.45 in the experimental group, which was assessed according to mRECIST criteria
\cite{lencioni_modified_2010}. In this subsection, we fit the Dual Cox model to the 672 patients in the Ox-CT cohort, classified the
control population in the Ox-CT cohort, and investigated the effect of panitumumab plus Bevacizumab as compared with Bevacizumab alone.

Similar to the previous subsection 5.1.1 for the PRIME dataset, we randomly choose 1000 different initial values to fit the Dual Cox model
algorithm, from which we select the initial value that make the maximum value of the log-likelihood function.

Finally, we obtain the form of the the Dual Cox model as

$$S(t,\delta\mid\textbf{x})=\pi S_1(t,\delta\mid\textbf{x})+(1-\pi)S_2(t,\delta\mid\textbf{x}),$$

the hazard function in the response group

$$h_1(t \mid \boldsymbol{x})=h_{0 1}(t) \exp \left(\boldsymbol{\beta}_1^{\mathrm{T}} \boldsymbol{x}\right),$$the hazard function in the
non-response group
$$h_2(t \mid \boldsymbol{x})=h_{0 2}(t) \exp \left(\boldsymbol{\beta}_2^{\mathrm{T}} \boldsymbol{x}\right),$$where\begin{align*}
  \boldsymbol{x} &= (ATRT, ECOG, SITES\_NUM\_STRATA, SEX, \\
  &\qquad AGE\_65, IS\_WHITE, DIAGTYPE, LDH,KRASCD,PRADJYN)^{\mathrm{T}},\\
  \boldsymbol{\Pi}&=(0.53,0.47),\\
\boldsymbol{\beta}_1 &= (-\num{1.58}, \num{0.73}, \num{0.47}, \num{0.32}, \num{-0.11}, \num{-0.54}, \num{0.30}, \num{0.22}, \num{0.00},
\num{0.10})^{\mathrm{T}}, \\
\boldsymbol{\beta}_2 &= (\num{1.66}, \num{0.61}, \num{0.07}, \num{-0.01}, \num{0.30}, \num{-0.10}, \num{0.14}, \num{0.18}, \num{0.00},
\num{-0.14})^{\mathrm{T}}.
\end{align*}

\noindent{\textbullet}\quad\parbox[t]{0.9\textwidth}{\textbf{Subgroup Classification}}

The Dual Cox model divides the unobserved population into two groups, responders and non-responders, and the results are shown in Table
\ref{tab:5.9}. The response group contains 348 individuals with a censoring rate of 17.0$\%$, and the non-response group contains 305
individuals with a censoring rate of 31.8$\%$. The mean values of individual convergence probabilities $\hat{u}_{ik}$ in different subgroups
are 0.96 and 0.96, respectively, which are close to 1, indicating that the model can well converge individuals into the same subgroup, and
there is little difference between individuals within the group. In addition, as shown in Figure \ref{fig:5.11}, the Dual Cox model can
classify the non-censored data well, and the maximum posterior probability $\hat{u}_{i1},\hat{u}_{i2}$ is 1 for all non-censored data for the
unobserved control group. However, for the censored data, the model has some uncertainty in classifying patients, but the certainty of
classification increases as the progression-free survival increases, suggesting that our algorithm has difficulty classifying patients with
censored and short progression-free survival. These conclusions for the convergence results are consistent with those of the previous
subsection for the PRIME dataset. Based on the subgroup classification results, the objective response rate of 0.61 is estimated for the
control group, compared to an estimated response rate of 0.45 in the experimental group, and the convergence results are consistent with the
true situation, i.e., the response rate of patients with panitumumab combined with bevacizumab in the Ox-CT cohort is worse than that of
bevacizumab.

\begin{table}[H]
\centering
\caption{The Dual Cox model algorithm classifies patients with unobserved response conditions in the Ox-CT cohort}
\begin{tabular}{|c|c|c|c|c|}
\hline
& Total Patients & Censored Rate($\%$) & Mean of $\hat{u}_{ik}$  \\
\hline
Responders & 348 & 17.0$\%$ & 0.96 \\
\hline
Non-Responders & 305 & 31.8$\%$ & 0.96  \\
\hline
Overall Population & 653 & 23.9$\%$ & --\\
\hline
\end{tabular}
    \label{tab:5.9}
\end{table}

\begin{figure}[H]
    \centering
    \includegraphics[width=1\textwidth,height=0.5\textwidth]{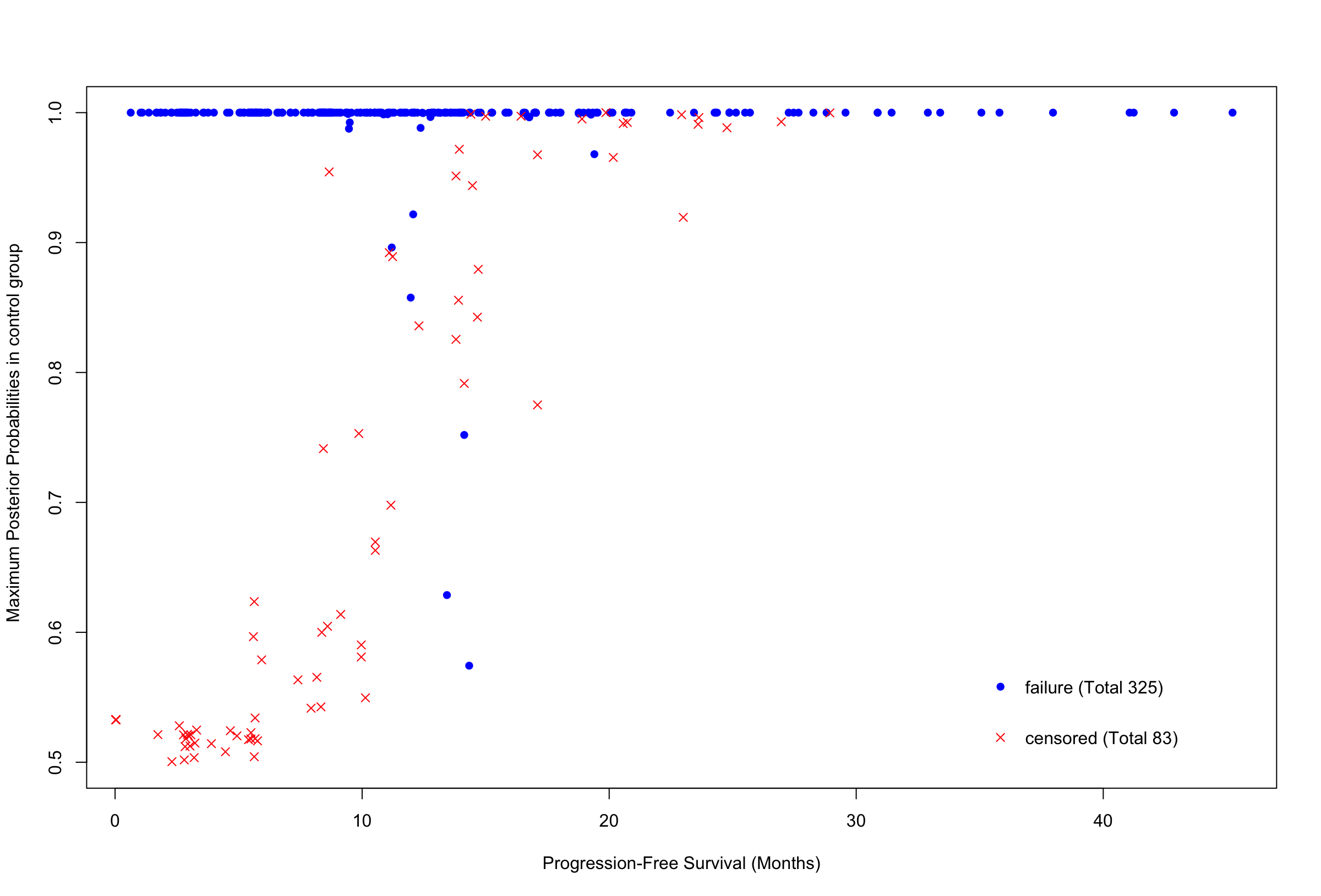}
    \caption{The maximum posterior probability $\max(\hat{u}_{i1},\hat{u}_{i2})$  fitted by the Dual Cox model algorithm for the Ox-CT cohort
    of patients with unobserved response conditions.}
    \label{fig:5.11}
\end{figure}

\noindent{\textbullet}\quad\parbox[t]{0.9\textwidth}{\textbf{Parameter Estimation}}

We fit the Cox proportional hazards model to the overall population and compare the results with those based on the Dual Cox Model. As shown
in Table \ref{tab:5.10} and Figure \ref{fig:5.12}, for the overall, response, and non-response groups on treatment efficacy ATRT, the hazard
ratio is 0.21 for the response group and 5.26 for the non-response group, and the p-values for both response and non-response groups are less
than the significance level of 0.05, indicating the responding and non-responding populations on treatment efficacy heterogeneity in ATRT;
some of the baseline covariates have the same effect on the population in the overall, responder, and non-responder groups, e.g., the
estimated regression coefficients signs of the six baseline covariates ECOG, SITES, SEX, DIAGTYPE, LDH, and KRASCD are the same. For the
baseline covariate ECOG (1=less than 50$\%$ of time in bed, 0=fully active or symptomatic but mobile), it is significant in all three groups
and has the same sign in all three groups, showing that ECOG=1 had a significantly higher risk than ECOG=0, which is in line with our
perception that bedridden people are at greater risk; also, for example, for age$\geq$ 65 years old or with metastatic site number $\geq$3 or
primary tumor type of rectum, the risk is significantly higher in the response and overall groups, but not in the non-response group,
suggesting that age $\geq$ 65 years old, metastatic site number $\geq$3 or primary tumor type of rectum in the non-response group does not
raise the risk of patients. For the effect of the covariate SEX, it can be found that only the response group has an effect on risk, while in
the overall and non-response groups, there is no effect on patients' progression-free survival; for KRASCD, PRADJYN, LDH, and $IS \_ WHITE$,
they are non-significant in all three groups, which means that in all three groups, these baseline covariates do not have an impact on patient
risk in any of the three groups.

\begin{table}[H] %voc table result
 \centering
    \caption{Estimated coefficients, HR, and p-values of the Ox-CT cohort patients are obtained from the Dual Cox model and the Cox
    proportional hazards model}
        \begin{threeparttable}

\begin{tabular}{p{2cm}ccccccccc}
\hline & \multicolumn{3}{c}{ Responders } & \multicolumn{3}{c}{ Non-Responders } & \multicolumn{3}{c}{ Overall Population } \\
\hline Covariates & $\hat{\boldsymbol{\beta}_1}$ & HR &p-values & $\hat{\boldsymbol{\beta}_2}$ & HR &p-values  & $\hat{\boldsymbol{\beta}}$ &
HR & p-values \\
\hline
ATRT & -1.58 & 0.21& $\pmb{2\cdot 10^{-16}}$
 & 1.66 &5.26 & $\pmb{2\cdot10^{-16}}$ & 0.16&  1.18 & 0.07 \\
\hline
ECOG & 0.73 & 2.07& $\pmb{2\cdot 10^{-8}}$ & 0.61 & 1.85& $\pmb{5\cdot10^{-5}}$ & 0.31 &1.37 & $\pmb{7\cdot10^{-4}}$\\
\hline
SITES &  0.47 & 1.60& $\textbf{0.004}$  & 0.07 &1.08 & 0.71 & 0.30 &1.35 & $\textbf{0.01}$\\
\hline
SEX & -0.54& 0.58 & $\pmb{2\cdot 10^{-5}}$ & -0.10 &0.91 & 0.51 & -0.04 & 0.96& 0.64 \\
\hline
AGE\_65 & 0.32 &1.38 & $\textbf{0.01}$ & -0.01 & 0.99& 0.95 & 0.21 &1.23& $\textbf{0.02}$  \\
\hline
IS\_WHITE & -0.11 & 0.89& 0.81 & 0.30 & 1.35& 0.14 & 0.10  &1.10&0.45\\
\hline
DIAGTYPE & 0.30 & 1.35& $\textbf{0.05}$ & 0.14 & 1.15 & 0.43 & 0.24 &1.27& $\textbf{0.03}$\\
\hline
LDH & 0.22 &1.25 &  0.15 & 0.18 &1.20 & 0.29 & 0.03 &1.03&0.81\\
\hline
KRASCD & 0 &1.00 &  0.98 & 0 &1.00 & 0.16 & 0 &1.00&0.33\\
\hline
PRADJYN & 0.10 &1.11 &  0.58 & -0.14 &0.87 & 0.44 & 0.06 &1.06&0.63\\
\hline
\end{tabular}

\label{tab:5.10}

 \begin{tablenotes}
        \footnotesize
        \item[1] $SITES\_NUM\_STRATA$ is
    indicated here by SITES
        \item[2]  p-values below the significance level of 0.05 are indicated in bold
      \end{tablenotes}
          \end{threeparttable}

\end{table}

\begin{figure}[H]
    \centering
    \includegraphics[width=1\textwidth,height=0.65\textwidth]{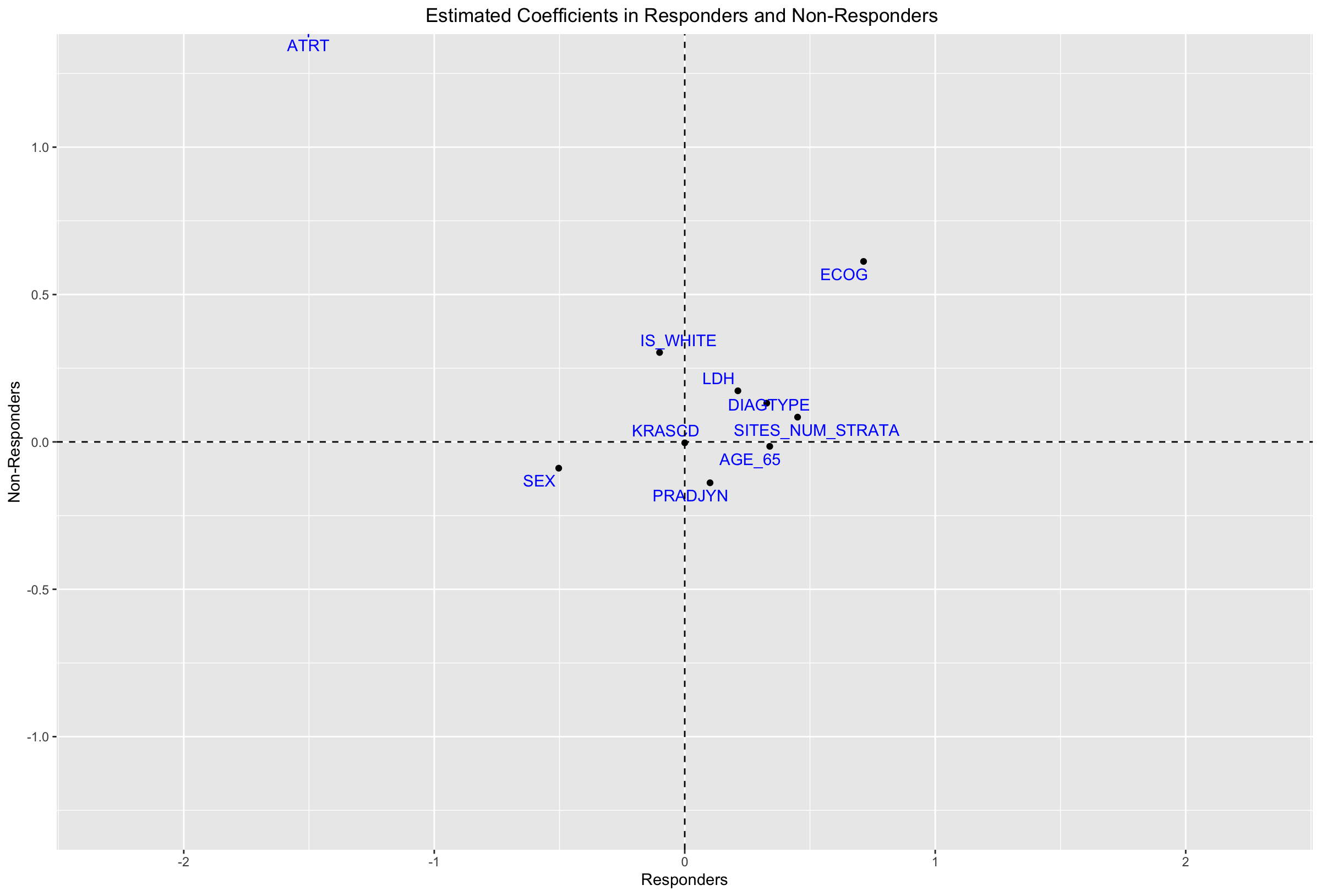}
    \caption{ Estimated coefficients for the response and non-response groups separately for patients in the Ox-CT cohort.}
    \label{fig:5.12}
\end{figure}

\noindent{\textbullet}\quad\parbox[t]{0.9\textwidth}{\textbf{Diagnostics}}

As in section 5.1.1 for the PRIME data, to assess the performance of the model, after classifying the subgroup, we fit three separate Cox
proportional hazards models for the responders, non-responders, and the overall population, after which time-dependent ROC curves are
analyzed.

From Figure \ref{fig:5.13}, it can be seen that from 0 to 900 days, for the response group, the AUC is approximately between 0.75 and 0.55.
For the non-response group, the AUC is between 0.75 and 0.68, and for the overall population it is consistently below 0.6. These illustrate
that the Cox model fitted for both the response and non-response groups performs better than the unclassified population, meaning that the
Dual Cox model is better. In addition, it can be found that for the response group, the AUC decreases gradually after 250 days. To better
illustrate that our model performs better than fitting the Cox proportional hazards model alone, we give the ROC curves for the three
populations at day 100, and from Figure \ref{fig:5.14}, we can see that the AUC at day 100 is 0.73 for the response group, 0.72 for the
non-response group, and 0.57 overall, all of which indicate that our model can effectively classify the subgroup, and the Dual Cox model
fitted to the corresponding group have higher model performance compared to the overall population.

\begin{figure}[H]
    \centering
    \includegraphics[width=1\textwidth,height=0.5\textwidth]{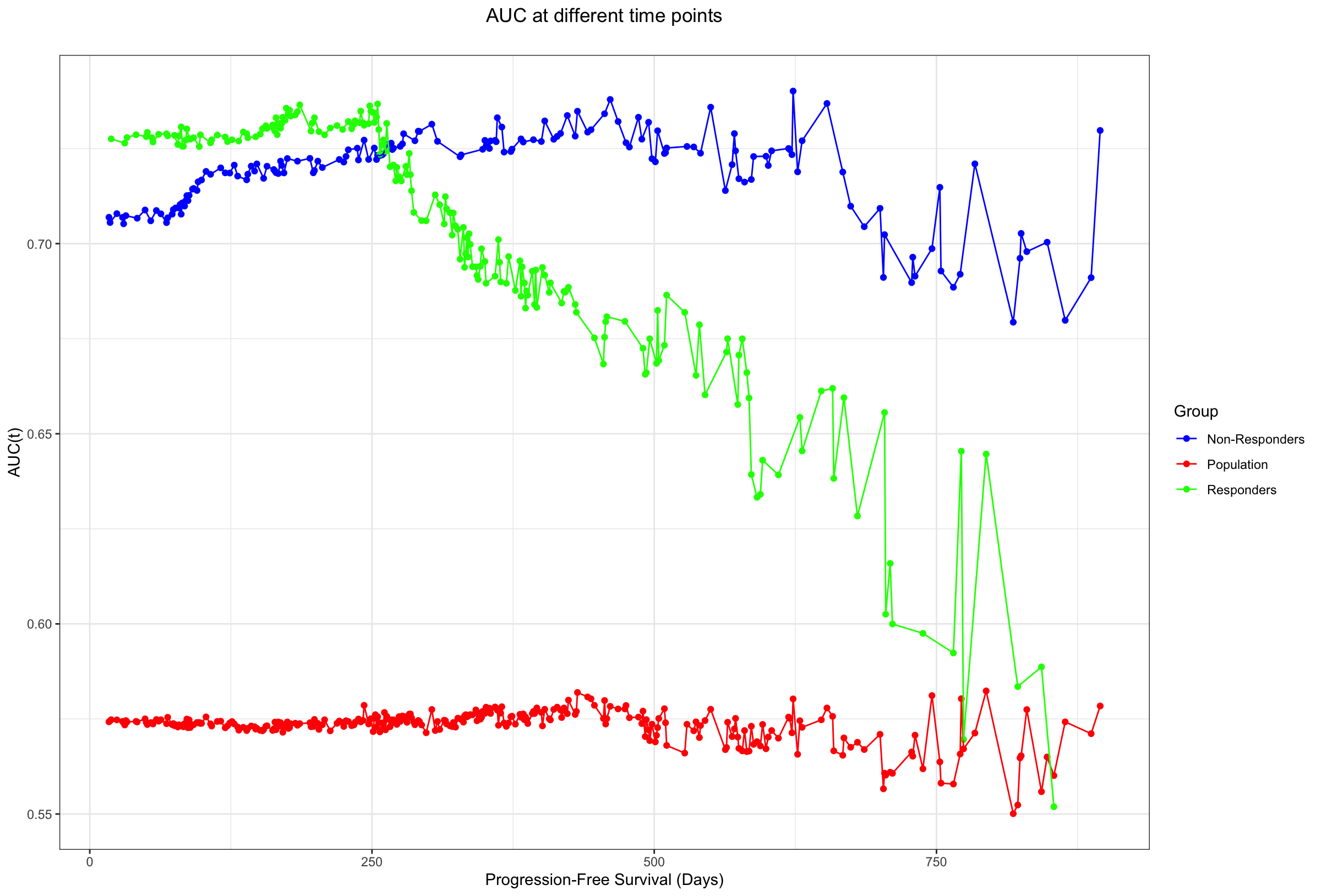}
    \caption{Values of AUC at different time points by fitting separate Cox proportional hazards models for the  responders, non-responders,
    and overall population in the Ox-CT cohort}
    \label{fig:5.13}
\end{figure}

\begin{figure}[H]
    \centering
    \includegraphics[width=1\textwidth,height=0.5\textwidth]{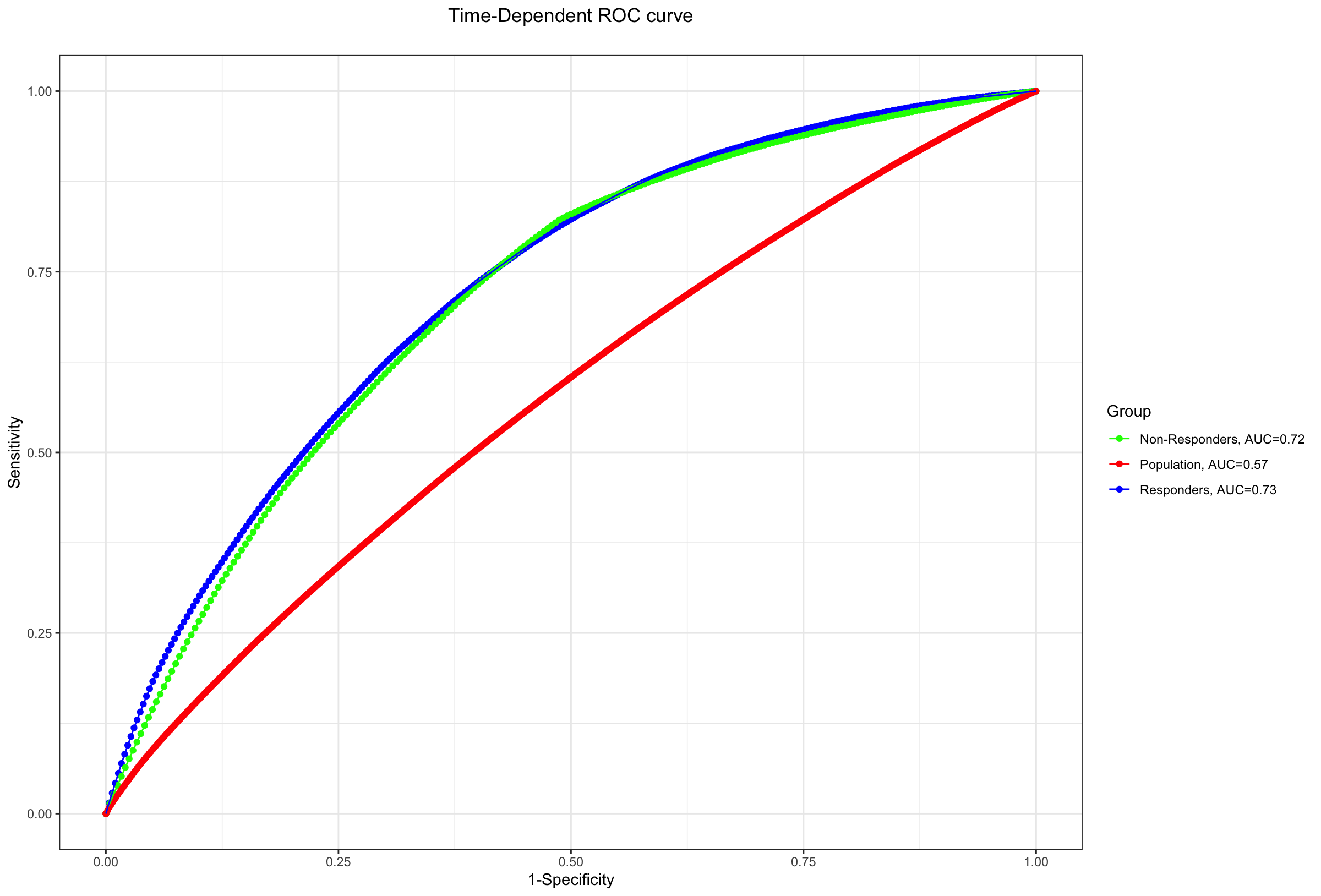}
    \caption{The ROC curves of the Cox proportional hazards model at day 100 for responders, non-responders, and the overall population in the
    Ox-CT cohort.}

    \label{fig:5.14}
\end{figure}

\section{Conclusion and Discussion}

\subsection{Conclusion}

Targeted therapy and immunotherapy differ significantly from conventional cancer treatment. Cox proportional hazards models may not be as
effective when long-term survival and delayed clinical effects occur due to individual differences in drug response.  Since the innovative
cancer treatment and the control treatment have such distinct mechanisms of action, we can only observe responses in the
experimental group. These considerations prompt us to investigate a semi-supervised model that would classify patients in the control group
into responders and non-responders, thereby reducing the number of samples and time required to complete clinical trials and lowering research
and development expenses. In other words, we can predict how the control population will respond to the medicine and restrict the assumption
of proportional hazards to just the responders and non-responders.  Finally, we analyzed the drug's effectiveness by modeling both the
responders and the non-responders simultaneously. The Dual Cox Model is therefore what we advocate for.

For the fitting of the Dual Cox model, as with the general finite mixture model, the semi-supervised mixture Cox model uses the EM
algorithm to calculate the posterior probabilities of individuals belonging to the response and non-response groups by the E step, updates the
mixing probabilities by the M step, and estimates the regression coefficients by maximizing partial likelihood function in the weighted Cox
proportional hazards model. The algorithm iterates until the log-likelihood function converges by successive iterations.

This paper expands on the work of Anderson and Gill (1982)\cite{andersen1982cox} and You et al. (2018)\cite{you2018subtype} to further show
that our model has estimators with consistency and asymptotic normality, proving the validity of the Dual Cox model with large samples. A
series of simulation tests are carried out in this study to confirm the validity of the model fit further once the theoretical properties are
guaranteed. We first investigate the impact of various initial values selections on the model’s convergence outcomes. The results of the
experiments demonstrate that the method will converge to a local maximum if the initial values are chosen on the boundary, thus leading to
unguaranteed convergence. In contrast, the initial values chosen away from the boundary can produce better convergence outcomes. The impact of
various sample sizes and censoring rates on the efficacy of the model is then covered. The experimental findings demonstrate that our approach
has high classification accuracy and can guarantee unbiased parameter estimation. The stability is assured with 1000 repetitions of an
experiment with minor variance.

To demonstrate the utility of the Dual Cox model theory in oncology clinical trials, we select the PRIME and PACCE trial datasets. For the subgroup
classification problem, the Dual Cox model shows good ability to classify non-censored data accurately, but there is some
uncertainty for the classification of censored data. However, the certainty of classification gradually improves with increasing
progression-free survival. In estimating the objective response rate in the control group, the model can better reflect the actual situation.
In addition, the model can effectively demonstrate heterogeneity in the drug's efficacy in estimating regression coefficients and are
consistent with our knowledge of the homogeneity of other baseline covariates. Finally, we use the diagnostics to demonstrate that
the Dual Cox model outperforms the Cox proportional hazards model adapted to the overall population at various time points.

\subsection{Discussion}

Finite mixture models usually assume that each subgroup has a separate distribution to deal with heterogeneity among subgroups. Although the
Cox proportional hazards model has become the most popular regression model applicable to censored time, its non-parametric properties make
the construction of a Cox model with a finite mixture more challenging\cite{you2018subtype}. Therefore, Eng and Hanlon
(2014)\cite{eng_discrete_2014} proposed an algorithm to fit a finite mixture Cox model. Our proposed Dual Cox model serves as an extension of
the finite mixture Cox model, whose main concern is to classify patients into responders' and non-responders' groups, thus fixing the number
of subgroups to 2. In addition, a semi-supervised mixture Cox model is necessary to take into account some of the patients for whom a
medication response has been observed. The distributions selected in the finite mixture model can be those used in common survival analyses in
addition to the Cox proportional hazards model to characterize different subgroups of hazards, including Log-normal distribution, Weibull
distribution, Gamma distribution, etc.\cite{erisoglu2012pak}\cite{mclachlan_role_1994}. Similar to our work, a mixture model based on these
distributions can be extended to the semi-supervised case and compared with the Dual Cox model, which can help us better to assess the merits
of the Dual Cox model.

The setting of the number of subgroups of the finite mixture Cox model generally needs to be set in advance, and the number of subgroups and
the selection of initial values may lead to unstable convergence results of the EM algorithm, thus posing a challenge to the analysis.
Compared with this, the advantage of the Dual Cox model is that the number of fixed groups is two, which are the responses’ group and the
non-responders’ group, and the convergence result of its algorithm is more stable than that of the mixture Cox model with multiple groups.

The reasonableness of the $\boldsymbol{\beta}$ estimation has been shown in Section \ref{sec:dual3}, but the strict concavity of the
likelihood function is based on the assumption of regularity conditions. Therefore, after fitting using the EM algorithm, there is no
guarantee that the likelihood function of the Dual Cox model necessarily converges to a global maximum. In our experiments, with a 1:1 ratio
of labeled to unlabeled data, the log-likelihood function converges to higher values well enough to achieve high classification accuracy and
ensure consistency and asymptotic normality in the estimated coefficients, as long as we avoid the choice of initial values in the boundary.
To mitigate the effect of initial values, in addition to using the EM algorithm, Ueda and Nakano (1998)\cite{ueda1998deterministic} proposed
using the deterministic annealing EM algorithm, while Celeux, Chauveau, and Diebolt (1996)\cite{Celeux1996StochasticVO} proposed a stochastic
EM algorithm, all of these techniques seek to allow the EM algorithm to iterate from less-than-ideal initial values and converge to better
results. Consequently, we will consider how to appropriately modify the algorithm for the Dual Cox model in the future.

With the increase in data collection and baseline covariates, variable selection becomes an effective method for selecting risk factors. The
importance of variable selection is particularly highlighted when there are many baseline covariates or even when the number of variables is
larger than the sample size. The variables that are statistically significantly different for the three groups, the response group, the
non-response group, and the overall population, are also different, so a model that allows variable selection on different subgroups becomes
especially necessary. Tibshirani (1997)\cite{tibshirani1997lasso} proposed the use of LASSO on a Cox proportional hazards model with penalty
term, Fan and Li (2002)\cite{fan2002variable} used SCAD penalty term on Cox proportional hazards model, You et al. (2018)\cite{you2018subtype}
used adaptive LASSO penalty term on finite mixture Cox model. Thus, the Dual Cox model can be easily extended to include the use of variable
selection methods to identify risk factors. By selecting variables, a better understanding of risk factors in different groups can be
achieved, providing stronger evidence for clinical decision-making.

After determining the optimal subgroup, an important step is to evaluate its actual effects. Although statistical inference is usually used to
analyze the selected subgroups, we must note that the subgroups may not be completely independent of the data selection. This may lead to an
overly optimistic evaluation of the selected subgroups, which may influence our judgment of the treatment effect \cite{guo_inference_2021}. In
addition, the lack of ability to detect treatment effects in smaller patient populations may produce false-positive results, which is
particularly important because most clinical trials are designed to initially test treatment effects on the entire study population
\cite{bonetti_patterns_2004}. Thus, the Dual Cox model algorithm’s fitting results concerning subgroup classification may be biased. To reduce
this bias, we need to adopt an approach to minimize the impact of false-positive results on the evaluation of treatment effects and provide
more accurate individualized treatment plans to achieve the best results averaged over the entire patient population \cite{zhang_robust_2012}.
This will also be an essential research direction in our future work.

\section*{Acknowledgments}
This was supported in part by the Fundamental Research Funds for the Central Universities, Sun Yat-sen University (Grant No. 20lgpy145 \&
2021qntd21) and the Science and Technology Program of Guangzhou Project, Fundamental and Applied Research Project (202102080175).

%Bibliography
\bibliographystyle{unsrt}
\bibliography{references}

\begin{thebibliography}{10}

\bibitem{wyld_evolution_2015}
Lynda Wyld, Riccardo~A. Audisio, and Graeme~J. Poston.
\newblock The evolution of cancer surgery and future perspectives.
\newblock {\em Nature Reviews Clinical Oncology}, 12(2):115--124, 2015.

\bibitem{warren_inhalation_1846}
John~C. Warren.
\newblock Inhalation of {Ethereal} {Vapor} for the {Prevention} of {Pain} in
  {Surgical} {Operations}.
\newblock {\em The Boston Medical and Surgical Journal}, 35(19):375--379, 1846.

\bibitem{lister1867antiseptic}
Joseph Lister.
\newblock On the antiseptic principle in the practice of surgery.
\newblock {\em British medical journal}, 2(351):246, 1867.

\bibitem{devita_two_2012}
Vincent~T. DeVita and Steven~A. Rosenberg.
\newblock Two {Hundred} {Years} of {Cancer} {Research}.
\newblock {\em New England Journal of Medicine}, 366(23):2207--2214, 2012.

\bibitem{papac_origins_2001}
R.~J. Papac.
\newblock Origins of cancer therapy.
\newblock {\em The Yale Journal of Biology and Medicine}, 74(6):391--398, 2001.

\bibitem{gianfaldoni_overview_2017}
Serena Gianfaldoni, Roberto Gianfaldoni, Uwe Wollina, Jacopo Lotti, Georgi
  Tchernev, and Torello Lotti.
\newblock An overview on radiotherapy: from its history to its current
  applications in dermatology.
\newblock {\em Open access Macedonian journal of medical sciences}, 5(4):521,
  2017.

\bibitem{devita2008history}
Vincent~T DeVita~Jr and Edward Chu.
\newblock A history of cancer chemotherapy.
\newblock {\em Cancer research}, 68(21):8643--8653, 2008.

\bibitem{amjad_cancer_2022}
Muhammad~T. Amjad, Anusha Chidharla, and Anup Kasi.
\newblock Cancer {Chemotherapy}.
\newblock In {\em {StatPearls}}. StatPearls Publishing, Treasure Island (FL),
  2022.

\bibitem{kohler_continuous_1975}
G.~K\"{o}hler and C.~Milstein.
\newblock Continuous cultures of fused cells secreting antibody of predefined
  specificity.
\newblock {\em Nature}, 256(5517):495--497, 1975.

\bibitem{padma2015overview}
Viswanadha~Vijaya Padma.
\newblock An overview of targeted cancer therapy.
\newblock {\em BioMedicine}, 5:1--6, 2015.

\bibitem{lee_molecular_2018}
Yeuan~Ting Lee, Yi~Jer Tan, and Chern~Ein Oon.
\newblock Molecular targeted therapy: {Treating} cancer with specificity.
\newblock {\em European Journal of Pharmacology}, 834:188--196, 2018.

\bibitem{wu2006targeted}
Han-Chung Wu, De-Kuan Chang, and Chia-Ting Huang.
\newblock Targeted therapy for cancer.
\newblock {\em J Cancer Mol}, 2(2):57--66, 2006.

\bibitem{davis_overview_2000}
Ian~D Davis.
\newblock An overview of cancer immunotherapy.
\newblock {\em Immunology and Cell Biology}, 78(3):179--195, 2000.

\bibitem{srivastava2015engineering}
Shivani Srivastava and Stanley~R. Riddell.
\newblock Engineering {CAR}-{T} cells: {Design} concepts.
\newblock {\em Trends in Immunology}, 36(8):494--502, 2015.

\bibitem{melenhorst_decade-long_2022}
J.~Joseph Melenhorst, Gregory~M. Chen, Meng Wang, David~L. Porter, Changya
  Chen, McKensie~A. Collins, Peng Gao, Shovik Bandyopadhyay, Hongxing Sun,
  Ziran Zhao, Stefan Lundh, Iulian Pruteanu-Malinici, Christopher~L. Nobles,
  Sayantan Maji, Noelle~V. Frey, Saar~I. Gill, Alison~W. Loren, Lifeng Tian,
  Irina Kulikovskaya, Minnal Gupta, David~E. Ambrose, Megan~M. Davis, Joseph~A.
  Fraietta, Jennifer~L. Brogdon, Regina~M. Young, Anne Chew, Bruce~L. Levine,
  Donald~L. Siegel, C\'{e}cile Alanio, E.~John Wherry, Frederic~D. Bushman,
  Simon~F. Lacey, Kai Tan, and Carl~H. June.
\newblock Decade-long leukaemia remissions with persistence of {CD4}+ {CAR} {T}
  cells.
\newblock {\em Nature}, 602(7897):503--509, 2022.

\bibitem{chen_statistical_2013}
Tai-Tsang Chen.
\newblock Statistical issues and challenges in immuno-oncology.
\newblock {\em Journal for ImmunoTherapy of Cancer}, 1(1):18, 2013.

\bibitem{raval_tumor_2014}
Raju~R Raval, Andrew~B Sharabi, Amanda~J Walker, Charles~G Drake, and Padmanee
  Sharma.
\newblock Tumor immunology and cancer immunotherapy: summary of the 2013 {SITC}
  primer.
\newblock {\em Journal for ImmunoTherapy of Cancer}, 2(1):14, 2014.

\bibitem{mahoney_next_2015}
Kathleen~M. Mahoney, Gordon~J. Freeman, and David~F. McDermott.
\newblock The {Next} {Immune}-{Checkpoint} {Inhibitors}: {PD}-1/{PD}-{L1}
  {Blockade} in {Melanoma}.
\newblock {\em Clinical Therapeutics}, 37(4):764--782, 2015.

\bibitem{waldman_guide_2020}
Alex~D. Waldman, Jill~M. Fritz, and Michael~J. Lenardo.
\newblock A guide to cancer immunotherapy: from {T} cell basic science to
  clinical practice.
\newblock {\em Nature Reviews Immunology}, 20(11):651--668, 2020.

\bibitem{falzone_evolution_2018}
Luca Falzone, Salvatore Salomone, and Massimo Libra.
\newblock Evolution of {Cancer} {Pharmacological} {Treatments} at the {Turn} of
  the {Third} {Millennium}.
\newblock {\em Frontiers in Pharmacology}, 9:1300, 2018.

\bibitem{renfro_statistical_2017}
L.A. Renfro and D.J. Sargent.
\newblock Statistical controversies in clinical research: basket trials,
  umbrella trials, and other master protocols: a review and examples.
\newblock {\em Annals of Oncology}, 28(1):34--43, 2017.

\bibitem{catenacci_next-generation_2015}
Daniel~V.T. Catenacci.
\newblock Next-generation clinical trials: {Novel} strategies to address the
  challenge of tumor molecular heterogeneity.
\newblock {\em Molecular Oncology}, 9(5):967--996, 2015.

\bibitem{kaplan_nonparametric_1958}
E.~L. Kaplan and Paul Meier.
\newblock Nonparametric {Estimation} from {Incomplete} {Observations}.
\newblock {\em Journal of the American Statistical Association},
  53(282):457--481, 1958.

\bibitem{bland2004logrank}
J~Martin Bland and Douglas~G Altman.
\newblock The logrank test.
\newblock {\em Bmj}, 328(7447):1073, 2004.

\bibitem{cox_regression_1972}
D.~R. Cox.
\newblock Regression {Models} and {Life}-{Tables}.
\newblock {\em Journal of the Royal Statistical Society: Series B
  (Methodological)}, 34(2):187--202, 1972.

\bibitem{liao_flexible_2019}
Jason~J.Z. Liao and Guanghan~Frank Liu.
\newblock A flexible parametric survival model for fitting time to event data
  in clinical trials.
\newblock {\em Pharmaceutical Statistics}, 18(5):555--567, 2019.

\bibitem{kalbfleisch_estimation_1981}
John~D. Kalbfleisch and Ross~L. Prentice.
\newblock Estimation of the average hazard ratio.
\newblock {\em Biometrika}, 68(1):105--112, 1981.

\bibitem{boag1949maximum}
John~W Boag.
\newblock Maximum likelihood estimates of the proportion of patients cured by
  cancer therapy.
\newblock {\em Journal of the Royal Statistical Society. Series B
  (Methodological)}, 11(1):15--53, 1949.

\bibitem{tsodikov1996stochastic}
Alexander~D Tsodikov, Andrei~Yu Yakovlev, and B~Asselain.
\newblock {\em Stochastic models of tumor latency and their biostatistical
  applications}, volume~1.
\newblock World Scientific, 1996.

\bibitem{mclachlan_role_1994}
Gj~McLachlan and Dc~McGiffin.
\newblock On the role of finite mixture models in survival analysis.
\newblock {\em Statistical Methods in Medical Research}, 3(3):211--226, 1994.

\bibitem{mclachlan_finite_2000}
Geoffrey McLachlan and David Peel.
\newblock {\em Finite {Mixture} {Models}: {McLachlan}/{Finite} {Mixture}
  {Models}}.
\newblock Wiley {Series} in {Probability} and {Statistics}. John Wiley \& Sons,
  Inc., Hoboken, NJ, USA, 2000.

\bibitem{erisoglu2012pak}
Ulku Erisoglu, Murat Erisoglu, and Hamza Erol.
\newblock Pak. j. statist. 2012 vol. 28 (1), 115-130 mixture model approach to
  the analysis of heterogeneous survival data.
\newblock {\em Pak. J. Statist}, 28(1):115--130, 2012.

\bibitem{liu_analysis_2020}
Guanghan~Frank Liu and Jason J.~Z. Liao.
\newblock Analysis of time-to-event data using a flexible mixture model under a
  constraint of proportional hazards.
\newblock {\em Journal of Biopharmaceutical Statistics}, 30(5):783--796, 2020.

\bibitem{jia_inferring_2021}
Beilin Jia, Donglin Zeng, Jason J.~Z. Liao, Guanghan~F. Liu, Xianming Tan,
  Guoqing Diao, and Joseph~G. Ibrahim.
\newblock Inferring latent heterogeneity using many feature variables
  supervised by survival outcome.
\newblock {\em Statistics in Medicine}, 40(13):3181--3195, 2021.

\bibitem{wang_exponential_2014}
Chi Wang.
\newblock An {Exponential} {Tilt} {Mixture} {Model} for {Time}-to-{Event}
  {Data} to {Evaluate} {Treatment} {Effect} {Heterogeneity} in {Randomized}
  {Clinical} {Trials}.
\newblock {\em Biometrics \& Biostatistics International Journal}, 1(2), 2014.

\bibitem{qin_empirical_1999}
Jing Qin.
\newblock Empirical {Likelihood} {Ratio} {Based} {Confidence} {Intervals} for
  {Mixture} {Proportions}.
\newblock {\em The Annals of Statistics}, 27(4):1368--1384, 1999.
\newblock Publisher: Institute of Mathematical Statistics.

\bibitem{zou_empirical_2002}
F.~Zou.
\newblock On empirical likelihood for a semiparametric mixture model.
\newblock {\em Biometrika}, 89(1):61--75, 2002.

\bibitem{rosen_mixtures_1999}
Ori Rosen and Martin Tanner.
\newblock Mixtures of proportional hazards regression models.
\newblock {\em Statistics in Medicine}, 18(9):1119--1131, 1999.

\bibitem{wu_subgroup_2016}
Ruo-fan Wu, Ming Zheng, and Wen Yu.
\newblock Subgroup {Analysis} with {Time}-to-{Event} {Data} {Under} a
  {Logistic}-{Cox} {Mixture} {Model}: {Subgroup} analysis with censored data.
\newblock {\em Scandinavian Journal of Statistics}, 43(3):863--878, 2016.

\bibitem{eng_discrete_2014}
Kevin~H. Eng and Bret~M. Hanlon.
\newblock Discrete mixture modeling to address genetic heterogeneity in
  time-to-event regression.
\newblock {\em Bioinformatics}, 30(12):1690--1697, February 2014.
\newblock \_eprint:
  https://academic.oup.com/bioinformatics/article-pdf/30/12/1690/48926918/bioinformatics\_30\_12\_1690.pdf.

\bibitem{dempster_maximum_1977}
A.~P. Dempster, N.~M. Laird, and D.~B. Rubin.
\newblock Maximum {Likelihood} from {Incomplete} {Data} {Via} the \textit{{EM}}
  {Algorithm}.
\newblock {\em Journal of the Royal Statistical Society: Series B
  (Methodological)}, 39(1):1--22, 1977.

\bibitem{you2018subtype}
Na~You, Shun He, Xueqin Wang, Junxian Zhu, and Heping Zhang.
\newblock Subtype classification and heterogeneous prognosis model construction
  in precision medicine.
\newblock {\em Biometrics}, 74(3):814--822, 2018.

\bibitem{andersen1982cox}
Per~Kragh Andersen and Richard~D Gill.
\newblock Cox's regression model for counting processes: a large sample study.
\newblock {\em The annals of statistics}, pages 1100--1120, 1982.

\bibitem{zhou_understanding_2001}
Mai Zhou.
\newblock Understanding the {Cox} {Regression} {Models} {With} {Time}-{Change}
  {Covariates}.
\newblock {\em The American Statistician}, 55(2):153--155, 2001.

\bibitem{therneau2000cox}
Terry~M Therneau, Patricia~M Grambsch, Terry~M Therneau, and Patricia~M
  Grambsch.
\newblock {\em The cox model}.
\newblock Springer, 2000.

\bibitem{therneau2015package}
Terry~M Therneau and Thomas Lumley.
\newblock Package ‘survival’.
\newblock {\em R Top Doc}, 128(10):28--33, 2015.

\bibitem{wong2017piecewise}
George~YC Wong, Michael~P Osborne, Qinggang Diao, and Qiqing Yu.
\newblock Piecewise cox models with right-censored data.
\newblock {\em Communications in Statistics-Simulation and Computation},
  46(10):7894--7908, 2017.

\bibitem{berkson1952survival}
Joseph Berkson and Robert~P Gage.
\newblock Survival curve for cancer patients following treatment.
\newblock {\em Journal of the American Statistical Association},
  47(259):501--515, 1952.

\bibitem{lambert2007modeling}
Paul~C Lambert.
\newblock Modeling of the cure fraction in survival studies.
\newblock {\em The Stata Journal}, 7(3):351--375, 2007.

\bibitem{felizzi_mixture_2021}
Federico Felizzi, Noman Paracha, Johannes P\"{o}hlmann, and Joshua Ray.
\newblock Mixture {Cure} {Models} in {Oncology}: {A} {Tutorial} and {Practical}
  {Guidance}.
\newblock {\em PharmacoEconomics - Open}, 5(2):143--155, 2021.

\bibitem{amico2018cure}
Mailis Amico and Ingrid Van~Keilegom.
\newblock Cure models in survival analysis.
\newblock {\em Annual Review of Statistics and Its Application}, 5:311--342,
  2018.

\bibitem{jia_cure_2013}
Xiaoyu Jia, Camelia~S. Sima, Murray~F. Brennan, and Katherine~S. Panageas.
\newblock Cure models for the analysis of time-to-event data in cancer studies:
  {Cure} {Models} in {Cancer} {Studies}.
\newblock {\em Journal of Surgical Oncology}, 108(6):342--347, 2013.

\bibitem{sy_estimation_2000}
Judy~P. Sy and Jeremy M.~G. Taylor.
\newblock Estimation in a {Cox} {Proportional} {Hazards} {Cure} {Model}.
\newblock {\em Biometrics}, 56(1):227--236, 2000.

\bibitem{wu_convergence_1983}
C.~F.~Jeff Wu.
\newblock On the {Convergence} {Properties} of the {EM} {Algorithm}.
\newblock {\em The Annals of Statistics}, 11(1), 1983.

\bibitem{owen2001empirical}
Art~B Owen.
\newblock {\em Empirical likelihood}.
\newblock CRC press, 2001.

\bibitem{shen_inference_2015}
Juan Shen and Xuming He.
\newblock Inference for {Subgroup} {Analysis} {With} a {Structured}
  {Logistic}-{Normal} {Mixture} {Model}.
\newblock {\em Journal of the American Statistical Association},
  110(509):303--312, 2015.

\bibitem{breslow1972contribution}
Norman~E Breslow.
\newblock Contribution to discussion of paper by dr cox.
\newblock {\em Journal of the Royal Statistical Society, Series B},
  34:216--217, 1972.

\bibitem{world1979handbook}
World~Health Organization et~al.
\newblock {\em WHO handbook for reporting results of cancer treatment}.
\newblock World Health Organization, 1979.

\bibitem{therasse_new_2000}
Patrick Therasse, Susan~G. Arbuck, Elizabeth~A. Eisenhauer, Jantien Wanders,
  Richard~S. Kaplan, Larry Rubinstein, Jaap Verweij, Martine Van~Glabbeke,
  Allan~T. van Oosterom, Michaele~C. Christian, and Steve~G. Gwyther.
\newblock New {Guidelines} to {Evaluate} the {Response} to {Treatment} in
  {Solid} {Tumors}.
\newblock {\em JNCI: Journal of the National Cancer Institute}, 92(3):205--216,
  2000.

\bibitem{eisenhauer2009new}
Elizabeth~A Eisenhauer, Patrick Therasse, Jan Bogaerts, Lawrence~H Schwartz,
  Danielle Sargent, Robert Ford, Janet Dancey, S~Arbuck, Steve Gwyther,
  Margaret Mooney, et~al.
\newblock New response evaluation criteria in solid tumours: revised recist
  guideline (version 1.1).
\newblock {\em European journal of cancer}, 45(2):228--247, 2009.

\bibitem{lencioni_modified_2010}
Riccardo Lencioni and Josep Llovet.
\newblock Modified {RECIST} ({mRECIST}) {Assessment} for {Hepatocellular}
  {Carcinoma}.
\newblock {\em Seminars in Liver Disease}, 30(01):052--060, 2010.

\bibitem{choi2007correlation}
Haesun Choi, Chuslip Charnsangavej, Silvana~C Faria, Homer~A Macapinlac,
  Michael~A Burgess, Shreyaskumar~R Patel, Lei~L Chen, Donald~A Podoloff, and
  Robert~S Benjamin.
\newblock Correlation of computed tomography and positron emission tomography
  in patients with metastatic gastrointestinal stromal tumor treated at a
  single institution with imatinib mesylate: proposal of new computed
  tomography response criteria.
\newblock {\em Journal of clinical Oncology}, 25(13):1753--1759, 2007.

\bibitem{smith2010assessing}
Andrew~Dennis Smith, Michael~L Lieber, and Shetal~N Shah.
\newblock Assessing tumor response and detecting recurrence in metastatic renal
  cell carcinoma on targeted therapy: importance of size and attenuation on
  contrast-enhanced ct.
\newblock {\em American Journal of Roentgenology}, 194(1):157--165, 2010.

\bibitem{wolchok2009guidelines}
Jedd~D Wolchok, Axel Hoos, Steven O'Day, Jeffrey~S Weber, Omid Hamid, Celeste
  Lebb{\'e}, Michele Maio, Michael Binder, Oliver Bohnsack, Geoffrey Nichol,
  et~al.
\newblock Guidelines for the evaluation of immune therapy activity in solid
  tumors: immune-related response criteria.
\newblock {\em Clinical cancer research}, 15(23):7412--7420, 2009.

\bibitem{tanner1993tools}
Martin~A Tanner.
\newblock {\em Tools for statistical inference}, volume~3.
\newblock Springer, 1993.

\bibitem{dillon2010asymptotic}
Joshua~V Dillon, Krishnakumar Balasubramanian, and Guy Lebanon.
\newblock Asymptotic analysis of generative semi-supervised learning.
\newblock {\em arXiv preprint arXiv:1003.0024}, 2010.

\bibitem{cozman2003semi}
Fabio~Gagliardi Cozman, Ira Cohen, Marcelo~Cesar Cirelo, et~al.
\newblock Semi-supervised learning of mixture models.
\newblock In {\em ICML}, volume~4, page~24, 2003.

\bibitem{douillard_randomized_2010}
Jean-Yves Douillard, Salvatore Siena, James Cassidy, Josep Tabernero, Ronald
  Burkes, Mario Barugel, Yves Humblet, Gy{\"o}rgy Bodoky, David Cunningham,
  Jacek Jassem, et~al.
\newblock Randomized, phase iii trial of panitumumab with infusional
  fluorouracil, leucovorin, and oxaliplatin (folfox4) versus folfox4 alone as
  first-line treatment in patients with previously untreated metastatic
  colorectal cancer: the prime study.
\newblock {\em Journal of clinical oncology}, 28(31):4697--4705, 2010.

\bibitem{douillard_panitumumabfolfox4_2013}
Jean-Yves Douillard, Kelly~S Oliner, Salvatore Siena, Josep Tabernero, Ronald
  Burkes, Mario Barugel, Yves Humblet, Gyorgy Bodoky, David Cunningham, Jacek
  Jassem, et~al.
\newblock Panitumumab--folfox4 treatment and ras mutations in colorectal
  cancer.
\newblock {\em New England Journal of Medicine}, 369(11):1023--1034, 2013.

\bibitem{douillard2014final}
Jean-Yves Douillard, S~Siena, J~Cassidy, J~Tabernero, R~Burkes, M~Barugel,
  Y~Humblet, G~Bodoky, D~Cunningham, J~Jassem, et~al.
\newblock Final results from prime: randomized phase iii study of panitumumab
  with folfox4 for first-line treatment of metastatic colorectal cancer.
\newblock {\em Annals of Oncology}, 25(7):1346--1355, 2014.

\bibitem{kamarudin_time-dependent_2017}
Adina~Najwa Kamarudin, Trevor Cox, and Ruwanthi Kolamunnage-Dona.
\newblock Time-dependent {ROC} curve analysis in medical research: current
  methods and applications.
\newblock {\em BMC Medical Research Methodology}, 17(1):53, 2017.

\bibitem{Blanche2012TimedependentAW}
Paul Blanche, Aur{\'e}lien Latouche, and Vivian Viallon.
\newblock Time-dependent auc with right-censored data: a survey study.
\newblock {\em arXiv: Methodology}, 2012.

\bibitem{heagerty_survival_2005}
Patrick~J. Heagerty and Yingye Zheng.
\newblock Survival {Model} {Predictive} {Accuracy} and {ROC} {Curves}.
\newblock {\em Biometrics}, 61(1):92--105, 2005.

\bibitem{xu_proportional_2000}
Ronghui Xu and John O'Quigley.
\newblock Proportional {Hazards} {Estimate} of the {Conditional} {Survival}
  {Function}.
\newblock {\em Journal of the Royal Statistical Society Series B: Statistical
  Methodology}, 62(4):667--680, 2000.

\bibitem{hecht2009randomized}
J~Randolph Hecht, Edith Mitchell, Tarek Chidiac, Carroll Scroggin, Christopher
  Hagenstad, David Spigel, John Marshall, Allen Cohn, David McCollum, Philip
  Stella, et~al.
\newblock A randomized phase iiib trial of chemotherapy, bevacizumab, and
  panitumumab compared with chemotherapy and bevacizumab alone for metastatic
  colorectal cancer.
\newblock {\em Journal of Clinical Oncology}, 27(5):672--680, 2009.

\bibitem{ueda1998deterministic}
Naonori Ueda and Ryohei Nakano.
\newblock Deterministic annealing em algorithm.
\newblock {\em Neural networks}, 11(2):271--282, 1998.

\bibitem{Celeux1996StochasticVO}
Gilles Celeux, Didier Chauveau, and Jean Diebolt.
\newblock Stochastic versions of the em algorithm: an experimental study in the
  mixture case.
\newblock {\em Journal of Statistical Computation and Simulation}, 55:287--314,
  1996.

\bibitem{tibshirani1997lasso}
Robert Tibshirani.
\newblock The lasso method for variable selection in the cox model.
\newblock {\em Statistics in medicine}, 16(4):385--395, 1997.

\bibitem{fan2002variable}
Jianqing Fan and Runze Li.
\newblock Variable selection for cox's proportional hazards model and frailty
  model.
\newblock {\em The Annals of Statistics}, 30(1):74--99, 2002.

\bibitem{guo_inference_2021}
Xinzhou Guo and Xuming He.
\newblock Inference on {Selected} {Subgroups} in {Clinical} {Trials}.
\newblock {\em Journal of the American Statistical Association},
  116(535):1498--1506, 2021.

\bibitem{bonetti_patterns_2004}
M.~Bonetti.
\newblock Patterns of treatment effects in subsets of patients in clinical
  trials.
\newblock {\em Biostatistics}, 5(3):465--481, 2004.

\bibitem{zhang_robust_2012}
Baqun Zhang, Anastasios~A. Tsiatis, Eric~B. Laber, and Marie Davidian.
\newblock A {Robust} {Method} for {Estimating} {Optimal} {Treatment} {Regimes}.
\newblock {\em Biometrics}, 68(4):1010--1018, 2012.

\end{thebibliography}

\newpage

\section*{Appendix A: Selected R Code}

\subsection*{SPIRLS-EM algorithm for the Dual Cox model}
\begin{lstlisting}[language=R]
dualcox <-
  function(Time,
           #observed time
           Delta,
           #survival status
           X,
           #a data matrix of explanatory variables, where each column
           #corresponds  to one predictor and each row indicates one sample.
           K = 2,
           #Dual Cox model has two components of  the finite-mixture Cox model
           iter.max = 1000,
           #maximum number of EM iterations
           u.init = NULL,
           #initial value of U
           #a data matrix that gives the probability that each unlabeled
           #sample belongs to each component. And the probability
           #that each labeled sample belongs to each component is 1.
           abstol = 1e-5,
           #	absolute tolerance of EM algorithm
           reltol = 1e-7,
           #relative tolerance of EM algorithm
           seed = 1,
           #random seed for initialing U if it is not given
           pi.init=NULL,
           #The initial value of pi
           unfixed=NULL
           #row index of the unlabeled data
  ){
    set.seed(seed)
    n = length(Time)
    X = as.data.frame(X)
    Xm = as.matrix(X)
    Time = matrix(Time, ncol = 1)
    Delta = matrix(Delta, ncol = 1)
    p = dim(Xm)[2]
    surv = cbind(time = Time, status = Delta)
    colnames(surv) = c("time", "status")
    ####check u.init and initial U matrix
    if(is.matrix(u.init)){
      if(all(dim(u.init) == c(n,K)))
        U = u.init
      else
        stop("parameter u.init wrong!")
    }else if(is.vector(u.init)){
      if(length(u.init) != n | any(!(u.init %in% (1:K))))
        stop("parameter u.init wrong!")
      else{
        U = matrix(0.001, K, K)
        diag(U) = 1 - (K-1)*0.001
        U = U[u.init, ]
        if(K==1)
          U = matrix(U, ncol=1)
      }
    }else if(is.null(u.init)){
      u.init = sample(1:K, size = n, replace = TRUE)
      U = matrix(0.001, K, K)
      diag(U) = 1 - (K-1)*0.001
      U = U[u.init, ]
      if(K==1)
        U = matrix(U,ncol=1)
    }else{
      stop("parameter u.init wrong!")
    }
    ####initial pi
    pi =pi.init
    ####EM iteration
    ploglik = NULL
    mixloglik = NULL
    lastploglik = -Inf
    lastmixloglik = -Inf
    convergence = FALSE
    iter = 0
    while(!convergence & iter<iter.max){

      ##Mstep
fit = lapply(1:K, function(k){
  idx = U[,k] >= 1e-6
  weights = U[idx, k]
  result = survival::coxph(survival::Surv(surv[idx,1],surv[idx,2])~.,
  data =data.frame(Xm[idx,]),weights = weights)
          return(as.numeric(result$coefficients))
        })

h0 = NULL
H0 = NULL
EE = NULL
PLL = NULL
allUEXB = NULL

for (k in 1:K) {
  Uk = U[,k]
  XB = Xm%*%fit[[k]]
  EXB = exp(XB)
  EE = cbind(EE, EXB)

  UEXB = Uk*EXB
  allUEXB = cbind(allUEXB, UEXB)
  aa = sapply(Time, function(x)sum(UEXB[Time>=x]))
  bb = 1/aa
  bb[bb == Inf] = 0
  h00 = Uk*bb*Delta
  H00 = sapply(Time, function(x)sum(h00[Time<=x]))
  h0 = cbind(h0, h00)
  H0 = cbind(H0, H00)
  PLL = c(PLL, sum((Uk*(XB-log(aa)))[Delta==1&aa>0]))

}

CC = (h0 * EE)^(Delta %*% t(rep(1, K)))


fk=CC * exp(-H0 * EE)
CC = (CC * exp(-H0 * EE)) %*% diag(as.vector(pi), K)
#The complete log-likelihood
nowploglik = sum(PLL) + sum((U%*%diag(log(pi), K))[U!= 0])
#log-likehood based on observed data
nowmixloglik = sum(log(apply(CC, 1, sum)))

      ploglik = c(ploglik, nowploglik)
      mixloglik = c(mixloglik, nowmixloglik)

      if(!convergence){
        U[unfixed,] = do.call(rbind, lapply(unfixed,function(i){
          x = CC[i,]
          return(x / sum(x))
        }))
            pi =apply(U, 2, mean)
      }
       class = apply(U, 1, function(x){ which.max(x)[1]})



      if(is.finite(nowploglik) & is.finite(lastploglik) &
         (abs(nowploglik-lastploglik) < abstol |
          abs(nowploglik-lastploglik) <
          abs(reltol*(nowploglik + reltol))))
        convergence = TRUE
      lastploglik = nowploglik
      lastmixloglik = nowmixloglik
      iter = iter + 1
    }
    if(iter==iter.max)
      warning("EM iterations reach iter.max!")
    class = apply(U, 1, function(x){ which.max(x)[1]})

    finalmixloglik = mixloglik[length(mixloglik)]

    ret = list(U = U,
               fit = fit,
               pi = pi,
               class = class,
               ploglik = ploglik,
               mixloglik = mixloglik,
               iter = iter,
               convergence = iter < iter.max)
    return(ret)
  }
\end{lstlisting}

\subsection*{1000 sample size with 20\texorpdfstring{$\%$}{Lg} censoring rate in the simulation }

\begin{lstlisting}[language=R]

mixing_proportions <- c(0.3, 0.7)
iter1=NULL
beta1=c(-1,0.5,3,0.8)
beta2=c(2 ,-0.1,-3,0.2)
accuracy_value=NULL
censoring_value=NULL
pi.est=NULL
beta1.1=NULL
beta2.1=NULL
beta3.1=NULL
beta4.1=NULL
beta1.2=NULL
beta2.2=NULL
beta3.2=NULL
beta4.2=NULL

for (j in 1:1000) {
  set.seed(j)
  covariates <- matrix(c(rbinom(2000,1,0.5),rnorm(2000)), ncol = 4)
  fixed=which(covariates[,1]==1)
  unfixed=which(covariates[,1]==0)
  U=runif(1000,min = 0,max = 1)
  ##  the survival time of sample i in the k-th sub-population
  vk=1
  lambda1=1/35
  #The first 300 samples are from the response group

  #and the last 700 samples are from non-response
  U=runif(1000,min = 0,max = 1)
  survival1=(-log(U[1:300])/(lambda1)*exp(-covariates[1:300,]%*%beta1))^(1/vk)
  survival2=(
  -log(U[301:1000])/(lambda1)*exp(-covariates[301:1000,]%*%beta2))^(1/vk)
  survival_times=c(survival1,survival2)
  #The censoring time
  c=runif(1000,min = 0,max = exp(6.5))
  status=ifelse(survival_times<c,1,0)
  censor=1-mean(status)
  censoring_value=c(censoring_value,censor)
  u.init   <- cbind(c(rep(1,300),rep(0,700)),c(rep(0,300),rep(1,700)))

  group=c(rep(1,300),rep(2,700))

  group2=c(rep(1,300),rep(0,700))
  pi.init=apply(u.init[fixed,],2,mean)

  for (i in unfixed) {
    u.init[i,1]=pi.init[1]
    u.init[i,2]=1-pi.init[1]

  }
  semi=dualcox(survival_times,status,X=covariates,abstol = 0.00001,
  reltol = 1e-05,K = 2,iter.max = 1000,seed = 123,u.init = u.init,
  unfixed = unfixed,pi.init=pi.init)

  beta1.1=c(beta1.1,semi$fit[[1]][1])
  beta2.1=c(beta2.1,semi$fit[[1]][2])
  beta3.1=c(beta3.1,semi$fit[[1]][3])
  beta4.1=c(beta4.1,semi$fit[[1]][4])
  iter1=c(iter1,semi$iter)
  beta1.2=c(beta1.2,semi$fit[[2]][1])
  beta2.2=c(beta2.2,semi$fit[[2]][2])
  beta3.2=c(beta3.2,semi$fit[[2]][3])
  beta4.2=c(beta4.2,semi$fit[[2]][4])
  pi.est=c(pi.est,semi$pi[1])
  accuracy=mean(ifelse((semi$class[unfixed]-group[unfixed])==0,1,0))
  accuracy_value=c(accuracy_value,accuracy)
}

mean=c(mean(beta1.1),mean(beta2.1),mean(beta3.1),
mean(beta4.1),mean(beta1.2),mean(beta2.2),
mean(beta3.2),mean(beta4.2),mean(pi.est),mean(accuracy_value))
sd=c(sd(beta1.1),sd(beta2.1),sd(beta3.1),sd(beta4.1),sd(beta1.2),sd(beta2.2),
sd(beta3.2),sd(beta4.2),sd(pi.est),sd(accuracy_value)
)
bias=c((mean(beta1.1)-beta1[1]),(mean(beta2.1)-beta1[2]),
(mean(beta3.1)-beta1[3]),(mean(beta4.1)-beta1[4]),
(mean(beta1.2)-beta2[1]),(mean(beta2.2)-beta2[2]),
(mean(beta3.2)-beta2[3]),(mean(beta4.2)-beta2[4]),
(mean(pi.est)-mixing_proportions[1]),0)
relative.bias=c(  (mean(beta1.1)-beta1[1])/beta1[1],
                  (mean(beta2.1)-beta1[2])/beta1[2],
                  (mean(beta3.1)-beta1[3])/beta1[3],
                  (mean(beta4.1)-beta1[4])/beta1[4],
                  (mean(beta1.2)-beta2[1])/beta2[1],
                  (mean(beta2.2)-beta2[2])/beta2[2],
                  (mean(beta3.2)-beta2[3])/beta2[3],
                  (mean(beta4.2)-beta2[4])/beta2[4],0,0)
sample.1000=round(data.frame(MEAN=mean,SD=sd,BIAS=bias,
RELATIVE.BIAS=relative.bias),digits = 3)
sample.1000
mean(censoring_value)
range(censoring_value)
mean(iter1)
range(iter1)
\end{lstlisting}

\end{document}